\documentclass[reqno,12pt,letterpaper]{amsart}
\usepackage{amsmath,amssymb,amsthm,graphicx,mathrsfs,url}
\usepackage[usenames,dvipsnames]{color}
\usepackage[colorlinks=true,linkcolor=Red,citecolor=Green]{hyperref}
\usepackage{amsxtra}
\usepackage{epstopdf}
\def\arXiv#1{\href{http://arxiv.org/abs/#1}{arXiv:#1}}
\usepackage{braket}

\usepackage{hyperref}

\def\?[#1]{\textbf{[#1]}\marginpar{\Large{\textbf{??}}}}

\def\smallsection#1{\smallskip\noindent\textbf{#1}.}
\let\epsilon=\varepsilon 

\usepackage{caption}
\usepackage{subcaption}
\newtheorem{theo}{Theorem}
 \newtheorem{innercustomprop}{ }
\newenvironment{ent}[1]
  {\renewcommand\theinnercustomprop{#1}\innercustomprop}
  {\endinnercustomprop}
  \newenvironment{capa}[1]
  {\renewcommand\theinnercustomprop{#1}\innercustomprop}
  {\endinnercustomprop}
  \newtheorem{innercustomthm}{Theorem}
\newenvironment{customthm}[1]
  {\renewcommand\theinnercustomthm{#1}\innercustomthm}
  {\endinnercustomthm}
\newtheorem{prop}{Proposition}[section]	
\newtheorem{defi}[prop]{Definition}
\newtheorem{Assumption}{Assumption}

\newtheorem{lemm}[prop]{Lemma}
\newtheorem{corr}[prop]{Corollary}
\newtheorem{rem}{Remark}
\newtheorem{ex}{Example}
\numberwithin{equation}{section}

\let\Re=\Real

\DeclareMathOperator{\supp}{supp}

\DeclareMathOperator{\tr}{tr}

\def\indic{\operatorname{1\hskip-2.75pt\relax l}}

\title[Rates for quantum evolution \& entropic continuity bounds]{Convergence rates for quantum evolution \& entropic continuity bounds in infinite dimensions}

\author{Simon Becker}
\email{simon.becker@damtp.cam.ac.uk}
\address{University of Cambridge,
DAMTP, Wilberforce Rd, Cambridge CB3 0WA, UK}
\author{Nilanjana Datta}
\email{n.datta@damtp.cam.ac.uk}
\address{University of Cambridge,
DAMTP, Wilberforce Rd, Cambridge CB3 0WA, UK}
\begin{document}

\begin{abstract}
By extending the concept of energy-constrained diamond norms, we obtain continuity bounds on the dynamics of both closed and open quantum systems in infinite-dimensions, which are stronger than previously known bounds. We extensively discuss applications of our theory to quantum speed limits, attenuator and amplifier channels, the quantum Boltzmann equation, and quantum Brownian motion. Next, we obtain explicit log-Lipschitz continuity bounds for entropies of infinite-dimensional quantum
systems, and classical capacities of infinite-dimensional
quantum channels under energy-constraints. These bounds are determined by the high energy
spectrum of the underlying Hamiltonian and can be evaluated using Weyl's law.
\end{abstract}

\maketitle

\section{Introduction}
\label{sec_intro}
 
Infinite-dimensional quantum systems play an important role in quantum theory. The quantum harmonic oscillator, which is the simplest example of such a system, has various physical realizations, \emph{e.g.}~in vibrational modes of molecules, lattice vibrations of crystals, electric and magnetic fields of electromagnetic waves, \emph{etc.}. Even though much of quantum information science focusses on finite-dimensional quantum systems, the relevance of infinite-dimensional (or continuous variable) quantum systems in quantum computing, and various other quantum technologies, has become increasingly apparent~(see \emph{e.g.}~\cite{E06} and references therein).

In this paper we make a detailed analysis of the time evolution of autonomous, infinite-dimensional quantum systems. The dynamics of such a system is described by a quantum dynamical semigroup (QDS) $(T_t)_{t \geq 0}$. In the Schr\"odinger picture, this is a one-parameter family of linear, completely positive, trace-preserving maps (\emph{i.e.}~quantum channels) acting on states of the quantum system. In the Heisenberg picture, the dynamics of observables is given by the adjoint semigroup $(T_t^*)_{t \geq 0}$ where $\forall \, t \geq 0$, $T_t^*$ is a linear, completely positive, unital map on the space of bounded operators acting on the system\footnote{$T_t^*$ is the adjoint of $T_t$ with respect to the Hilbert-Schmidt inner product.}. 

There are different notions of continuity of QDSs. The case of {\em{uniformly continuous}} QDSs is the simplest, and is easy to characterize (see Section~\ref{semigroups} for a compendium on semigroup theory). A semigroup is uniformly continuous if and only if the generator is bounded. In this paper, we focus on the analytically richer case of {\em{strongly continuous}} semigroups, which appear naturally when the generator is unbounded. 

QDSs are used to describe the dynamics of both closed and open quantum systems\footnote{{{For closed quantum system, the QDS consists of unitary operators $T_t$. Since $T_{-t}=T_t^*$ this semigroup extends to a group with $t \in \mathbb R$.}}}. Open quantum systems are of particular importance in quantum information theory since systems which are of relevance in quantum information-processing tasks undergo unavoidable interactions with their environments, and hence are inherently open. In fact, any realistic quantum-mechanical system is influenced by its interactions with its environment, which typically has a large number of degrees of freedom. A prototypical example of such a system is an atom interacting with its surrounding radiation field. In quantum information-processing tasks, interactions between a system and its environment leads to loss of information (encoded in the system) due to processes such as decoherence and dissipation. QDSs are useful in describing these processes. The theory of open quantum systems has also found applications in various other fields including condensed matter theory and quantum optics.

Infinite-dimensional {\em{closed}} quantum systems to which our results apply are \emph{e.g.} described by time-independent Schr\"odinger operators $H = - \Delta + V$, which are ubiquitous in the literature. Examples of infinite-dimensional {\em{open}} quantum systems, to which our results apply, include, among others, amplifier and attenuator channels, the Jaynes-Cummings model of quantum optics, quantum Brownian motion, and the quantum Boltzmann equation (which describes how the motion of a single test particle is affected by collisions with an ideal background gas). These will be discussed in detail in Section \ref{OQS}.

\subsection{Rates of convergence for quantum evolution}
 Let us focus on the defining property of a strongly continuous semigroup $(T_t)_{t \ge 0}$ on a Banach space $X$, namely, the convergence property for all $x \in X$
\[\lim_{t \to 0^+} T_t x = x.\]
In this paper, we are interested in a refined analysis of this convergence, \emph{i.e.}, our aim is to determine the {\em{rate}} at which $ T_t$ converges to the identity map $I$ as $t$ goes to zero, and study applications of it. 

This rate of convergence of a semigroup on a Banach space $X$ is linear in time uniformly for all normalized $x \in X$, if and only if the generator of the semigroup is a bounded operator, since then, denoting the generator as $A$, we have 
\begin{equation}
\label{eq:ftc}
 \left\lVert T_tx - x \right\rVert   \le \int_0^t \left\lVert T_s Ax  \right\rVert \ ds \le \left\lVert Ax \right\rVert  \sup_{s \in [0,t]} \left\lVert T_s \right\rVert t. 
 \end{equation}
For general strongly continuous semigroups with unbounded generators, however, one merely knows that $\lim_{t \rightarrow 0} \left\lVert T_tx -x \right\rVert=0$ by strong continuity, and there is no information on the rate of convergence. If the generator, $A$, of the semigroup is unbounded, all elements $x \in X$ that are also in the domain, $D(A)$, of the generator still satisfy a linear time asymptotics by \eqref{eq:ftc}. This is because $\left\lVert Ax \right\rVert$ is well-defined for $x \in D(A)$, and thus (\ref{eq:ftc}) holds. However, if the generator $A$ is unbounded, then the bound (\ref{eq:ftc}) is not uniform for normalized $x \in D(A)$, since $||Ax||$ can become arbitrarily large.
\medskip

To obtain more refined information on the rate of convergence, we study spaces that interpolate between the convergence with linear rate $t^1$ (that holds for elements in the domain $D(A) \subseteq X$ of the generator, by \eqref{eq:ftc}) and the convergence without an a priori rate, which we might formally interpret as $t^0$, for general elements of the space $X$. More precisely, we consider interpolation spaces, known as \emph{Favard spaces} $F_{\alpha}=F_{\alpha}((T_t)_t)$ in semigroup theory \cite{T78}, of elements $x \in X$ such that for some $C_x>0$
\[\left\lVert T_tx - x \right\rVert \le  C_x t^{\alpha}\quad {\hbox{with}}\quad \alpha \in (0,1], \text{ for all } t>0.\] 

In order to study convergence rates and analyze continuity properties of QDSs we need to choose a suitable metric on the set of quantum channels\footnote{This is because if $(T_t)_{t \ge 0}$ is a QDS, then for any $t$, $T_t$ is a quantum channel.}. A natural metric which is frequently used is the one induced by the so-called completely bounded trace norm or \emph{diamond norm}, denoted as $\left\lVert \bullet \right\rVert_{\diamond}$. However, it has been observed in \cite{W17} that if the underlying Hilbert space $\mathcal{H}$ is infinite-dimensional, then the convergence generated by the diamond norm is, in general, too strong to capture the empirical observation that channels whose parameters differ only by a small amount, should be close to each other. Examples of Gaussian channels for which convergence in the
diamond norm does hold are, for example, studied in \cite{Wi18}.

In this case, a weaker norm, namely the {\em{energy-constrained diamond norm}}, (or {\em{ECD norm}}, in short), introduced independently by Shirokov~\cite[(2)]{S18} and Winter~\cite[Def. $3$]{W17}, proves more useful for studying convergence properties of QDSs in the Schr\"odinger picture (see Example \ref{ex:Attenchannel}). It is denoted as $\left\lVert \bullet \right\rVert_{\diamond}^{E}$, where $E$ characterizes the energy constraint. 

In this paper, we introduce a one-parameter family of ECD norms, $\left\lVert \bullet \right\rVert_{\diamond^{2\alpha}}^{S,E}$; here $S$ denotes a positive semi-definite operator, $E$ is a scalar taking values above the bottom of the spectrum of $S$, and $\alpha \in (0,1]$ is a parameter (see Definition \ref{defi:ecdn}). We refer to these norms as $\alpha$-ECD norms. They reduce to the usual ECD norm for the choice $\alpha=1/2$, when $S$ is chosen to be the Hamiltonian of the system. A version of the $\alpha=1/2$-ECD norm, for $S$ being the number operator, was first introduced in the context of bosonic channels by Pirandola \emph{et al.} in~\cite[(98)]{PLOB}.

To illustrate the power of the $\alpha$-ECD norms over the standard diamond norm, and even over the usual ECD norm, we discuss the example of the (single mode bosonic quantum-limited) attenuator channel with time-dependent attenuation parameter $\eta(t):=e^{-t}$ (see Example \ref{attchannel} for details): 
\begin{ex}[Attenuator channel]
\label{ex:Attenchannel}
Let $N:=a^*a$ be the number operator, with $a^*$, $a$ being the standard creation and annihilation operators.
Consider the \emph{attenuator channel} $\Lambda^{\text{att}}_t$, with time-dependent attenuation parameter $\eta(t):=e^{-t}$. This one is uniquely defined by its action on coherent states $\vert \alpha \rangle = e^{-\left\lvert \alpha \right\rvert^2/2} \sum_{n=0}^{\infty} \frac{\alpha^n}{\sqrt{n!}} \vert n \rangle$, where $\left\{\ket{n}\right\}_n$ is the standard eigenbasis of the number operator, as follows:
\begin{equation}\label{coh}
 \Lambda^{\text{att}}_t(\vert \alpha \rangle\langle \alpha \vert):= \vert e^{-t}\alpha \rangle\langle e^{-t} \alpha \vert.
\end{equation} 
The family $(\Lambda^{\text{att}}_t)$ is then a QDS.
 
As pointed out in \cite{W17}, the diamond norm is too strong in many situations. In fact, for any times $t \neq s$ it is shown in \cite[Proposition $1$]{W17} that 
\[ \left\lVert \Lambda^{\text{att}}_t - \Lambda^{\text{att}}_{s}  \right\rVert_{\diamond}=2.\]
Thus, no matter whether $t$ and $s$ are close or far apart, the diamond norm is always equal to $2.$  The ECD norm serves to overcome this problem, since it follows from 
\cite[Section IV B]{W17} that
\[ \lim_{t \to s}\left\lVert \Lambda^{\text{att}}_t - \Lambda^{\text{att}}_{s}  \right\rVert_{\diamond}^{E}= 0.\] 
However, as we will show in Example \ref{attchannel}, considering the entire family of $\alpha$-ECD norms provides further improvement, since it allows us to capture the {\em{rate of convergence}} of the channels as $t$ converges to $s$:
\begin{equation*}
\begin{split}
\left\lVert \Lambda^{\text{att}}_t-\Lambda^{\text{att}}_{s} \right\rVert^{N,E}_{\diamond^{2\alpha}} \le C_{\alpha} E^{\alpha} \ \vert t- s \vert^{\alpha}
\end{split}
\end{equation*}
for some constant $C_{\alpha}>0$ that is explicitly given in Example \ref{attchannel}.
\end{ex} 
\smallskip
\subsubsection{Quantum Speed Limits:} The bounds which we obtain on the dynamics of closed and open quantum systems, immediately lead to lower bounds on the minimal time needed for a quantum system to evolve from one quantum state to another. Such bounds are known as \emph{quantum speed limits}.
Mandelstam and Tamm \cite{MT} were the first to derive a bound on the minimal time, $t_{\text{min}}$, needed for a given pure state to evolve to a pure state orthogonal to it. It is given by \footnote{We use dimensionless notation in this paper.}
\[ t_{\text{min}} \ge \frac{\pi}{2 \Delta E},\]
where $\Delta E$ is the variance of the energy of the initial state.
From the work of \cite{ML,LT} it followed then that the minimal time needed to reach any state of expected energy $E$, which is orthogonal to the initial state, satisfies

\begin{equation}
\label{eq:tmin}
t_{\text{min}} \ge \operatorname{max} \left\{\frac{\pi}{2 \Delta E},\frac{\pi}{2} \frac{1}{E} \right\}.
\end{equation} 
Moreover, this bound was shown to be tight.  
If one includes physical constants and formally studies the semiclassical limit $\hbar\rightarrow 0,$ one discovers that the lower bound in $t_{\text{min}}$ vanishes. However, it was shown in \cite{SCMC} that speed limits also exist in the classical regime.
The study of speed limits was generalized in \cite{P} to the case of initial and target pure states which are not necessarily orthogonal, but are instead separated by arbitrary {\em{angles}}. It has also been generalized to mixed states and open systems with bounded generators.  Although the 
quantum speed limit for closed quantum systems that we obtain from the ECD norm (\emph{i.e.}~for $\alpha=1/2$), stated in \eqref{eq:alpha=1/2}, is smaller than \eqref{eq:tmin}, we obtain better estimates on the 
quantum speed limit for many states using different $\alpha$-ECD norm. In particular, the approach pursued in this article allows us to deal with:
\begin{itemize}
\item open quantum systems with unbounded generators, 
\item states with infinite expected energy, and
\item systems whose dynamics is generated by an operator which is different from that which penalizes the energy. 
\end{itemize}

\noindent
\subsection{Explicit convergence rates for entropies and capacities}

It is well-known that on infinite-dimensional spaces, the von Neumann entropy is discontinuous \cite{We78}. Hence, in order to obtain explicit bounds on the difference of the von Neumann entropies of two states, it is necessary to impose further restrictions on the set of admissible states. In \cite{W15}, continuity bounds for the von Neumann entropy of states of infinite-dimensional quantum systems were obtained by imposing an additional energy-constraint condition on the states, and imposing further assumptions on the class of admissible Hamiltonians. The latter are assumed to satisfy the so-called (\hyperref[ass:Gibbs]{Gibbs hypothesis}). 
Under the energy-constraint condition and the Gibbs hypothesis it is true that for any energy $E$ above the bottom of the spectrum of the Hamiltonian $H$, the Gibbs state $\gamma(E)=e^{-\beta(E) H}/Z_H(\beta(E))$\footnote{Here $Z_H$ denotes the partition function and $\beta$ denotes the inverse temperature.} is the maximum entropy state of expected energy $E$ \cite[p.196]{GS}.
Bounds on the difference of von Neumann entropies stated in \cite{W15} are fully explicit up to the occurrence of the entropy of a Gibbs state 
of the form $\gamma(E/\varepsilon),$ where $\varepsilon$ is an upper bound on the trace distance of the two states.

Since entropic continuity bounds are tight in the limit $\varepsilon \downarrow 0,$ we study (in Section \ref{ECB}) the entropy of such a Gibbs state in this limit. Note that for the Gibbs state $\gamma(E/{\varepsilon})$, the limit $\varepsilon \downarrow 0$ translates into a high energy limit.
By employing the so-called Weyl law \cite{I16}, which states that certain classes of Schr\"odinger operators have asymptotically the same high energy spectrum, we show that the asymptotic behaviour of the entropy of the Gibbs state is \emph{universal} for such classes of operators. This in turn yields fully explicit convergence rates both for the von Neumann entropy and for the conditional entropy (see \hyperref[corr:entropyesm]{Proposition [Entropy convergence]}).

In finite dimensions, continuity bounds on conditional entropies have found various applications, \emph{e.g.}~in establishing continuity properties of capacities of quantum channels \cite{LS09} and entanglement measures  \cite{CW03, YHW08}, and in the study of so-called approximately degradable quantum channels \cite{SSRW15}.
Analogously, in infinite dimensions, continuity bounds on the conditional entropy for states satisfying an energy constraint  \cite{W15}, were used by Shirokov \cite{S18} to derive continuity bounds for various constrained classical capacities of quantum channels\footnote{For a discussion of these capacities see Section~\ref{ECB}}. These bounds were once again given in terms of the entropy of a Gibbs state of the form $\gamma(E/\varepsilon)$. Here $\varepsilon$ denotes the upper bound on the ECD norm distance between the pair of channels considered, and $E$ denotes the energy threshold appearing in the energy constraint. Our result on the high energy asymptotics of Gibbs states yields a refinement of Shirokov's results, by providing the explicit behaviour of these bounds for small $\varepsilon$.

The bounds that we obtain on the dynamics of closed and open quantum systems (see Proposition \ref{prop:tarate} and Theorem \ref{theo:theo3}) also allow us to identify explicit time intervals 
over which the evolved state is close to the initial state.
Since entropic continuity bounds require such a smallness condition for the trace distance between pairs of states, 
we can then bound the entropy difference between the initial state and the time-evolved state (see Example~\ref{ex:achannel}).

We start the rest of the paper with some mathematical preliminaries in Section~\ref{prelim}. These include a discussion of QDSs, definition and properties of the $\alpha$-ECD norms, and some basic results from functional analysis that we use. In Section~\ref{sec_main} we state our main results. These consist of $(i)$ rates of convergence for quantum evolution in both closed and open quantum systems, and $(ii)$ explicit convergence rates for entropies and certain constrained classical capacities of quantum channels. The results concerning $(i)$ are proved in Sections~\ref{CQS} and \ref{OQS}, while those on $(ii)$ are proved in Section~\ref{ECB}. In Section~\ref{QSL} we discuss some interesting applications of our results, in particular to generalized relative entropies and quantum speed limits. We end the paper with some open problems in Section~\ref{probs}. Certain auxiliary results and technical proofs are relegated to the appendices.

\section{Mathematical Preliminaries}
\label{prelim}
 
In the sequel, all Hilbert spaces $\mathcal H$ are infinite-dimensional, separable and complex.
We denote the space of trace class operators on a Hilbert space $\mathcal H$ by $\mathcal T_1(\mathcal H)$, that of Hilbert-Schmidt operators by $\mathcal{T}_2(\mathcal H)$, and the $q$-th Schatten norm by $\left\lVert \bullet \right\rVert_q.$ The set of all {\em{quantum states}} (\emph{i.e.}~positive semidefinite operators of unit trace) on a Hilbert space $\mathcal{H}$ is denoted as $\mathscr{D}(\mathcal{H}).$ We denote the spectrum of a self-adjoint operator $H$ by $\sigma (H)$, and its spectral measure by $\mathcal E^H$ \cite[p.224]{RS1}. For the state $\rho_{AB}$ of a bipartite system $AB$ with Hilbert space ${\mathcal{H}}_A \otimes {\mathcal H}_B$, the reduced state of $A$ is given by $\rho_A = \tr_B \rho_{AB}$, where $\tr_B$ denotes the partial trace over ${\mathcal H}_B$. Occasionally, we also write $\rho_{\mathcal H_A}$ instead of $\rho_A.$
The form domain of a positive semi-definite operator $S$, \emph{i.e.}\ $\left\langle Sx,x \right\rangle \ge 0$ for all $x \in D(S),$ is denoted by $\mathfrak D(S):=D(\sqrt{S}).$ We denote the space of bounded linear operators between normed spaces $X,Y$ as $\mathcal B(X,Y)$, and as $\mathcal B(X)$ if $X=Y.$

If there is a constant $C>0$ such that $\left\lVert x \right\rVert \le C \left\lVert y \right\rVert$ we use the notation $\left\lVert x \right\rVert = \mathcal O( \left\lVert y \right\rVert).$
For closable operators $A,B$ the tensor product $A \otimes B$ is also closable on $D(A) \otimes D(B)$ and we denote the closure by $A \otimes B$ as well.
For Banach spaces $X,Y$ one has the projective cross norm on the algebraic tensor product $X \otimes Y$
\[ \pi(x) = \inf \left\{ \sum_{i=1}^n \left\lVert a_i \right\rVert \left\lVert b_i \right\rVert; x = \sum_{i=1}^n a_i \otimes b_i  \in X \otimes Y\right\}.\]
The completion of the tensor product space with respect to the projective cross norm is denoted by $X \otimes_{\pi} Y.$ In particular, $\mathcal{H} \otimes_{\pi} \mathcal{H}$ is naturally identified with the space of trace class operators on $\mathcal{H}$. 

Let $A,B$ be positive operators, we write $A \ge B$ if $\mathfrak D(A) \subseteq \mathfrak D(B)$ and $\left\lVert \sqrt{A}x \right\rVert \ge \left\lVert \sqrt{B}x \right\rVert.$ Furthermore, we say $B$ is relatively $A$-bounded with $A$-bound $a$ and bound $b$, if $D(A) \subseteq D(B)$ and for all $\varphi \in D(A)$: $\left\lVert B\varphi \right\rVert  \le a \left\lVert A \varphi \right\rVert + b \left\lVert \varphi \right\rVert.$

We employ a version of Baire's theorem \cite[Theorem $3.8$]{RS1} in our proofs:
\begin{customthm}{[Baire]}
\label{Baire}
Let $X \neq \emptyset$ be a complete metric space and $(A_n)_{n\in {\mathbb{N}}}$ a family of closed sets covering $X$, then there is $k_0 \in {\mathbb{N}}$ for which $A_{k_0}$ has a non-empty interior.
\end{customthm}

\subsection{Quantum Dynamical Semigroups (QDS)}
\label{semigroups}
A quantum dynamical semigroup (QDS) $(T_t)_{t \geq 0}$ in the Schr\"odinger picture is a one-parameter family of bounded linear operators $T_t:\mathcal T_1(\mathcal H) \rightarrow \mathcal T_1(\mathcal H)$ on some Hilbert space $\mathcal H$ with the property that $T_0=\operatorname{id}$ (where $\operatorname{id}$ denotes the identity operator between operator spaces and $I$ the identity acting on the underlying Hilbert space), and $T_tT_s=T_{t+s}$ for all $t,s \geq 0$ (the semigroup property)\footnote{For notational simplicity, we will henceforth suppress the subscript $t \geq 0$ in denoting a QDS.}. In addition, they are completely positive (CP) and trace-preserving (TP). The adjoint semigroup is denoted as $(T^*_t)$, where for each $t \geq 0$, $T^*_t$ is a bounded linear operator on $\mathcal B(\mathcal H)$, which is CP and unital, \emph{i.e.}~$T_t^*(I)= I $ for all $t \geq 0$. Moreover, $T^*_t$
is the adjoint of $T_t$ with respect to the Hilbert Schmidt inner product. Due to unitality, the QDS $(T^*_t)$ is said to be a quantum Markov semigroup (QMS). 

For our purposes we consider the following notions of continuity for semigroups $(S_t)$ defined on a Banach space $X$: 
\begin{itemize}
 \item {\em{uniform continuity}} if $\lim_{t \downarrow 0} \sup_{x \in X; \left\lVert x \right\rVert=1} \left\lVert S_tx-x \right\rVert=0,$
 \item {\em{strong continuity}} if for all $x \in X: \lim_{t \downarrow 0}S_tx =x,$ and
 \item {\em{weak$^*$ continuity}} if for all $y \in X_*$, where $X_{*}$ is the predual Banach space of $X$, and $x \in X$ the map $t \mapsto (S_tx)(y)$ is continuous.
 \end{itemize}

 Uniformly continuous semigroups describe the quantum dynamics of autonomous systems with bounded generators (see \emph{e.g.}~\cite[Theorem $3.7$]{EN}). 
More precisely, every uniformly continuous semigroup $(T_t)$ is of the form $T_t = e^{tA}$ for some bounded {linear} operator $A$. Such an operator $A$ is called the generator of the QDS.
Strongly continuous semigroups describe the quantum dynamics of closed and open quantum systems with unbounded generators in the Schr\"odinger picture, and will be the main object of interest in this paper. The QMS in the Heisenberg picture on infinite-dimensional spaces is, in general, only weak$^*$ continuous: Denoting this QMS by $(\Lambda_t^*)$ for an open quantum system, we have that for all $y \in \mathcal B(\mathcal H)_{*} \equiv \mathcal T_1(\mathcal H)$ and $x \in \mathcal B(\mathcal H)$ the map $t \mapsto (\Lambda_t^* x)(y)$ is continuous. The predual of a weak$^*$ continuous semigroup is known to be strongly continuous \cite[Theorem $1.6$]{EN2}.

The generator of a strongly continuous semigroup $(T_t)$ on a Banach space $X$ is the operator $A$ on $X$ such that
\begin{equation*}
A x =  \lim_{t \downarrow 0} \frac{1}{t} (T_t - I) x, \ \forall \, x \in D(A), \text{ where }D(A) = \left\{ x \in X \,:\, \lim_{t \downarrow 0} \frac{1}{t} (T_tx - x) \,\, {\hbox{exists}}\right\}.
\end{equation*}
In this case, $\frac{d}{dt} T_t x = A T_t x = T_t A x$ and by integrating we obtain for all $x \in D(A)$
\begin{equation}
\label{eq:ftoc}
T_t x - x = \int_{0}^t T_s A x \ ds  = \int_{0}^t A T_s x \ ds. 
\end{equation}
A semigroup $(T_t)$ is called a contraction semigroup if $\left\lVert T_t \right\rVert \le 1$ for all $t \ge 0$ and for any $\lambda >0$ the generator $A$ of such a semigroup satisfies the dissipativity condition
\begin{equation}
\label{eq:contractionprop}
 \left\lVert \lambda (\lambda I-A)^{-1} \right\rVert \le 1
 \end{equation}
For $\lambda>0$ the resolvent of the generator of a contraction semigroup can then be expressed by 
\begin{equation}
\label{eq:resolvent}
(\lambda I -A)^{-1}x = \int_0^{\infty} e^{-\lambda s} T_s x \ ds. 
\end{equation}
QDSs in the Schr\"odinger picture are examples of contraction semigroups. 

\subsection{Functional Analytic Intermezzo}
\label{FAI}
An (unbounded) operator $A$ on some Banach space $X$ with domain $D(A)$ is called closed if its {\em{graph}}, that is \newline 
$\left\{(x,Ax); x \in D(A) \right\} \subseteq X \times X,$ is closed. For a closed operator, a vector space $Y \subseteq D(A)$ is a {\em{core}} if the closure of $A\vert_Y$ coincides with $A$. The spectrum of a closed operator $A$ is the set 
\[\sigma(A):=\left\{ \lambda \in \mathbb C; \lambda I-A \text{ is not bijective } \right\}.\] Its complement is the resolvent set $r(A)$, \emph{i.e.}\ the set of $\lambda$ for which $(\lambda I-A)^{-1}$ exists as a bounded operator.   
Let $A,B$ be two operators defined on the same space and $\lambda \in r(A) \cap r(B)$ then the following resolvent identity holds 
\begin{equation}
\label{eq:resolventidentity}
(\lambda I-A)^{-1}-(\lambda I -B)^{-1} = (\lambda I -A)^{-1} (B-A) (\lambda I -B)^{-1}.
\end{equation}

For any self-adjoint operator $S$ on some Hilbert space $\mathcal{H}$ there is, by the spectral theorem, a spectral measure $\mathcal E^S$ mapping Borel sets to orthogonal projections such that the self-adjoint operator $S$ can be decomposed as \cite[Sec.VII]{RS1}
\[ \left\langle Sx,y \right\rangle  = \int_{\sigma(S)} \lambda \ d\left\langle \mathcal E^S_{\lambda}x,y \right\rangle.\]

In particular, this representation allows us to define a functional calculus for $S$, \emph{i.e.} we can define operators $f(S)$, by setting for any Borel measurable function $f: \mathbb R \rightarrow \mathbb C$
\[ \left\langle f(S)x,y \right\rangle  := \int_{\sigma(S)} f(\lambda) \ d\left\langle \mathcal E^S_{\lambda}x,y \right\rangle,\]
with domain $D(f(S)):=\left\{ x \in \mathcal{H}: \int_{\sigma(S)} \left\lvert f(\lambda) \right\rvert^2 d\left\langle \mathcal{E}_{\lambda}^Sx,x \right\rangle < \infty \right\}.$ In particular, if $f$ is bounded, then $f(S)$ is a bounded operator as well.

The dynamics of a {\em{closed}} quantum system is described by strongly continuous one-parameter QDSs\footnote{As mentioned earlier, since a QDS for a closed system consists of unitary operators, it extends to a group.} according to the following definition:
\begin{defi}
\label{defsemigroup}
Let $\mathcal{H}$ be a Hilbert space. The unitary one-parameter group $(T^{\text{S}}_t)$ (S for Schr\"odinger) on $\mathcal{H}$ is defined through the equation ${\ket{\varphi(t)}} = T^{\text{S}}_t {\ket{\varphi_0}}:=e^{-itH} {\ket{\varphi_0}}$, where ${\ket{\varphi(t)}}$ satisfies the Schr\"odinger equation with initial state $\ket{\varphi_0}$
\begin{equation}
\label{eq:Schroe}
  \partial_t {\ket{\varphi(t)}} = -iH {\ket{\varphi(t)}}, \quad  \ket{ \varphi(0)}=\ket{\varphi_0}.
 \end{equation}
The unitary one-parameter group $(T^{\operatorname{vN}}_t)$ (vN for von Neumann) is defined through the equation $\rho(t)=T^{\operatorname{vN}}_t(\rho_0) := e^{-itH} \rho_0 e^{itH}$, where $\rho(t)$ satisfies the von Neumann equation (on the space of trace class operators $\mathcal T_1(\mathcal H)$) with initial state $\rho_0$
\begin{equation}
\label{eq:vN}
 \partial_t \rho(t) = -i[H,\rho(t)], \quad \rho(0)=\rho_0.
 \end{equation}
\end{defi}
Since the self-adjoint time-independent Hamiltonian $H$ fully describes the above QDSs, we will refer to both $T_t^S$ and $T_t^{\operatorname{vN}}$ as {\em{$H$-associated}} QDSs.

\subsection{A generalized family of energy-constrained diamond norms}
Motivated by the ECD norm introduced in \cite{S18} and \cite{W17} we introduce a generalized family of such energy-constrained norms labelled by a parameter $\alpha \in (0,1]$, which coincides with the ECD norm for $\alpha = 1/2$. We refer to these norms as $\alpha$-energy-constrained diamond norms, or {\em{$\alpha$-ECD norms}} in short. The notion of a {\em{regularized trace}} is employed in the definition of these norms.
 \begin{defi}[Regularized trace]
\label{defi:regtrace}
For positive semi-definite operators $S: D(S) \subseteq \mathcal{H} \rightarrow \mathcal{H}$ and $\rho \in \mathscr D(\mathcal{H})$, we recall that $S^{\alpha} \mathcal E^S_{[0,n]}$ for any $\alpha >0$ is a bounded operator and thus $S^{\alpha} \mathcal E^S_{[0,n]}\rho$ is a trace class operator for which the regularized trace
\[\operatorname{tr}(S^{\alpha}\rho) := \sup_{n \in \mathbb N}\operatorname{tr}(S^{\alpha} \mathcal E^S_{[0,n]}\rho) \in [0,\infty]\text{ is well-defined. }\]
\end{defi}
\begin{defi}[$\alpha$-Energy-constrained diamond ($\alpha$-ECD) norms]
\label{defi:ecdn}
Let $S$ be a positive semi-definite operator and $E > \operatorname{inf}(\sigma(S))$ (where $\sigma(S)$ denotes the spectrum of $S$) then we define for quantum channels $T$, acting between spaces of trace class operators, the $\alpha$-energy constrained diamond norms induced by $S$ for $\alpha \in (0,1]$ as follows:
\begin{equation*}
\begin{split}
 \left\lVert T \right\rVert_{\diamond^{2\alpha}}^{S,E} 
 &= \sup_{n \in \mathbb N} \sup_{\rho \in \mathscr D(\mathcal{H} \otimes \mathbb C^n); E^{2\alpha} \ge \operatorname{tr}(S^{2\alpha} \rho_\mathcal{H} )} \left\lVert T \otimes \operatorname{id}_{B(\mathbb C^n)}(\rho) \right\rVert_{1},
 \end{split}
 \end{equation*}
 where $\rho_{\mathcal H} = \operatorname{tr}_{\mathbb C^n} \rho.$
Moreover, any $\alpha$-ECD norm can be expressed as a standard ECD norm by rescaling both the operator and parameter $E$ as
$\left\lVert T \right\rVert_{\diamond^{2\alpha}}^{S,E}  =  \left\lVert T \right\rVert_{\diamond^{1}}^{S^{2\alpha},E^{2\alpha}}.$ The diamond norm is obtained by setting $E=\infty$ in the above definition. The maximum distance of the $\alpha$-ECD norm between two quantum channels is two.
\end{defi}
Of particular interest to us will be $(i)$ the 1/2-ECD norm $\left\lVert \bullet \right\rVert_{\diamond^{1}}^{S,E},$ which reduces to the ECD norm $\left\lVert \bullet \right\rVert_{\diamond}^{E}$ considered in \cite{S18} and \cite{W17} when $S$ is chosen to be the underlying Hamiltonian, as well as $(ii)$ the $1$-ECD norm $\left\lVert \bullet \right\rVert_{\diamond^{2}}^{S,E}$, since they penalize the first and second moments of the operator $S$, respectively.
Although the operator $S$ in the ECD norm is not necessarily an energy observable (\emph{i.e.}~Hamiltonian), we will refer to the condition $E^{2\alpha} \ge \operatorname{tr}(S^{2\alpha} \rho_\mathcal{H} )$ as an \emph{energy-constraint}.

We show that by studying the entire family of norms, we obtain a more refined analysis for convergence rates of QDSs. Moreover, we allow the generator of the dynamics of the QDS to be different from the operator penalizing the states in the condition $E^{2\alpha} \ge \operatorname{tr}(S^{2\alpha} \rho_\mathcal{H} )$. This does not only allow greater flexibility but also enables us to study open quantum systems since the generator of the dynamics of an open quantum system is not self-adjoint in general and therefore also not positive.

By extending the properties for the ECD norm with $\alpha=1/2$ stated in \cite[Lemma $4$]{W17}, we conclude that:
\begin{itemize}
\item The $\alpha$-ECD norm $\left\lVert \bullet \right\rVert_{\diamond^{2\alpha}}^{S,E}$ defines a norm on QDSs.
\item The $\alpha$-ECD norm $\left\lVert \bullet \right\rVert_{\diamond^{2\alpha}}^{S,E}$ is increasing in the energy parameter $E$ and satisfies for $E'\ge E> \inf(\sigma(S))$ 
\[\left\lVert \bullet \right\rVert_{\diamond^{2\alpha}}^{S,E} \le \left\lVert \bullet \right\rVert_{\diamond^{2\alpha}}^{S,E'} \le \left(\frac{E'}{E}\right)^{2\alpha} \left\lVert \bullet \right\rVert_{\diamond^{2\alpha}}^{S,E}.  \]
\item In the limit $E \rightarrow \infty$ we recover the actual diamond norm 
\[ \sup_{E> \inf(\sigma(S))} \left\lVert \bullet \right\rVert_{\diamond^{2\alpha}}^{S,E} = \left\lVert \bullet \right\rVert_{\diamond}. \]
\item The following calculation shows that the topology for $\alpha \le \beta$ is induced by the $\diamond^{\beta}$ norm is weaker than the one induced by $\diamond^{\alpha}$, \emph{i.e.}\ $\left\lVert T \right\rVert_{\diamond^{2\beta}}^{S,E} \lesssim \left\lVert T \right\rVert_{\diamond^{2\alpha}}^{S,E}$ 
\begin{equation}
\begin{split}
\operatorname{tr} \left( S^{2\alpha} \rho \right)&\stackrel{(1)}{=} \int_{\sigma(S)} \sum_{i=1}^{\infty}  (s^{2\beta} \lambda_i)^{\frac{\alpha}{\beta}} \lambda_i^{\frac{(\beta-\alpha)}{\beta}} d\langle \mathcal E^S_{s} \varphi_i,\varphi_i  \rangle    \\
&\stackrel{(2)}{\le } \left(\int_{\sigma(S)} \sum_{i=1}^{\infty}  s^{2\beta} \lambda_i \ d\langle \mathcal E^S_{s} \varphi_i,\varphi_i  \rangle \right)^{\alpha/\beta}    
\stackrel{(3)}{=}\operatorname{tr} \left( S^{2\beta} \rho  \right)^{\alpha/\beta} 
\end{split}
\end{equation}
where we used the spectral decomposition $\rho = \sum_{i \in \mathbb N} \lambda_i \ket{\varphi_i} \bra{\varphi_i}$ in (1), applied H\"older's inequality such that $1  = \frac{\alpha}{\beta}+\frac{(\beta-\alpha)}{\beta}$ in (2), and rearranged in (3).
 \\
\end{itemize}

\section{Main results}
\label{sec_main}

\subsection{Rates of convergence for quantum evolution} 
Our first set of results concerns bounds on the dynamics of both closed and open quantum systems. The following quantities arise in the bounds for $\alpha \in (0,1]$:
\begin{equation}
\begin{split}
\label{zeta}
\zeta_{\alpha}&:=\left(\tfrac{2\alpha}{1-\alpha} \right)^{1-\alpha} + 2 \left(\tfrac{2\alpha}{1-\alpha} \right)^{-\alpha} \text{ where } \zeta_{1}:=1 \\
g_{\alpha}&:=\zeta_{\alpha}(1-\alpha)^{\frac{1-\alpha}{2}} \alpha^{\frac{\alpha}{2}} .
\end{split}
\end{equation}
When $\alpha=1/2$, the above two expressions reduce to $\zeta_{1/2} = 2\sqrt{2}$ and $g_{1/2}=2.$
Our first Proposition provides a bound on the dynamics of the Schr\"odinger equation \eqref{eq:Schroe}, both in the autonomous and non-autonomous setting: 
\begin{prop}[Closed systems 1]
\label{prop:schroe}
Consider a closed quantum system whose dynamics is governed by an unbounded self-adjoint time-independent Hamiltonian $H$ according to \eqref{eq:Schroe}. Let $\ket{\varphi_0} \in D(\left\lvert H \right\rvert^{\alpha})$ with $\alpha \in (0,1]$. Then the one-parameter group $(T_t^S)$ (\emph{c.f.} \eqref{eq:Schroe} of Definition \ref{defsemigroup}) satisfies, with $g_{\alpha}$ as in \eqref{zeta} and $t,s \ge 0$
\begin{equation}
\label{eq:autonom}
 \left\lVert T^{\text{S}}_t\ket{\varphi_0} -T^{\text{S}}_s \ket{\varphi_0} \right\rVert \le g_{\alpha} \left\lVert \left\lvert H \right\rvert^{\alpha} \ket{\varphi_0} \right\rVert \left\lvert t- s \right\rvert^{\alpha}.
 \end{equation}
For the non-autonomous Schr\"odinger equation
\begin{equation}
  \partial_t \ket{\varphi(t)} = -i(H_0 +V(t))  \ket{\varphi(t)}, \quad  \ket{\varphi(0)}=\ket{\varphi_0},
 \end{equation}
 where $H_0$ and $V(t)$ are self-adjoint and $\int_0^T \left\lVert V(t) \right\rVert \ dt < \infty,$ the time-dependent evolution operators $(U_t)_{t \ge 0}$ defined by ${\ket{\varphi(t)}} =U_t {\ket{\varphi(0)}}$ for any $0  \le s \le t \le T,$ and $\ket{\varphi(0)} \in D(\left\lvert H_0 \right\rvert^{\alpha})$ satisfy
\begin{equation}
\label{eq:nonautonom}
 \left\lVert U_t\ket{\varphi_0} -U_s \ket{\varphi_0} \right\rVert \le g_{\alpha} \left\lVert \left\lvert H_0 \right\rvert^{\alpha} \ket{\varphi_0} \right\rVert ( t-s )^{\alpha}+\int_{s}^t \left\lVert V(r) \right\rVert \ dr.
 \end{equation}
\end{prop}
\begin{center}
\begin{figure}[h!]
\centerline{\includegraphics[height=5cm]{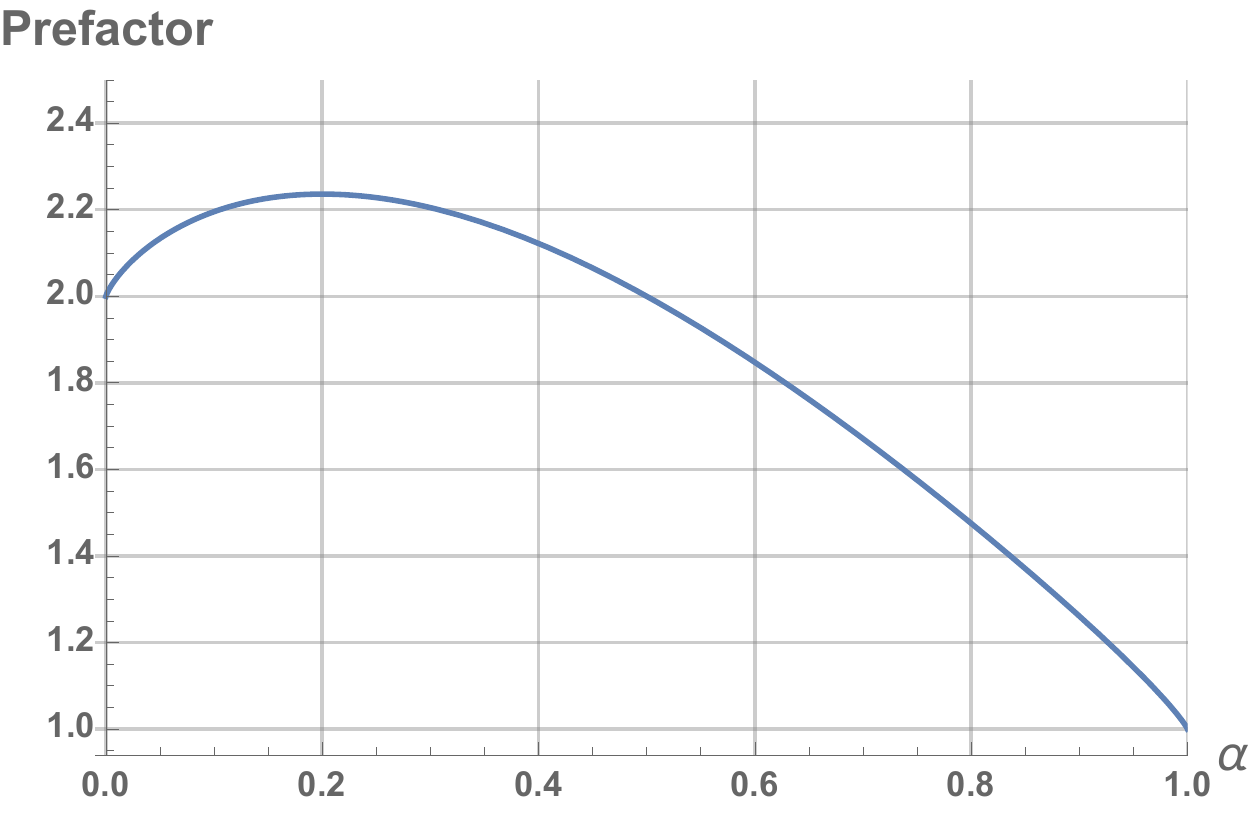}} 
\caption{The dependence of the prefactor  $g_{\alpha}$, in the bound of Proposition \ref{prop:schroe} on the Schr\"odinger dynamics.}
\label{Figure0}
\end{figure}
\end{center}
The bound \eqref{eq:autonom} shows that the dynamics governed by the Schr\"odinger equation is $\alpha$-H\"older continuous in time on sets of $\ket{\varphi} \in \mathcal H$ with uniformly bounded 
$\left\lVert \left\lvert H \right\rvert^{\alpha} \ket{\varphi} \right\rVert.$ 
The bound is also tight, at least for $\alpha=1$, as the prefactor becomes exactly one as $\alpha \rightarrow 1$ which is illustrated in Figure \ref{Figure0}.
From the bound on the dynamics of the Schr\"odinger equation in Proposition \ref{prop:schroe}, we obtain an analogous result for the dynamics of the von Neumann equation \eqref{eq:vN}. The latter result generalizes and improves the bound in \cite[Theorem $6$]{W17}, by providing a bound with rate $t^{1/2}$ rather than $t^{1/3}$ for the ECD norm, which implies faster convergence to zero (see \eqref{t1/3} of the following Proposition and Figure \ref{Figurecomp}):
\begin{prop}[Closed systems 2]
\label{prop:tarate}
Let $\alpha \in (0,1]$. The one-parameter group $T^{\operatorname{vN}}_t(\rho)= e^{-itH}\rho e^{itH}$ solving the von Neumann equation (\eqref{eq:vN} of Definition \ref{defsemigroup}) is $\alpha$-H\"older continuous in time with respect to the $\alpha$-ECD norm introduced in Definition \ref{defi:ecdn} for $E> \inf(\sigma(\vert H \vert))$ where $\sigma(\vert H \vert)$ is the spectrum of $\vert H \vert:$
\begin{equation}
\label{eq:densoper}
\left\lVert T^{\operatorname{vN}}_t-T^{\operatorname{vN}}_s \right\rVert^{\left\lvert H \right\rvert, E}_{\diamond^{2\alpha}} \le 2 g_{\alpha}E^{\alpha}\ \vert t -s \vert^{\alpha}. 
\end{equation}

In particular, when $\alpha=1/2$ we find for the ECD norm 
\begin{equation}
\label{t1/3}
\left\lVert T^{\operatorname{vN}}_t-T^{\operatorname{vN}}_s \right\rVert^{\left\lvert H \right\rvert, E}_{\diamond^{1}} \le 4 E^{1/2}\ \vert t -s \vert^{1/2}. 
\end{equation}
Moreover, for times $\vert t-s \vert^{\alpha} \le 1/(\sqrt{2}g_{\alpha}) $ and pure states $\ket{\varphi} \bra{\varphi} \in \mathscr D(\mathcal H \otimes \mathbb C^n)$ satisfying the energy constraint condition $\operatorname{tr} \left( \ket{\varphi} \bra{\varphi} \left\lvert H \right\rvert^{2\alpha}\right)\le E^{2\alpha}$ one can slightly ameliorate \eqref{eq:densoper} such that
\begin{equation}
\label{eq:purestates}
\left\lVert (T^{\operatorname{vN}}_t-T^{\operatorname{vN}}_s)\ket{\varphi} \bra{\varphi} \right\rVert_{1} \le 2 g_{\alpha} E^{\alpha}\ \vert t-s \vert^{\alpha} \sqrt{1-\tfrac{g_{\alpha}^2 E^{2\alpha}\ \vert t-s \vert^{2\alpha}}{4}}. 
\end{equation}
\end{prop}
\begin{center}
\begin{figure}[h!]
\centerline{\includegraphics[height=5cm]{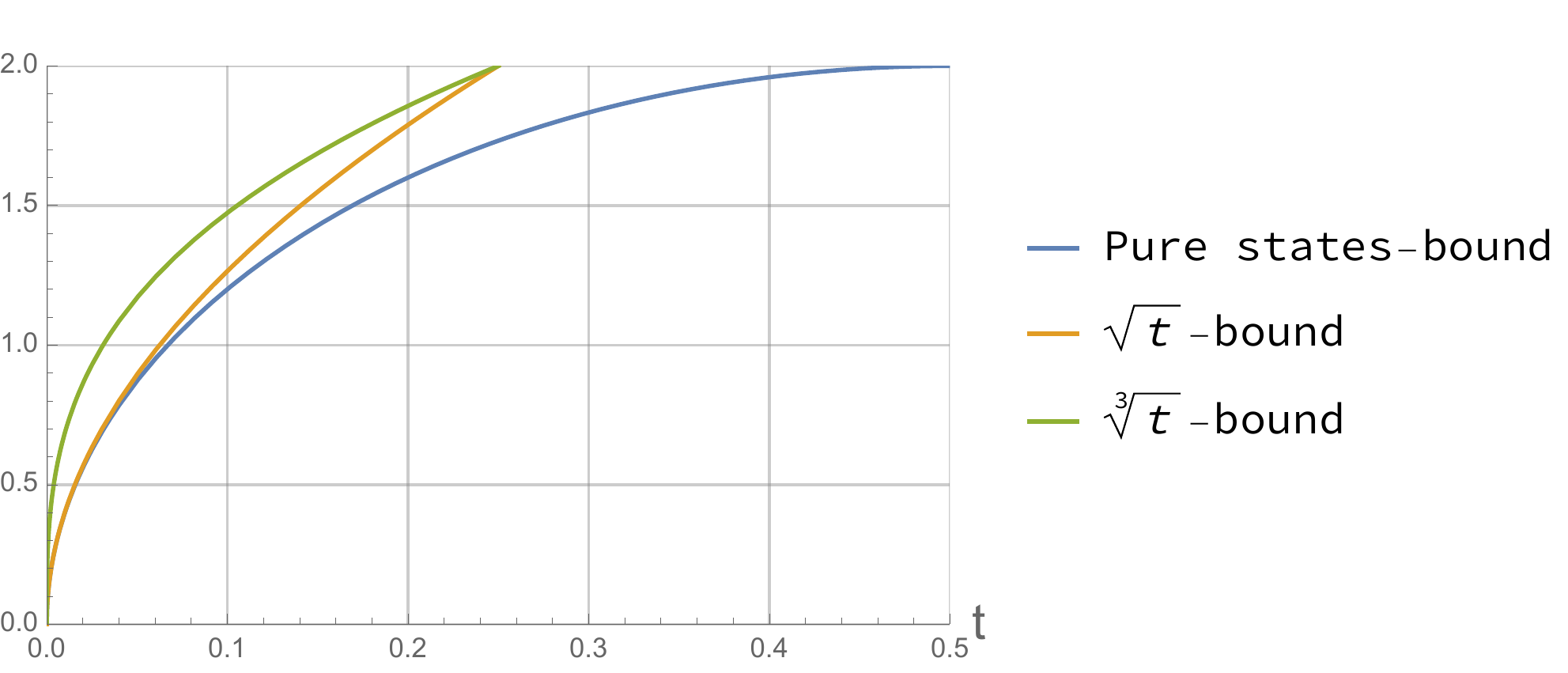}} 
\caption{The $t^{1/3}$-bound obtained in \cite[Theorem $6$]{W17}, the $t^{1/2}$-bound \eqref{eq:densoper}, and the ameliorated bound for pure states \eqref{eq:purestates} for $E=1$.}
\label{Figurecomp}
\end{figure}
\end{center}
In Figure \ref{Figurecomp} we see that estimate \eqref{eq:densoper} globally improves the estimate stated in \cite[Theorem $6$]{W17}. For times larger than the time interval $[0,1/4]$ that is shown in Figure \ref{Figurecomp} the estimates \cite[Theorem $6$]{W17} and \eqref{eq:densoper} exceed the maximal diamond norm distance two of two quantum channels and therefore only provide trivial bounds. The bound on the pure states \eqref{eq:purestates} however, is especially an improvement over the other two \eqref{eq:densoper} for large times.

The above results which are proved in Section \ref{CQS} provide estimates on the dynamics of {\em{closed}} quantum systems.
In Section \ref{OQS} we develop perturbative methods to obtain bounds on the evolution of {\em{open}} quantum systems which have the same time-dependence, \emph{i.e.}~$\alpha$-H\"older continuity in time, as the estimates on the dynamics of closed quantum systems stated in Proposition \ref{prop:tarate}.

We focus on open quantum systems governed by a QDS $(\Lambda_t)$ with a generator which is unbounded but still has a GKLS-type form. The latter is obtained by a direct extension of Theorem \ref{GKLS} under some straightforward assumptions, which are discussed in detail in Section \ref{OQS}.
To state our results on open systems, we define
\begin{equation}
\begin{split}
\label{eq:sigma}
\omega_{H}(\alpha,a,b,c,E)&:= \ 4 \zeta_{\alpha} \operatorname{max} \left\{2c^{\alpha},3bc^{\alpha-1}+(1+3a)   (1-\alpha)^{(1-\alpha)/2} \alpha^{\alpha/2} E^{\alpha} \right\} \text{ and }\\
\omega_{K}(\alpha,a,b,c.E)&:= \ 4 \zeta_{\alpha} \operatorname{max} \left\{2c^{\alpha},3bc^{\alpha-1}+(1+3a)   (1-\alpha)^{(1-\alpha)} \alpha^{\alpha} E^{\alpha}  \right\}.
\end{split}
\end{equation}
In the sequel, we write $\omega_{\bullet}$ to denote either one of them.

\begin{theo}[Open systems]
\label{theo:theo3}
Let $H$ be a self-adjoint operator on a Hilbert space $\mathcal H$ and $(L_l)_{l \in \mathbb N}$ a family of Lindblad-type operators, generalizing the Lindblad operators $L_l$ of Theorem \eqref{GKLS}: $L_l: D(L_l) \subseteq \mathcal H \rightarrow \mathcal H$ with domains satisfying $D(H) \subseteq \bigcap_{l \in \mathbb N} D(L_l)$ such that  $K=-\frac{1}{2}\sum_{l \in \mathbb N} L_l^*L_l$ is dissipative\footnote{ \ $\forall x \in D(K): \langle Kx,x\rangle \le 0$} and self-adjoint with $D(K) \subseteq  \bigcap_{l \in \mathbb N} D(L_l)$. 
Then, let $\alpha \in (0,1]$ and let either of the following conditions be satisfied:

\begin{enumerate}
\item Assume that $K$ is relatively $H$-bounded with $H$-bound $a$ and bound $b.$
If $G:=K-iH$ on $D(H)$ is the generator of a contraction semigroup, then for energies $E>\inf(\sigma(\left\lvert H \right\rvert))$ the QDS $(\Lambda_t)$ of the open system in the Schr\"odinger picture, generated by $\mathcal L$ as in \eqref{eq:L}, satisfies, for any $c>0$ the $\alpha$-H\"older continuity estimate
\[\left\lVert \Lambda_t- \Lambda_s \right\rVert_{\diamond^{2\alpha}}^{\left\lvert H \right\rvert, E} \le \omega_H(\alpha,a,b,c,E) \vert t- s \vert^{\alpha}.\]
For $\alpha=1/2$ the above inequality reduces to
\begin{equation}
\left\lVert \Lambda_t- \Lambda_s \right\rVert_{\diamond^{1}}^{\left\lvert H \right\rvert, E} \le \ \ 8\sqrt{2}  \operatorname{max} \left\{2\sqrt{c},\tfrac{3b}{\sqrt{c}}+(1+3a)    \sqrt{\tfrac{E}{2}}   \right\} \ \sqrt{\vert t- s \vert}.
\end{equation}
For $\alpha=1$ one can take $c \downarrow 0$ to obtain $\omega_H(1,a,b,0,E) = 4(3b+(1+3a)E).$
\item  Assume that $H$ is relatively $K$-bounded with $K$-bound $a$ and bound $b.$
If $G:=K-iH$ on $D(K)$ is the generator of a contraction semigroup, then for energies $E>\inf(\sigma(\left\lvert K \right\rvert))$ the QDS $(\Lambda_t)$ of the open system in the Schr\"odinger picture satisfies, for any $c>0$, the $\alpha$-H\"older continuity estimate
\[\left\lVert \Lambda_t- \Lambda_s \right\rVert_{\diamond^{2\alpha}}^{\left\lvert K \right\rvert, E} \le \ \omega_K(\alpha,a,b,c,E) \vert t- s \vert^{\alpha}.\]
\end{enumerate}
In particular, if $a<1$ then $G$ automatically generates, in either case, a contraction semigroup on $D(H)$.
\end{theo}
While many open quantum systems describe the effect of small dissipative perturbations on Hamiltonian dynamics which is the situation of framework $(1)$ of Theorem \ref{theo:theo3}, there are also examples of open quantum systems which do not have a Hamiltonian dynamics such as the attenuator channel discussed in Example \ref{attchannel}. These systems can be analyzed by case $(2)$ in Theorem \ref{theo:theo3}. 
From these bounds on the dynamics, one can then derive new quantum speed limits which outperform and extend the currently established quantum speed limits in various situations (see also Remark \ref{rem:fail}):
\begin{theo}[Quantum speed limits]
\label{theo:speedlimit}
\smallskip
\noindent
(A) Consider a closed quantum system with self-adjoint Hamiltonian $H$ and fix $E >\inf(\sigma(\left\lvert H \right\rvert))$ and $\alpha \in (0,1]$.
\begin{itemize}
\item The minimal time needed for an initial state $\ket{\varphi(0)}=\ket{\varphi_0}$, for which $E^{2\alpha} \ge \operatorname{tr}(\left\lvert H \right\rvert^{2\alpha}\vert \varphi_0 \rangle \langle \varphi_0 \vert)$, to evolve under the Schr\"odinger equation \eqref{eq:Schroe} to a state $| \varphi(t) \rangle$ with relative angle $\theta := \arccos \left(\Re\langle \varphi(0)| \varphi(t) \rangle  \right) \in [0,\pi]$, satisfies
\begin{equation}
\label{eq:Schroedinger}
  t_{\text{min}} \ge \left(\frac{2-2\cos(\theta)}{g_{\alpha}^2}\right)^{1/(2\alpha)}\frac{1}{E}. 
\end{equation}
For $\alpha=1/2$ this expression reduces to 
\begin{equation}
\label{eq:alpha=1/2}
t_{\text{min}} \ge (1-\cos(\theta))/2 \frac{1}{E}.
\end{equation}
\item Consider an initial state $\rho(0)=\rho_0$ to the von Neumann equation \eqref{eq:vN} with $E^{2\alpha} \ge \operatorname{tr}(\left\lvert H \right\rvert^{2\alpha}\rho_0)$. The minimal time for it to evolve to a state $\rho(t)$ which is at a Bures angle
\begin{equation}\label{Bures}
\theta := \arccos \left( \left\lVert \sqrt{\rho(0)}\sqrt{\rho(t)} \right\rVert_1 \right) \in [0,\pi/2] 
\end{equation}
relative to $\rho(0)$, satisfies
\begin{equation}
\label{eq:vonNeumann}
t_{\text{min}} \ge \left(\frac{1-\cos(\theta)}{g_{\alpha}}\right)^{1/\alpha} \frac{1}{E}.  
\end{equation}
\end{itemize}

(B) Consider an open quantum system governed by a QDS $(\Lambda_t)$ satisfying the conditions of Theorem \ref{theo:theo3}. Let $\rho_0$ denote an initial state, with purity $p_{\text{start}} = \operatorname{tr}(\rho_0^2)$, for which $E^{2\alpha} \ge \operatorname{tr}(\left\lvert H \right\rvert^{2\alpha}\rho_0)$ (or $E^{2\alpha} \ge \operatorname{tr}(\left\lvert K \right\rvert^{2\alpha}\rho_0)$). Then the minimal time needed for this state to evolve to a state with Bures angle $\theta$, satisfies either for $\omega_{H}$ or $\omega_{K}$ as in \eqref{eq:sigma}, where the choice of $\omega_{\bullet}$ depends on whether one considers the situation (1) or (2) in Theorem \ref{theo:theo3},
\begin{equation}
\label{eq:opensys}
  t_{\text{min}} \ge \left(\frac{2-2\cos(\theta)}{\omega_{\bullet}}\right)^{1/\alpha}.  
  \end{equation}

Moreover, the minimal time to reach a state with purity $p_{\text{fin}}$ satisfies
\begin{equation}
\label{eq:puritysys}
 t_{\text{min}} \ge \left(\frac{\vert p_{\text{start}}-p_{\text{fin}} \vert}{\omega_{\bullet}}\right)^{1/\alpha} . 
 \end{equation}

\end{theo}
\subsection{Explicit convergence rates for entropies and capacities}
Our next set of results comprises explicit convergence rates for entropies of infinite-dimensional quantum states and several classical capacities of infinite-dimensional quantum channels, under energy constraints. See Section~\ref{ECB} for definitions, details and proofs. The Hamiltonian arising in the energy constraint is assumed to satisfy the \hyperref[ass:Gibbs]{Gibbs hypothesis}. Continuity bounds on these entropies and capacities rely essentially on the behaviour of the entropy of the Gibbs state $\gamma(E):=e^{-\beta(E)  H}/Z_H(\beta(E)) \in \mathscr D(\mathcal H)$ (where $Z_H(\beta(E))$ is the partition function, for some positive semi-definite Hamiltonian $H$) in the limit $E \rightarrow \infty$. This asymptotic behaviour is studied in Theorem \ref{theo:enttheo}, and discussed for standard classes of Schr\"odinger operators in Example \ref{standard-examples}.
 \begin{Assumption}[Gibbs hypothesis]
\label{ass:Gibbs}
A self-adjoint operator $H$ satisfies the Gibbs hypothesis, if for all $\beta>0$ the operator $e^{-\beta H}$ is of trace class such that the partition function $Z_H(\beta(E)) = \operatorname{tr}(e^{-\beta H})$ is well-defined.
\end{Assumption}
The asymptotic behaviour of the entropy of the Gibbs states allows us then to obtain explicit convergence rates for entropies of quantum states and capacities of quantum channels. 

Consider the following auxiliary functions 
 \begin{equation*}
 \begin{split}
N_{H}^{\uparrow}(E)&:= \sum_{\lambda+\lambda' \le E; \lambda,\lambda' \in \sigma(H)}\lambda^{2} \text{ and } 
N_{H}^{\downarrow}(E):= \sum_{\lambda+\lambda' \le E; \lambda,\lambda' \in \sigma(H)}\lambda \lambda'
 \end{split}
 \end{equation*}
 which depend only on the spectrum of $H.$

We obtain the following explicit convergence rates for the von Neumann entropy $S(\rho)$ of a state $\rho$, and the conditional entropy $H(A|B)_{\rho}$ of a bipartite state $\rho_{AB}$ (defined through \eqref{eq:condent}). For $x \in [0,1]$, we define $h(x):=-x\log(x)-(1-x)\log(1-x)$ (the binary entropy), $g(x):=(x+1)\log(x+1)-x\log(x),$ and  $r_{\varepsilon}(t) = \frac{1+\tfrac{t}{2}}{1-\varepsilon t}$ a function on $(0,\frac{1}{2\varepsilon}]$, with $\varepsilon \in (0,1)$.
 
\begin{ent}{Proposition [Entropy convergence]}
\label{corr:entropyesm}
Let $H$ be a positive semi-definite operator on a quantum system $A$ satisfying the \hyperref[ass:Gibbs]{Gibbs hypothesis} and assume that the limit $\xi:=\lim_{\lambda \rightarrow \infty} \frac{N_{H}^{\uparrow}(\lambda) }{N_{H}^{\downarrow}(\lambda)}>1$ exists such that $\eta:=\left(\xi-1\right)^{-1}$ is well-defined.  \\
For any two states $\rho,\sigma \in \mathscr D(\mathcal H_A)$ satisfying energy bounds $\operatorname{tr}(\rho H), \operatorname{tr}(\sigma H)\le E$ such that $\frac{1}{2} \left\lVert \rho-\sigma \right\rVert_{1} \le \varepsilon \le 1:$
\begin{enumerate}
\item $ \left\lvert S(\rho)-S(\sigma) \right\rvert \le 2 \varepsilon\eta \log\left( E/\varepsilon \right) (1+o(1)) + h(\varepsilon).$
\item Let $\varepsilon< \varepsilon'\le 1$ and $\delta=\frac{\varepsilon'-\varepsilon}{1+\varepsilon'}$, then 
\begin{equation}\label{vN-cont}
\left\lvert S(\rho)-S(\sigma) \right\rvert \le (\varepsilon'+2 \delta) \eta \log\left( E/\delta \right) (1+o(1)) + h(\varepsilon')+h(\delta).\end{equation}
\item For states $\rho,\sigma \in \mathscr D(\mathcal{H}_A \otimes \mathcal H_B)$ with $\operatorname{tr}(\rho_A H), \operatorname{tr}(\sigma_A H)\le E$, $\tfrac{1}{2}\left\lVert \rho-\sigma \right\rVert \le \varepsilon$, and $\varepsilon'$ and $\delta$ as in (2), the conditional entropy \eqref{eq:condent} satisfies
\begin{equation}\label{condl-cont}
\begin{split}
 \left\lvert H(A \vert B)_{\rho}-H(A \vert B)_{\sigma} \right\rvert \le &2(\varepsilon'+4 \delta) \eta \log\left( E/\delta \right) (1+o(1))\\
 & +(1+\varepsilon') h(\tfrac{\varepsilon'}{1+\varepsilon'})+2h(\delta).
 \end{split}
\end{equation}
\end{enumerate}

\end{ent}

For the constrained product-state classical capacity $C^{(1)}$, whose expression is given by \eqref{eq:Hol-cap}, and the constrained classical capacity $C$, defined through \eqref{eq:classcap}, we obtain the following convergence results:
\begin{capa}{Proposition [Capacity convergence]}
\label{corr:Capact}
Consider positive semi-definite operators $H_{A}$ on a Hilbert space $\mathcal H_A$ and $H_{B}$ on a Hilbert space $\mathcal H_B$, where $H_B$ satisfies the \hyperref[ass:Gibbs]{Gibbs hypothesis}. We also assume that the limit $\xi:=\lim_{\lambda \rightarrow \infty} \frac{N_{ H_{B} }^{\uparrow}(\lambda) }{N_{ H_{B} }^{\downarrow}(\lambda)}>1$ exists such that $\eta:=\left(\xi-1\right)^{-1}$ is well-defined.  \\
Let $\Phi, \ \Psi: \mathcal T_1(\mathcal H_A) \rightarrow \mathcal T_1(\mathcal H_B)$ be two quantum channels such that $\tfrac{1}{2} \left\lVert \Phi-\Psi \right\rVert_{\diamond^{1}}^{H_A,E} \le \varepsilon$ for some $\varepsilon \in (0,1)$, and there is a common function $k:\mathbb R^+ \rightarrow \mathbb R^+$ such that
\[ \sup_{\operatorname{tr}(H_{A}\rho) \le E}\operatorname{tr}(H_{B}\Phi(\rho))\le k(E)E \text{ and } \sup_{\operatorname{tr}(H_{A}\rho) \le E}\operatorname{tr}(H_{B}\Psi(\rho))\le k(E)E.\]  
 Then for $t \in (0,\frac{1}{2\varepsilon}]$ the capacities satisfy
\begin{equation}
\begin{split}
\vert C^{(1)}(\Phi,H_{A}, E)- C^{(1)}(\Psi,H_{A}, E)\vert \le &\varepsilon(2t+r_{\varepsilon}(t))\eta \log(k(E)E/(\varepsilon t))(1+o(1))\\
&+2g(\varepsilon r_{\varepsilon}(t)) +2h(\varepsilon t), \text{ as }\varepsilon \downarrow 0 \text{ and} \\
\vert C(\Phi,H_{A}, E)- C(\Psi,H_{A}, E)\vert \le &2\varepsilon(2t+r_{\varepsilon}(t))\eta \log(k(E)E/(\varepsilon t))(1+o(1))\\
&+2g(\varepsilon r_{\varepsilon}(t)) +4h(\varepsilon t), \text{ as }\varepsilon \downarrow 0.
\end{split}
\end{equation}
\end{capa}

\section{Closed quantum systems}
\label{CQS}
In this section we study the dynamics of closed quantum systems in $\alpha$-ECD norms.

From Proposition \ref{prop1} in the appendix it follows that if a state $\rho = \sum_{i=1}^{\infty} \lambda_i \vert \varphi_i \rangle \langle \varphi_i \vert$ satisfies the energy constraint $\operatorname{tr}(S^{2\alpha}\rho)<\infty$ for some positive operator $S$, then all $\ket{\varphi_i}$ for which $\lambda_i \neq 0,$ are contained in the domain of $S^{\alpha}.$ However, the expectation value $\operatorname{tr}(S\rho)$ of an operator $S$ in a state $\rho$ can be infinite even if all the eigenvectors of $\rho$ are in the domain of $S$. This is shown in the following example.

\begin{ex}
\label{ex:potentialwell}
Consider the free Schr\"odinger operator $S:=-\frac{d^2}{dx^2}$ on the interval $[0,\sqrt{1/8}]$ with Dirichlet boundary conditions modeling a particle in a box of length $1/\sqrt{8}$. This operator possesses an eigendecomposition with eigenfunctions $(\psi_i)$ such that $-\frac{d^2}{dx^2} = \sum_{i=1}^{\infty} i^2 \vert \psi_i \rangle \langle \psi_i \vert.$
However, the state $\rho=\sum_{i=1}^{\infty} \frac{1}{i(i+1)} \vert \psi_{i} \rangle \langle \psi_{i} \vert$, here $\sum_{i=1}^{\infty} \frac{1}{i(i+1)}=1,$ satisfies $\operatorname{tr}(S\rho)=\infty.$ 
\end{ex}

Proposition \ref{prop:tarate} implies that any group $T^{\operatorname{vN}}_t(\rho)= e^{-itH}\rho e^{itH}$, with self-adjoint operator $H$, is continuous with respect to the diamond norm induced by $\left\lvert H \right\rvert$ without any further assumptions on $H$ besides self-adjointness.
Before proving this result, we start with the definition of the Favard spaces \cite[Ch.2., Sec.5.5.10]{EN} and an auxiliary lemma:

\begin{defi}[Favard spaces]
Let $(T_t)$ be a contraction semigroup, \emph{i.e.}\ for all $x \in X: \ \left\lVert T_tx \right\rVert \le \left\lVert x \right\rVert$, on some Banach space $X$, then for each $\alpha \in (0,1]$ we introduce \emph{Favard spaces} of the semigroup:
\[ F_{\alpha} \equiv F_{\alpha}((T_t)):=\left\{ x \in X: \left\lvert x \right\rvert_{F_{\alpha}}:=\sup_{t>0} \left\lVert \tfrac{1}{t^{\alpha}} (T_tx-x) \right\rVert < \infty \right\}. \]
\end{defi}
In order to link Favard spaces to QDSs, we require a characterization of these spaces in terms of the resolvent of the associated generator.
\begin{lemm}
\label{Favard}
Let $\alpha \in (0,1]$. Consider a contraction semigroup $(T_t)$ with generator $A$, then $x \in F_{\alpha}$ if and only if  
\[\sup_{\lambda>0} \left\lVert  \lambda^{\alpha} A(\lambda I-A)^{-1}x \right\rVert< \infty\]
in which case $\left\lvert x \right\rvert_{F_{\alpha}}\le \zeta_{\alpha}\sup_{\lambda>0} \left\lVert  \lambda^{\alpha} A(\lambda I-A)^{-1}x \right\rVert$ with $\zeta_{\alpha}$ defined in \eqref{zeta}.
In particular if $X=\mathcal{H}$ is a Hilbert space, then for any one-parameter group $T^{\text{S}}_t=e^{-itH}$ acting on $\mathcal{H},$ where $H$ is self-adjoint, any $x \in D(\left\lvert H \right\rvert^{\alpha})$ belongs to the Favard space $F_{\alpha}$ and
\begin{align}
\left\lvert x \right\rvert_{F_{\alpha}}^2 &\le g_{\alpha}^2 \left\lVert  \left\lvert H \right\rvert^{\alpha}  x \right\rVert^2.
\end{align}
\end{lemm}
\begin{proof}
Let $x \in F_{\alpha}$ then by definition of $F_{\alpha}$ we have $\left\lVert T_tx-x \right\rVert \le \left\lvert x \right\rvert_{F_{\alpha}} t^{\alpha}$ and for $\lambda>0$ 
\begin{equation}
\begin{split}
\label{eq:step1}
\lambda^{\alpha} A(\lambda I-A)^{-1}x &\stackrel{(1)}{=} \lambda^{\alpha+1}(\lambda I-A)^{-1}x - \lambda^{\alpha}x \stackrel{(2)}{=} \lambda^{\alpha+1} \int_0^{\infty} e^{-\lambda s} (T_sx-x) \ ds. 
\end{split}
\end{equation}
We rewrote $A = \lambda + (A-\lambda)$ to get (1) and we used the representation of the resolvent as in \eqref{eq:resolvent} for (2).
Hence, it follows that by taking the norm of \eqref{eq:step1}
\begin{equation*}
\begin{split}
 \sup_{\lambda>0} \left\lVert  \lambda^{\alpha} A(\lambda I-A)^{-1}x \right\rVert &\stackrel{(1)}{\le} \sup_{\lambda>0} \lambda^{\alpha+1} \int_0^{\infty} e^{-\lambda s} \left\lvert x \right\rvert_{F_{\alpha}} s^{\alpha} \ ds \stackrel{(2)}{=} \Gamma(\alpha+1) \left\lvert x \right\rvert_{F_{\alpha}} < \infty
 \end{split}
\end{equation*}
where we used the definition of the Favard spaces $F_{\alpha}$ in (1), and computed the integral to obtain (2).
Conversely, let $x$ satisfy $K:=\sup_{\lambda>0} \left\lVert  \lambda^{\alpha}A(\lambda I-A)^{-1}x \right\rVert  < \infty$ then by decomposing $I = (\lambda I-A)(\lambda I-A)^{-1}$ we can write
\[x = \lambda(\lambda I-A)^{-1}x - A(\lambda I-A)^{-1}x=:x_{\lambda}-y_{\lambda}\]
where now $x_{\lambda} \in D(A).$
Then, using identity \eqref{eq:ftoc} we get (1)
\begin{equation}
\begin{split}
\label{eq:xlambda}
\left\lVert T_tx_{\lambda}-x_{\lambda} \right\rVert  
&\stackrel{(1)}{=} \left\lVert \int_0^t T_sAx_{\lambda}\ ds \right\rVert \stackrel{(2)}{\le} \left\lVert Ax_{\lambda} \right\rVert t\\
&\stackrel{(3)}{\le} \left\lVert \lambda^{\alpha}A(\lambda I-A)^{-1} x  \right\rVert t \lambda^{1-\alpha} \stackrel{(4)}{\le} K t \lambda^{1-\alpha},
\end{split}
\end{equation}
where $(2)$ follows from contractivity of the semigroup, 
and used the definition of $x_{\lambda}$ and $K$ to obtain (3) and (4), respectively.
For $y_{\lambda}$, the triangle inequality and contractivity of the semigroup imply that
\begin{equation}
\label{eq:ylambda}
\left\lVert T_ty_{\lambda}-y_{\lambda} \right\rVert  \le 2 \left\lVert y_{\lambda} \right\rVert \le  2 K \lambda^{-\alpha}. 
\end{equation}
Combining both estimates \eqref{eq:xlambda} and \eqref{eq:ylambda} shows by the triangle inequality
\[  \left\lVert \tfrac{1}{t^{\alpha}} (T_tx-x) \right\rVert \le  K(t\lambda)^{1-\alpha} + 2K (t\lambda)^{-\alpha}.\]
Optimizing the right-hand side over $\lambda>0$ proves that $x \in F_{\alpha}$, since the right-hand side is finite, and $\left\lVert \tfrac{1}{t^{\alpha}} (T_tx-x) \right\rVert \le \zeta_{\alpha}K.$

For $x \in D\left(\left\lvert H \right\rvert^{\alpha}\right)$, one finds that 
\begin{equation}
\begin{split} \left\lVert  \lambda^{\alpha} (-iH)(\lambda I-(-iH))^{-1}x \right\rVert^2 
&\stackrel{(1)}{=} \int_{\mathbb R}  \tfrac{\lambda^{2\alpha} s^2}{\lambda^2+s^2} \ d\left\langle \mathcal E^H(s)x,x \right\rangle  \\
&\stackrel{(2)}{\le} (1-\alpha)^{1-\alpha} \alpha^{\alpha} \int_{\mathbb R}  \left\lvert s \right\rvert^{2\alpha} \ d\left\langle \mathcal E^H(s)x,x \right\rangle  \\&\stackrel{(3)}{=}  (1-\alpha)^{1-\alpha} \alpha^{\alpha} \left\lVert \left\lvert H \right\rvert^{\alpha} x \right\rVert^2.
\end{split}
\end{equation}
Here, we used the functional calculus, see Section \ref{FAI}, in (1), optimized over $\lambda$ to show (2), and used again the functional calculus in (3) which implies the claim.
\end{proof} 
It is known that if the generator $A$ is defined on a Hilbert space $\mathcal H$, then the Favard space $F_1$ coincides with the operator domain $D(A)$ \cite[Corollary $5.21$]{EN}.
As all QDSs are contractive, it suffices to establish a bound at $t=0$, since by contractivity of the semigroup for $t \ge t_0 \ge 0:$
\begin{equation}
\label{eq:ltest}
 \left\lVert (T_t-T_{t_0})x \right\rVert \le \left\lVert T_{t_0} \right\rVert  \left\lVert (T_{t-t_0} - I)x \right\rVert \le  \left\lVert (T_{t-t_0} - I)x \right\rVert 
 \end{equation}
The above lemma then implies Proposition \ref{prop:schroe}, which provides a bound on the dynamics of the Schr\"odinger equation $(T^{\text{S}}_t)$ as shown below.
\begin{proof}[Proof of Proposition \ref{prop:schroe}]
The result on the autonomous dynamics follows directly by rearranging the estimate $\left\lvert \ket{\varphi_0} \right\rvert_{F_{\alpha}}^2 \le \zeta_{\alpha}^2 (1-\alpha)^{1-\alpha} \alpha^{\alpha} \left\lVert  \left\lvert H \right\rvert^{\alpha}  \ket{\varphi_0} \right\rVert^2$
from Lemma \ref{Favard} and using \eqref{eq:ltest} to transfer the result to arbitrary times $t,s$. 
The non-autonomous result follows from the variation of constant formula
\[ U_t\ket{\varphi_0} = e^{-itH_0} \ket{\varphi_0} - i \int_{0}^t e^{-i(t-r)H_0} V(r) U_r \ket{\varphi_0} \ dr \] 
such that by using the result for the autonomous semigroup we obtain
\[ \left\lVert (U_t-I) \ket{\varphi_0} \right\rVert \le g_{\alpha} \left\lVert \left\lvert H \right\rvert^{\alpha} \ket{\varphi_0} \right\rVert t^{\alpha} +\int_{0}^t \left\lVert V(r) \right\rVert \ dr,  \]
where $g_{\alpha}$ is given by \eqref{zeta}.
The general result follows by considering the initial state $U_s \ket{\varphi_0}$ at initial time $t_0=s.$
\end{proof}
\begin{center}
\begin{figure}[h!]
\centerline{\includegraphics[height=5cm]{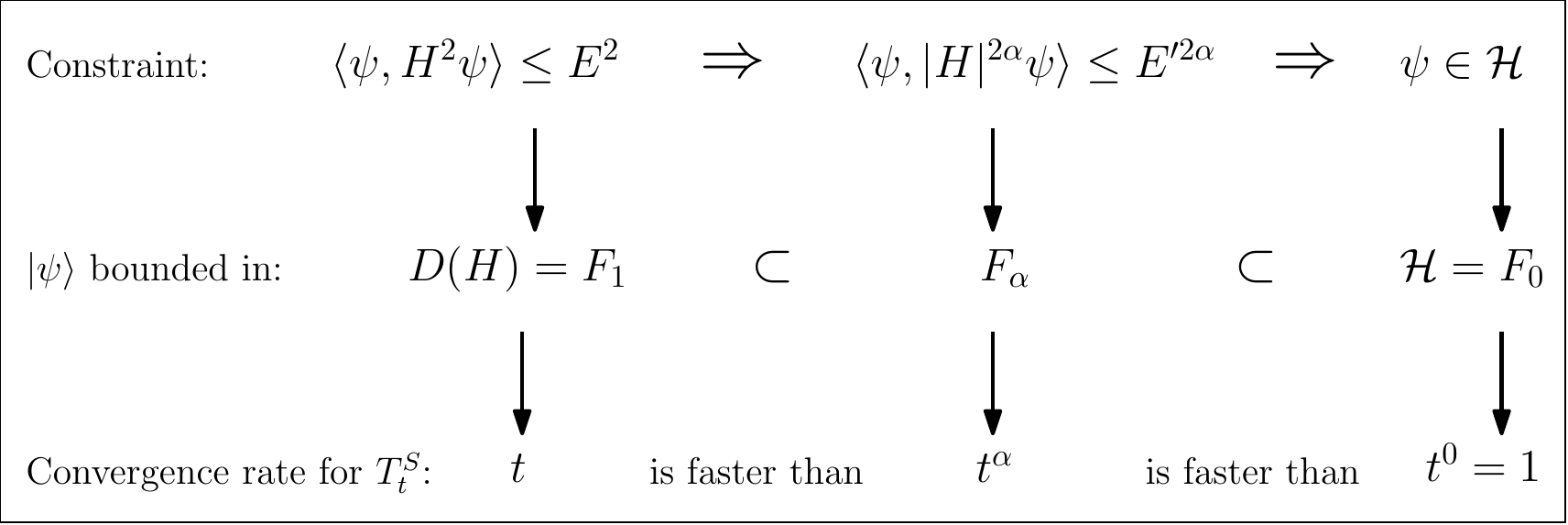}} 
\caption{For a normalized $\ket{\psi}$ in a Hilbert space $\mathcal H$ we illustrate the connection between energy constraints, Favard spaces, and convergence rates for the Schr\"odinger equation with Hamiltonian $H$ in closed quantum systems with $\alpha \in (0,1]$ as in Proposition \ref{prop:schroe}.}
\label{Figure1}
\end{figure}
\end{center}
Before extending the above result to the dynamics of the von Neumann equation \eqref{eq:vN} for states on the product space $\mathcal{H} \otimes \mathbb C^n$, we need another auxiliary Lemma on the action of the Schr\"odinger dynamics on states:
\begin{lemm}
\label{Schr}
The tensor product of the strongly continuous one-parameter group $T^{\text{S}}_t =e^{-itH}$ for $H$ self-adjoint on $\mathcal{H}$ with the identity $\operatorname{id}_{\mathcal B(\mathbb C^n)}$ acting on states $\rho \in \mathscr D(\mathcal{H} \otimes \mathbb C^n)$ satisfies for $\alpha \in (0,1]$
\begin{equation}
\label{eq:auxbound}
\left\lVert (T^{\text{S}}_t \otimes \operatorname{id}_{\mathcal B(\mathbb C^n)} -\operatorname{id})(\rho) \right\rVert_{1} \le g_{\alpha}\sqrt{\operatorname{tr}\left((\left\lvert H \right\rvert^{2\alpha} \otimes I_{\mathbb C^n})\rho\right)} t^{\alpha}. 
\end{equation}
\end{lemm}
\begin{proof}
The generator of $(T^{\text{S}}_t \otimes \operatorname{id}_{\mathcal B(\mathbb C^n)})$ acting on trace class operators is the operator $-iH \otimes  I_{\mathbb C^n}$ acting on some set of trace class operators \cite[Section A-I 3.7]{NS86}. 
Using the results from Lemma \ref{Favard} it suffices to bound for $\lambda>0$
\[ \left\lVert  \lambda^{\alpha} (-iH \otimes I_{\mathbb C^n})(\lambda I-(-iH\otimes I_{\mathbb C^n}))^{-1}\sqrt{\rho} \sqrt{\rho} \right\rVert^2_{1}\] 
accordingly. 
From the spectral decomposition $\rho= \sum_{i=1}^{\infty} \lambda_i \vert \varphi_i \rangle \langle \varphi_i \vert$ of a state, the claim then follows immediately from the following bound
\begin{equation}
\begin{split} 
\label{eq:Hestimate}
&\left\lVert  \lambda^{\alpha} (-iH \otimes I_{\mathbb C^n})(\lambda I-(-iH\otimes I_{\mathbb C^n}))^{-1}\sqrt{\rho} \sqrt{\rho} \right\rVert^2_{1} \\
&\stackrel{(1)}{\le} \lambda^{2\alpha} \operatorname{tr}\left(\frac{-iH\otimes I_{\mathbb C^n}}{\lambda I+iH\otimes I_{\mathbb C^n}}\rho \frac{iH\otimes I_{\mathbb C^n}}{\lambda I-iH\otimes I_{\mathbb C^n}}\right)\\
&\stackrel{(2)}{=} \sum_{i=1}^{\infty} \lambda_i \int_{\mathbb R} \tfrac{ \lambda^{2\alpha}s^2}{\lambda^2+s^2} \ d\left\langle \mathcal E^{H\otimes I_{\mathbb C^n}}_s\varphi_i,\varphi_i \right\rangle  \stackrel{(3)}{\le} \sum_{i=1}^{\infty} \lambda_i (1-\alpha)^{1-\alpha} \alpha^{\alpha} \left\lVert \left\lvert H \otimes  I_{\mathbb C^n}\right\rvert^{\alpha} \ket{\varphi_i }\right\rVert^2\\
& \stackrel{(4)}{\le} (1-\alpha)^{1-\alpha} \alpha^{\alpha} \operatorname{tr}\left(  \left\lvert H \right\rvert^{2\alpha} \otimes I_{\mathbb C^n} \rho \right),
\end{split}
\end{equation}
where we applied H\"older's inequality in (1), used the spectral decomposition of the state and the functional calculus, as in Section \ref{FAI}, in (2), optimized over $\lambda$ and applied the functional calculus again in (3), and used in (4) again the spectral decomposition of the state, as well as 
\[  \left\lvert H \otimes  I_{\mathbb C^n}\right\rvert^{2\alpha} = \left\lvert \operatorname{diag}(H,..,H) \right\rvert^{2\alpha} = \operatorname{diag}(\left\lvert H \right\rvert^{2\alpha},.., \left\lvert H \right\rvert^{2\alpha}) =\left\lvert H \right\rvert^{2\alpha}  \otimes  I_{\mathbb C^n}.\]
\end{proof}

From estimate \eqref{eq:auxbound} we can then state the proof of Proposition \ref{prop:tarate}:
\begin{proof}[Proof of Proposition \ref{prop:tarate}]
From a simple application of the triangle inequality and the unitary quantum evolution we conclude that
\begin{equation*}
\begin{split}
 \tfrac{1}{t^{\alpha}}\left\lVert(T^{\operatorname{vN}}_t \otimes \operatorname{id}_{\mathcal B(\mathbb C^n)} -\operatorname{id})( \rho) \right\rVert_{1} 
& = \tfrac{1}{t^{\alpha}}\left\lVert (T^{\operatorname{S}}_t \otimes I_{\mathbb C^n}) \rho (T^{\operatorname{S}}_{-t} \otimes I_{\mathbb C^n})- \rho \right\rVert_{1} \\
& \le  \tfrac{2}{t^{\alpha}} \left\lVert (T^{\text{S}}_t \otimes I_{\mathbb C^n} -I)\rho \right\rVert_{1}  
\end{split}
\end{equation*}
such that by applying Lemma \ref{Schr} in (1) and the energy constraint in (2), we obtain the result for the ECD norms
\begin{equation*}
\begin{split}
\tfrac{1}{t^{\alpha}}\left\lVert (T^{\operatorname{vN}}_t \otimes \operatorname{id}_{\mathcal B(\mathbb C^n)} -\operatorname{id})( \rho) \right\rVert_{1} &\le \tfrac{2}{t^{\alpha}} \left\lVert (T^{\text{S}}_t \otimes \operatorname{id}_{\mathcal B(\mathbb C^n)} -\operatorname{id})(\rho)  \right\rVert_{1}  \\
&\stackrel{(1)}{\le} 2\zeta_{\alpha}(1-\alpha)^{\tfrac{1-\alpha}{2}} \alpha^{\tfrac{\alpha}{2}}\sqrt{\operatorname{tr}\left((\left\lvert H \right\rvert^{2\alpha} \otimes I_{\mathbb C^n})\rho\right)}\\
&\stackrel{(2)}{\le}  2\zeta_{\alpha}(1-\alpha)^{\tfrac{1-\alpha}{2}} \alpha^{\tfrac{\alpha}{2}}E^{\alpha}. 
\end{split}
\end{equation*}
The estimate on pure states follows immediately from Proposition \ref{prop:schroe} after expressing the trace distance in terms of the Hilbert space norm.
\end{proof}
The preceding Propositions \ref{prop:schroe} and \ref{prop:tarate} show that the quantum dynamics of closed quantum systems generated by some self-adjoint operator $H$ is always continuous with respect to the $\alpha$-ECD norm induced by the absolute value of the same operator $H$.

We now do a perturbation analysis for the convergence in $\alpha$-ECD norm:
\begin{prop}
\label{corr2}
Let $H$ be a self-adjoint operator, $\alpha \in (0,1]$ and $\left\lvert H \right\rvert^{\alpha}$ relatively $S^{\alpha}$-bounded in the sense of squares where $S$ is a positive semi-definite operator, \emph{i.e.}\ $D(S^{\alpha}) \subseteq D(\left\lvert H \right\rvert^{\alpha})$ and there are $a,b \ge 0$ such that for all $\varphi \in  D(S^{\alpha}): \ \left\lVert \left\lvert H \right\rvert^{\alpha}  \varphi \right\rVert^2 \le a \left\lVert S^{\alpha}  \varphi \right\rVert^2 + b \left\lVert \varphi \right\rVert^2.$
Then, the $H$-associated strongly continuous semigroup $T^{\operatorname{vN}}_t\rho= e^{-itH}\rho e^{itH}$ is $\alpha$- H\"older continuous with respect to the $\alpha$-ECD norm. Moreover, there is the inequality of norms $\left\lVert \bullet \right\rVert_{\diamond^{2\alpha}}^{S,E} \le \left\lVert \bullet \right\rVert_{\diamond^{2\alpha}}^{\left\lvert H \right\rvert,(a E^{2\alpha}+b)^{1/(2\alpha)}}$ such that
\[ \left\lVert T_t-T_s \right\rVert_{\diamond^{2\alpha}}^{S,E} \le 2g_{\alpha} \sqrt{a E^{2\alpha}+b} \ \vert t- s \vert^{\alpha}.\]
In particular, if $S^{\alpha}$ is also relatively $\left\lvert H \right\rvert^{\alpha}$-bounded, then the $\alpha$-ECD norms $\left\lVert \bullet \right\rVert_{\diamond^{2\alpha}}^{S,E}$ and $\left\lVert \bullet \right\rVert_{\diamond^{2\alpha}}^{\vert H \vert,E}$ are equivalent.
\end{prop}
\begin{proof}
Consider a density matrix with spectral decomposition $\rho = \sum_{i=1}^{\infty} \lambda_i \ \vert \varphi_i \rangle \langle \varphi_i \vert.$ If any of the $\ket{\varphi_i} \notin D(S^{\alpha})$ then $\operatorname{tr}(S^{2\alpha}\rho) = \infty$ as in Proposition \ref{prop1}. Thus, we may assume that all $\ket{\varphi_i} \in D(S^{\alpha}).$ Therefore, if $\operatorname{tr}(S^{2\alpha}\rho) \le E^{2\alpha}$  then also $\operatorname{tr}(\left\lvert H \right\rvert^{2\alpha} \rho) \le a \operatorname{tr}( S^{2\alpha}\rho) + b \le  a E^{2\alpha}+b$ which proves the Proposition, since the estimate follows from Proposition \ref{prop:tarate}. 
\end{proof}
The previous result allows us to study QDSs generated by complicated Hamiltonians using more accessible operators penalizing the states in the ECD norms. We illustrate this in the following example where we see that it suffices to penalize the kinetic energy of a state and still obtain convergence for the semigroup of the Coulomb Hamiltonian. 
\begin{ex}[Coulomb potential]
Let $H:=-\Delta-\frac{1}{\left\lvert x \right\rvert}$ and $S:=-\Delta$, then $H$ is relatively $S$-bounded. Thus, the semigroup $T_t^{\operatorname{vN}}(\rho):=e^{-itH} \rho e^{itH}$ is $\alpha$-H\"older continuous in time with respect to $\left\lVert \bullet \right\rVert_{\diamond^{2\alpha}}^{S,E}.$
\end{ex}
We provide a simple example showing that it is impossible to select arbitrary unbounded self-adjoint operators to penalize the energy in the diamond norm and still have the same convergence rates in time:
 \begin{ex}[Harmonic oscillator]
Let $H_{\operatorname{osc}}:=-\Delta+\left\lvert x \right\rvert^2$ be the dimensionless Hamiltonian of the harmonic oscillator on $D(H_{\operatorname{osc}}):=\left\{ \varphi \in L^2(\mathbb R^d); \Delta \varphi, \ \left\lvert x \right\rvert^2\varphi \in L^2(\mathbb R^d) \right\}.$ 
The one-parameter group of the harmonic oscillator $T^{\operatorname{vN}}_t(\rho):=e^{-itH_{\operatorname{osc}}} \rho e^{itH_{\operatorname{osc}}}$ does not obey a uniform linear time-rate in the $1$-ECD norm induced by the negative Laplacian $-\Delta$ for any $E>0 = \inf(\sigma(-\Delta)).$ 
To see this, it suffices to study the dynamics generated by the Schr\"odinger equation \eqref{eq:Schroe} with Hamiltonian $H_{\operatorname{osc}}$. Then, the Favard space $F_1$ coincides with the operator domain $D(H_{\operatorname{osc}}),$ as stated in \cite[Corollary $5.21$]{EN}. However, the domain of the Laplacian penalizing the energy is $D(-\Delta)=\left\{ f \in L^2(\mathbb R^d); -\Delta f \in L^2(\mathbb R^d) \right\}$ which is strictly larger than $F_1=D(H_{\operatorname{osc}}),$ as for $f \in D(-\Delta)$ one does not require that $\vert x \vert^2 f \in L^2(\mathbb R^d).$
 \end{ex}

The perturbation result, Proposition \ref{corr2}, essentially relies on operator boundedness and provides explicit bounds to compare the two different $\alpha$-ECD norms induced by the perturbed and unperturbed operator. This result is a special case of a more abstract result, stated as Proposition \ref{theorem1} in Appendix B, that relies on the special geometry of the space of trace class operators. It yields the same rate $t^{\alpha}$ for the convergence with respect to the perturbed and unperturbed norms. However, it does not provide an explicit prefactor.

\section{Open quantum systems}
\label{OQS}
We start with an auxiliary Lemma that provides sufficient conditions under which a perturbation of the generator of a contraction semigroup leaves its Favard spaces invariant:
\begin{lemm}[Perturbation of Favard spaces]
\label{Favpet}
Let $A_0$ and $A=A_0+B$ be two generators of contraction semigroups on some Banach space $X$. Furthermore, we fix some $\alpha \in (0,1].$ 
Let $\lambda>0 $ and $B$ be relatively $A_0$-bounded with $A_0$-bound $a\ge 0$ and bound $b\ge 0.$
Then, for any $k \ge 0$ such that  
\[\sup_{\lambda >0} \left\lVert \lambda^{\alpha} A_0(\lambda I-A_0)^{-1} x\right\rVert \le k,\] we have for all $c>0$
\[ \sup_{\lambda >0}\left\lVert \lambda^{\alpha} A(\lambda I-A)^{-1} x \right\rVert \le \operatorname{max}\left\{2c^{\alpha}\left\lVert x \right\rVert,3bc^{\alpha-1}\left\lVert x \right\rVert+ (1+3a) k \right\}< \infty.\]
In particular, the Favard space $F_{\alpha}$ of the semigroup generated by $A_0$ is contained in the Favard space $F_{\alpha}$ of the semigroup generated by $A.$
\end{lemm} 
\begin{proof}

Fix $c>0$, then for $\lambda \in (0,c]$ it follows that  $\left\lVert \lambda^{\alpha} A(\lambda I-A)^{-1} \right\rVert \le 2 \lambda^{\alpha} \le 2c^{\alpha}$
where we used that by \eqref{eq:contractionprop} and the triangle inequality,
\begin{equation}
\label{eq:formel}
\left\lVert A(\lambda I-A)^{-1} \right\rVert \le \left\lVert (\lambda I-A)(\lambda I-A)^{-1} \right\rVert + \left\lVert \lambda(\lambda I-A)^{-1}\right\rVert \le 2.
\end{equation}
For $\lambda>c$ we obtain from the resolvent identity \eqref{eq:resolventidentity} and the triangle inequality
\begin{equation}
\label{eq:resident}
\left\lVert \lambda^{\alpha} A(\lambda I-A)^{-1} x\right\rVert  \le    \left\lVert \lambda^{\alpha}A (\lambda I-A_0)^{-1} x\right\rVert +  \left\lVert \lambda^{\alpha}A(\lambda I-A)^{-1} B (\lambda I-A_0)^{-1} x\right\rVert. 
\end{equation}
By relative $A_0$-boundedness of $B$ we obtain for the first term on the right-hand side of \eqref{eq:resident} by splitting up $A=A_0 + B$
\[ \left\lVert \lambda^{\alpha}A (\lambda I-A_0)^{-1} x\right\rVert \le (1+a) \left\lVert \lambda^{\alpha}A_0 (\lambda I-A_0)^{-1} x\right\rVert+b \left\lVert  \lambda^{\alpha}  (\lambda I-A_0)^{-1}x \right\rVert.\]
For the second term on the right-hand side of \eqref{eq:resident}, we can use \eqref{eq:formel} and submultiplicativity to bound
\begin{equation*}
\begin{split}
 \left\lVert \lambda^{\alpha}A(\lambda I-A)^{-1} B (\lambda I-A_0)^{-1} x\right\rVert \le 2  \left\lVert  \lambda^{\alpha}B (\lambda I-A_0)^{-1} x\right\rVert.
 \end{split}
\end{equation*}
Again, using the relative $A_0$-boundedness of $B$ we can estimate the last term 
\begin{equation*}
\begin{split}
&2  \left\lVert  \lambda^{\alpha}B (\lambda I-A_0)^{-1} x\right\rVert \le 2a \left\lVert \lambda^{\alpha}A_0 (\lambda I-A_0)^{-1} x\right\rVert+2b \left\lVert \lambda^{\alpha}  (\lambda I -A_0)^{-1}x \right\rVert.
\end{split}
\end{equation*}
Thus, since $A_0$ generates a contraction semigroup, it follows by \eqref{eq:contractionprop} that $\left\lVert \lambda (\lambda I-A_0)^{-1} \right\rVert \le 1$, and since $\lambda >c$ 
\[ \left\lVert \lambda^{\alpha}  (\lambda I-A_0)^{-1}x \right\rVert \le  \lambda^{\alpha-1} \left\lVert \lambda  (\lambda I-A_0)^{-1}x \right\rVert \le c^{\alpha-1} \left\lVert x \right\rVert\]
such that we finally obtain the claim of the lemma by putting all estimates together 
\begin{equation*}
\begin{split}
\left\lVert \lambda^{\alpha} A(\lambda I-A)^{-1} x\right\rVert \le (1+3a) \left\lVert \lambda^{\alpha}A_0 (\lambda I-A_0)^{-1} x\right\rVert+3bc^{\alpha-1}  \left\lVert x \right\rVert.
\end{split}
\end{equation*}
\end{proof}

The most general form of the generator of a \emph{uniformly continuous} QMS is the so-called GKLS representation, named after Lindblad~\cite{Lin} and Gorini, Kossakowski and Sudarshan~\cite{GKS}. 
\begin{customthm}{GKLS}
\label{GKLS}
Let $(\Lambda_t)$ be a uniformly continuous semigroup in the Schr\"odinger picture on the space of trace class operators $\mathcal T_1(\mathcal H).$ Its adjoint semigroup is a uniformly continuous semigroup $(\Lambda_t^*)$ on the space of bounded linear operators on $\mathcal{H}$ and defines a QMS on $\mathcal B(\mathcal{H})$ if and only if there are Lindblad operators $L_l \in \mathcal B(\mathcal{H})$ and an operator $G \in \mathcal B(\mathcal{H})$ such that the bounded generator $\mathcal L^*$ of $(\Lambda_t^*)$ satisfies for all $S \in \mathcal B(\mathcal{H})$
\begin{equation*}
\begin{split}
 \mathcal L^*(S)&=\sum_{l \in \mathbb N} L_l^*SL_l+ G^*S+SG \text{ and }\sum_{l \in \mathbb N} L_l^*L_l + G^*+G=0.
 \end{split}
 \end{equation*}
 In particular, $G$ can be written as $G = -\frac{1}{2}\sum_{l \in \mathbb N} L_l^*L_l - iH$ where $H$ is bounded and self-adjoint.
\end{customthm}
This construction has been generalized by Davies  \cite{Da} to unbounded generators which is discussed below:
\subsection{Extension of GKLS theorem to unbounded generators \cite{Da}}

Let $G: D(G) \subseteq \mathcal{H}\rightarrow \mathcal{H}$ be the generator of a contractive strongly continuous semigroup, that we denote by $(P_t)_{t\ge 0}$ in the sequel, and consider Lindblad-type operators $(L_l)_{l \in \mathbb{N}}$. These form a (possibly finite) sequence of bounded or unbounded operators on $\mathcal{H}$ satisfying $D(G) \subseteq D(L_l)$ for every $l \in \mathbb{N}$ such that for all $x,y \in D(G):$ 
\begin{equation}
\label{eq:conseq}
\langle Gx,y \rangle + \langle x, Gy \rangle + \sum_{l \in \mathbb N} \langle L_l x, L_l y \rangle =0.
\end{equation}
Acting on arbitrary bounded operators $S \in \mathcal B(\mathcal H)$ we introduce the generator of the QDS $(\Lambda_t^*)$ in a weak formulation for $x,y \in D(G)$ 
\begin{equation}
\label{eq:weakrepresent}
\mathcal L^*(S)(x,y) = \left\langle Gx,Sy \right\rangle  + \sum_{l \in \mathbb N} \left\langle L_l x,S L_l y \right\rangle + \left\langle x,SGy \right\rangle.
\end{equation}
Under the preceding assumptions, it can be shown \cite{Da} that the QDS $(\Lambda_t^*)$ is weak$^*$ continuous on $\mathcal B(\mathcal{H})$ satisfying for all $x,y \in D(G)$ and $S \in \mathcal B(\mathcal H)$
\begin{equation}
\label{eq:fdoc}
\left\langle x, \Lambda_t^*(S) y \right\rangle  = \left\langle x,S y \right\rangle + \int_0^t \mathcal L^*(\Lambda_s^*(S))(x,y) \ ds.
\end{equation}
Among all such semigroups satisfying the preceding equation, we consider henceforth the minimal semigroup, which always exists \cite[Theorem $6.1.9$]{Chang}, satisfying for all bounded operators $S$ the inequality $\Lambda_t^{\text{min}*}(S) \le \Lambda_t^*(S)$. The minimal semigroup will in the sequel just be denoted by $(\Lambda_t^*)$ again. We also assume that this semigroup $(\Lambda_t^*)$ is Markovian, \emph{i.e.}\ $\Lambda_t^*(I)=I.$ Direct methods to verify the Markov property for a minimal semigroup, are for example due to Chebotarev and Fagnola \cite[Theorem $4.4$]{CF}. 

Since $(\Lambda_t^*)$ is a weak$^*$ continuous semigroup, the predual semigroup $\Lambda_t$ acting on trace class operators is a strongly continuous semigroup generated by the adjoint of $\mathcal L$. By the Markov property of the adjoint semigroup \cite[Proposition $6.3.6$]{Chang}, the vector space given by $\operatorname{span}\left\{ \vert \varphi \rangle \langle \psi  \vert; \varphi, \psi \in D(G)\right\}$ is a core for $D(\mathcal L)$ and 
\begin{equation}
\label{eq:L}
\mathcal L(\vert \varphi \rangle \langle \psi \vert)=\vert G  \varphi \rangle \langle  \psi \vert+ \vert \varphi \rangle \langle G\psi \vert + \sum_{l \in \mathbb N} \vert L_l \varphi \rangle   \langle L_l \psi \vert, 
\end{equation}
where the series converges in trace norm. To keep the notation short, we write $\widehat{X}=X \otimes I_{\mathbb C^n}$ for operators $X$ on $\mathcal{H}$ and $\widehat{X} = X \otimes \operatorname{id}_{\mathcal B(\mathbb C^n)}$ for super-operators. Then, by inserting \eqref{eq:weakrepresent} into \eqref{eq:fdoc} it follows that for all $S \in \mathcal B(\mathcal{H} \otimes \mathbb C^n)$ and $x,y \in D(G) \otimes \mathbb C^n$
\begin{equation}
\begin{split}
\label{eq:standardform}
\left\langle x,\widehat{\Lambda}_t^*({S}) y \right\rangle &= \left\langle x,{S}y \right\rangle +  \sum_{l=1}^{\infty} \int_0^t \left\langle \widehat L_l x, \widehat{\Lambda}_s^*({S})\ \widehat L_l y \right\rangle \ ds \\
&+ \int_0^t \left(\left\langle x, \widehat{\Lambda}_s^*({S}) \widehat{G} y \right\rangle +\left\langle \widehat{G} x, \widehat{\Lambda}_s^*({S})  y \right\rangle \right) \ ds. 
\end{split}
\end{equation}
Direct computations show that the QMS satisfies \cite[Prop.\ $6.1.3.$]{Chang}:
\begin{equation}
\begin{split}
\label{standardform2}
\left\langle x,\widehat{\Lambda}_t^*({S}) y \right\rangle  &= \left\langle \widehat{P}_t x,{S} \widehat{P}_ty \right\rangle+ \sum_{l=1}^{\infty} \int_0^t \left\langle  \widehat{L_l}  \widehat{P}_{t-s} x, \widehat{\Lambda}_s^*({S}) \widehat{L_l} \widehat{P}_{t-s} y \right\rangle \ ds. 
\end{split}
\end{equation}
By the representation of the QMS in \eqref{standardform2}, bounds on the dynamics of the full, possibly intricate, QDS $(\widehat{\Lambda}_t)$ can be found using the simpler semigroup $\left( \widehat{P}_t \right)$ as the subsequent Lemma shows:
\begin{lemm}
\label{reduction}
For arbitrary $n \in \mathbb N$ and states $\rho \in \mathscr D(\mathcal H \otimes \mathbb C^n)$ we have  
\begin{equation*}
\begin{split}
&\left\lVert (\widehat{\Lambda}_t-\operatorname{id})(\rho) \right\rVert_{1} \le 4 \left\lVert  ( \widehat{ P_t} -\operatorname{id} )(\rho) \right\rVert_{1}.
\end{split}
\end{equation*}
\end{lemm}
\begin{proof}
Consider an approximation of $\rho \in \mathscr D(\mathcal H \otimes \mathbb C^n)$ in trace norm by finite-rank operators $\rho_m:=\sum_{i=1}^m \lambda_i  \ \vert u_i \rangle \langle u_i \vert \text{ with } \ket{u_i} \in D(G)\otimes \mathbb C^n$
and $\lambda_i \ge 0$ such that $\rho_m \xrightarrow[m \rightarrow \infty]{} \rho$ in trace norm. Then we estimate
\begin{equation}
\begin{split}
\label{eq:distrib}
& \left\lVert (\widehat{\Lambda}_t-\operatorname{id})(\rho_m) \right\rVert_{1} \stackrel{(1)}{=}
\sup_{ S \in \mathcal B(\mathcal{H}\otimes \mathbb C^n); \left\lVert  S \right\rVert=1}\operatorname{tr}\left(\rho_m\left(\widehat{\Lambda}^*_t-\operatorname{id}\right)(S)\right)\\
 &\stackrel{(2)}{=} \sup_{S \in \mathcal B(\mathcal{H}\otimes \mathbb C^n); \left\lVert  S \right\rVert=1}\sum_{i=1}^m \lambda_i \left\langle u_i, \left(\widehat{\Lambda}^*_t-\operatorname{id}\right)(S) u_i \right\rangle\\
 & \stackrel{(3)}{\le}   \sup_{S \in \mathcal B(\mathcal{H}\otimes \mathbb C^n); \left\lVert  S \right\rVert=1}\sum_{i=1}^m \lambda_i\left(\left\langle \left( \widehat{P}_t-I\right) u_i,S\widehat{P}_t u_i \right\rangle + \left\langle  u_i,S\left( \widehat{P}_t-I\right)u_i \right\rangle\right) \\
&\quad + \sup_{S \in \mathcal B(\mathcal{H}\otimes \mathbb C^n); \left\lVert S \right\rVert=1} \sum_{i=1}^m \lambda_i  \sum_{l=1}^{\infty} \int_0^t \left\langle \widehat{L_l} \widehat{P}_{t-s} u_i, \widehat{\Lambda}_s^*\left(S\right) \widehat{L_l} \widehat{P}_{t-s} u_i \right\rangle \ ds, 
\end{split}
\end{equation}
where we expressed the norm in a weak formulation in (1), applied the spectral decomposition of $\rho_m$ in (2), and used \eqref{standardform2} to obtain (3).

The two terms in the second-to-last line of \eqref{eq:distrib} satisfy, again by the spectral decomposition of $\rho_m,$
\begin{equation}
\begin{split}
\label{eq:ftwo}
&\sup_{S \in \mathcal B(\mathcal{H}\otimes \mathbb C^n); \left\lVert S \right\rVert=1}\sum_{i=1}^m \lambda_i\left(\left\langle ( \widehat{P}_t-I) u_i,S \widehat{P}_tu_i \right\rangle + \left\langle  u_i,S( \widehat{P}_t-I)u_i \right\rangle\right) \\
&=\sup_{S \in \mathcal B(\mathcal{H}\otimes \mathbb C^n); \left\lVert S \right\rVert=1}\operatorname{tr}\left((\widehat{P}_t^*-I)S\widehat{P}_t  \rho_m \right) + \operatorname{tr}\left( S(\widehat{P}_t-I)\rho_m \right) \\
&\stackrel{(1)}{\le}  \left\lVert \rho_m\left(\widehat{P}_t-I \right)^*  \right\rVert_{1}+ \left\lVert \left(\widehat{P}_t-I \right) \rho_m \right\rVert_{1}  \stackrel{(2)}{=} 2  \left\lVert \left(\widehat{P}_t-I \right) \rho_m \right\rVert_{1}. 
\end{split}
\end{equation} 
Here, we used H\"older's inequality and contractivity of the semigroup $(\widehat{P}_t)$ to get (1) and then used that the trace norm is the same for any operator and its adjoint to conclude (2). 
For the last term in \eqref{eq:distrib} we obtain by contractivity of the QMS
\begin{equation*}
\begin{split}
& \left\lvert  \sum_{i=1}^m \lambda_i \sum_{l=1}^{\infty} \int_0^t \left\langle  \widehat{L_l} \widehat{P}_{t-s} u_i, \widehat{\Lambda}_s^*\left(S \right) \widehat{L_l} \widehat{P}_{t-s} u_i \right\rangle \ ds \right\rvert \le \sum_{i=1}^m \lambda_i   \sum_{l=1}^{\infty} \int_0^t \left\lVert \widehat{L_l} \widehat{P}_{t-s} u_i \right\rVert^2 \ ds
\end{split}
\end{equation*} 
and thus 
\begin{equation}
\begin{split} 
&\sum_{i=1}^m \lambda_i   \sum_{l=1}^{\infty} \int_0^t \left\lVert \widehat{L_l} \widehat{P}_{t-s} u_i \right\rVert^2 \ ds \stackrel{(1)}{=} -2\sum_{i=1}^m \lambda_i \int_0^t \Re \left\langle \widehat{P}_{t-s}u_i, \widehat{G}\widehat{P}_{t-s} u_i \right\rangle  \ ds \\
&\stackrel{(2)}{=} \sum_{i=1}^m \lambda_i \int_0^t \frac{d}{ds} \left\lVert \widehat{P}_{t-s} u_i \right\rVert^2 \ ds \stackrel{(3)}{=}  \operatorname{tr}\left(\left(I- \widehat{P}_t^*\widehat{P}_t\right)\rho_m \right),
\end{split}
\end{equation}
where we used \eqref{eq:conseq} in (1), that $G$ is the generator of $(P_t)$ to obtain (2), and finally the fundamental theorem of calculus to obtain (3). 
We can then rewrite this term by decomposing it as follows
\begin{equation}
\begin{split} 
\label{eq:distrib2} 
  \operatorname{tr}\left(\left(I- \widehat{P}_t^*\widehat{P}_t\right)\rho_m \right) &= \operatorname{tr}\left(\left(I-\widehat{P}_t^*\right)\rho_m\right)+ \operatorname{tr}\left(\widehat{P}_t^*\left(I-\widehat{P}_t\right)\rho_m\right)\\
&\stackrel{(1)}{=} \operatorname{tr}\left(\rho_m\left(I-\widehat{P}_t^*\right)\right)+ \operatorname{tr}\left(\widehat{P}_t^*\left(I-\widehat{P}_t\right)\rho_m\right)\\
&\stackrel{(2)}{\le} 2 \left\lVert \left(\widehat{P}_t-I\right)\rho_m \right\rVert_{1}
\end{split}
\end{equation}
where we used cyclicity of the trace in (1). To obtain (2) we used H\"older's inequality together with the contractivity of the semigroup $\widehat{P}_t^*$ and the fact that the trace norm for operators and their adjoints coincide. Estimating \eqref{eq:distrib} by \eqref{eq:ftwo} and \eqref{eq:distrib2}, we can let $m$ tend to infinity and obtain the bound stated in the lemma.
\end{proof}

We are now able to prove Theorem \ref{theo:theo3} which shows that the uniform continuity for the $\alpha$-ECD norm which we obtained for closed quantum systems in Proposition \ref{prop:tarate} applies to open quantum systems as well: 

\begin{proof}[Proof of Theorem \ref{theo:theo3}]
We start by proving the first part of the theorem:
That $G$ is the generator of a contraction semigroup if $a<1$ follows from \cite[Theorem $2.7$]{EN}.

First, we observe that $ K \otimes I_{\mathbb C^n}$ is still relatively $H \otimes I_{\mathbb C^n}$-bounded with the same bound $a$ \cite[Theorem $7.1.20$]{S15}. 

 According to Lemma \ref{reduction} it suffices to obtain bounds on the rate of convergence for the semigroups $(\widehat{P_t})$ on density operators $\rho \in \mathcal D(\mathcal H \otimes \mathbb C^n)$ with spectral decomposition $\rho = \sum_{i=1}^{\infty} \lambda_i \ket{\varphi_i} \bra{\varphi_i}$
 \begin{equation*}
\begin{split} 
&\left\lVert  \lambda^{\alpha}  G\otimes I_{\mathbb C^n}(\lambda I- G\otimes I_{\mathbb C^n})^{-1}\sqrt{\rho} \sqrt{\rho} \right\rVert^2_{1} \\
&\stackrel{(1)}{\le} \lambda^{2\alpha} \operatorname{tr}\left( \frac{ G\otimes I_{\mathbb C^n}}{\lambda I-  G\otimes I_{\mathbb C^n}}\rho   \frac{G^*\otimes I_{\mathbb C^n}}{\lambda I- G^*\otimes I_{\mathbb C^n}}\right)\\
&\stackrel{(2)}{\le}   \sum_{i=1}^{\infty} \lambda_i \left\lVert \lambda^{\alpha}  \frac{ G\otimes I_{\mathbb C^n}}{\lambda I-  G\otimes I_{\mathbb C^n}}\ket{\varphi_i} \right\rVert^2  \\
&\stackrel{(3)}{\le} \sum_{i=1}^{\infty} \lambda_i   \left(\operatorname{max} \left\{2c^{\alpha},3bc^{\alpha-1}+(1+3a)    \left\lVert \lambda^{\alpha}\frac{-i  H\otimes I_{\mathbb C^n}}{\lambda I-(-i  H\otimes I_{\mathbb C^n})}\ket{\varphi_i} \right\rVert  \right\} \right)^2,
\end{split}
\end{equation*}
where we used H\"older's inequality to get (1), the spectral decomposition of $\rho$ in (2), and Lemma \ref{Favpet} to get (3).
Then, by expanding the expression above (1) and using the Cauchy-Schwarz inequality (2) we find
 \begin{equation*}
\begin{split} 
&\left\lVert  \lambda^{\alpha}  G\otimes I_{\mathbb C^n}(\lambda I- G\otimes I_{\mathbb C^n})^{-1}\sqrt{\rho} \sqrt{\rho} \right\rVert^2_{1}  \\
&\stackrel{(1)}{ \le}  \operatorname{max} \Bigg\{(2c^{\alpha})^2,(3bc^{\alpha-1})^2+6bc^{\alpha-1}(1+3a)\sum_{i=1}^{\infty} \sqrt{\lambda_i}\sqrt{\lambda_i}  \left\lVert \lambda^{\alpha}\tfrac{-i  H\otimes I_{\mathbb C^n}}{\lambda I-(-i  H\otimes I_{\mathbb C^n})}\ket{\varphi_i} \right\rVert\\
&\quad +(1+3a)^2 \lambda^{2\alpha} \operatorname{tr}\left(\tfrac{-i  H\otimes I_{\mathbb C^n}}{\lambda I-(-i  H\otimes I_{\mathbb C^n})}\rho\tfrac{i  H\otimes I_{\mathbb C^n}}{\lambda I-(i  H\otimes I_{\mathbb C^n})}\right)  \Bigg\}\\
&\stackrel{(2)}{ \le} \operatorname{max} \Bigg\{(2c^{\alpha})^2,(3bc^{\alpha-1})^2+6bc^{\alpha-1}(1+3a)\sqrt{ \lambda^{2\alpha} \operatorname{tr}\left(\tfrac{-i  H\otimes I_{\mathbb C^n}}{\lambda I-(-i  H\otimes I_{\mathbb C^n})}\rho\tfrac{i  H\otimes I_{\mathbb C^n}}{\lambda I-(i  H\otimes I_{\mathbb C^n})}\right)}\\
&\quad +(1+3a)^2 \lambda^{2\alpha} \operatorname{tr}\left(\tfrac{-i  H\otimes I_{\mathbb C^n}}{\lambda I-(-i  H\otimes I_{\mathbb C^n})}\rho\tfrac{i  H\otimes I_{\mathbb C^n}}{\lambda I-(i  H\otimes I_{\mathbb C^n})}\right)  \Bigg\}\\
& =  \operatorname{max} \Bigg\{2c^{\alpha},3bc^{\alpha-1}+(1+3a) \sqrt{\lambda^{2\alpha}\operatorname{tr}\left(\tfrac{-i  H\otimes I_{\mathbb C^n}}{\lambda I-(-i  H\otimes I_{\mathbb C^n})}\rho\tfrac{i  H\otimes I_{\mathbb C^n}}{\lambda I-(i  H\otimes I_{\mathbb C^n})}\right)} \Bigg\}^2.
\end{split}
\end{equation*}
Applying \eqref{eq:Hestimate} yields the desired estimate on the semigroup $(\widehat{P_t})$ and Lemma \ref{reduction} the one on $(\widehat{\Lambda_t})$. By \eqref{eq:ltest}, we then conclude that $\left\lVert \Lambda_t- \Lambda_s \right\rVert_{\diamond^{2\alpha}}^{\left\lvert H \right\rvert, E} \le \ \omega_H(\alpha,a,b,c)\ \vert t-s \vert^{\alpha}.$
The second part follows analogously with the only difference being that 
\begin{equation*}
\begin{split}\sum_{i=1}^{\infty} \lambda_i \left\lVert  \lambda^{\alpha} K(\lambda I-K)^{-1}\ket{\varphi_i} \right\rVert^2 
&=\sum_{i=1}^{\infty} \lambda_i  \int_{\mathbb R}  \tfrac{\lambda^{2\alpha} s^2}{(\lambda-s)^2} \ d\left\langle \mathcal E^K(s)\varphi_i,\varphi_i \right\rangle  \\
&\le\sum_{i=1}^{\infty} \lambda_i  (1-\alpha)^{1-\alpha} \alpha^{\alpha} \int_{\mathbb R}  \left\lvert s \right\rvert^{2\alpha} \ d\left\langle \mathcal E^K(s)\varphi_i,\varphi_i \right\rangle  \\
&= (1-\alpha)^{2(1-\alpha)} \alpha^{2\alpha} \operatorname{tr}\left(  \left\lvert   K \right\rvert^{2\alpha} \otimes I_{\mathbb C^n} \rho \right).
\end{split}
\end{equation*}
\end{proof}

\begin{corr}
\label{corr:pur}
For open quantum systems satisfying the assumptions of Theorem \ref{theo:theo3} the change in purity is bounded for states $\rho \in \mathscr D(\mathcal H \otimes \mathbb C^n)$ with $\operatorname{tr}(\left\lvert H \right\rvert^{2\alpha} \rho_{\mathcal H})\le E^{2\alpha}$ (or $\operatorname{tr}(\left\lvert K \right\rvert^{2\alpha} \rho_{\mathcal H})\le E^{2\alpha}$) and any $c>0$ for $\omega_{\bullet}$ as in \eqref{eq:sigma} by
\[ \left\lvert \operatorname{tr}\left((\widehat{\Lambda}_t(\rho))^2-(\widehat{\Lambda}_s(\rho))^2\right) \right\rvert   \le \ 2 \omega_{\bullet}(\alpha,a,b,c,E) \ \vert t- s \vert^{\alpha}.\]
\end{corr}
\begin{proof}
Applying Theorem \ref{theo:theo3} to the following estimate yields the claim
\[ \left\lvert \operatorname{tr}\left((\widehat{\Lambda}_t(\rho))^2-(\widehat{\Lambda}_s(\rho))^2\right) \right\rvert  \le 2\left\lVert \widehat{\Lambda}_t(\rho)-\widehat{\Lambda}_s(\rho) \right\rVert_{1}.\]
\end{proof}

We continue with a discussion of applications of Theorem \ref{theo:theo3}. Let us start by continuing our study of the quantum-limited attenuator and amplifier channels that we started in Example \ref{ex:Attenchannel}:
\begin{ex}[Attenuator and amplifier channel]
\label{attchannel} 
Let $N:=a^*a$ be the number operator and $M:=aa^*$, where $a$ and $a^*$ are the standard creation and annihilation operators.
Since coherent states span the entire space, \eqref{coh} uniquely defines the 
action of an attenuator channel $\Lambda^{\text{att}}_t$ (with time-dependent attenuation parameter $\eta(t):=e^{-t}$) on arbitrary states $\rho$ as follows~\cite[Lemma12]{DTG}
 \[ \Lambda^{\text{att}}_t(\rho) = \sum_{l=0}^{\infty} \tfrac{(1-e^{-t})^l}{l!} e^{-t N/2} a^l \ \rho \ (a^*)^l e^{-t N/2}.\]
The generator \cite[(II.16)]{DTG} of the corresponding QDS $(\Lambda^{\text{att}}_t)$ is then given by
\[ \mathcal L^{\text{att}}(\rho):=\frac{d}{dt}\Big \vert_{t=0} \Lambda^{\text{att}}_t(\rho)= a \rho a^* - \tfrac{1}{2} \left(N\rho + \rho N \right). \]
The QDS generated by
\[ \mathcal L^{\text{amp}}(\rho)=(\mathcal L^{\text{att}}-I)(\rho)=a^* \rho a - \tfrac{1}{2} \left(M \rho + \rho M \right). \]
is denoted as $(\Lambda_t^{\text{amp}})$, where $\Lambda_t^{\text{amp}}$ denotes the so-called quantum-limited amplifier channel.
 
Hence, by Theorem \ref{theo:theo3} with $H=0$ with $a=b=0$ and $K=N$ it follows that
\begin{equation}
\begin{split}
\label{eq:attbound}
\left\lVert \Lambda_t^{\text{att}}- \Lambda_s^{\text{att}} \right\rVert^{N,E}_{\diamond^{2\alpha}} \le 4 \zeta_{\alpha} (1-\alpha)^{1-\alpha} \alpha^{\alpha} E^{\alpha} \ \vert t-s \vert^{\alpha}. 
\end{split}
\end{equation}
At least for $\alpha=1/2$, we can compare the above asymptotics with the explicit bound that was obtained in \cite{N18}: 
Consider attenuation parameters $\eta = 1$, for the initial state, and  $\eta'=e^{-t}$, for the time evolved state, as in \cite{N18}. If we assume for simplicity that the energy $E$ is integer-valued, then the energy-constrained minimum fidelity, that is the infimum of the fidelity over all pure states of expected energy less or equal to $E$ evolved under the attenuator channel with parameters $\eta,\eta'$ respectively, defined in \cite[(11)]{N18}, satisfies $F_E(\eta,\eta')=e^{-tE/2} = 1-Et/2 + \mathcal O(t^2).$ By the Fuchs- van de Graaf inequality as in \eqref{eq:FvG} this yields the short-time asymptotics 
\[ \left\lVert \Lambda_t^{\text{att}}-I \right\rVert^{N,E}_{\diamond^{1}} \le 2\sqrt{Et/2} (1+o(1)), \ \text{ as } t \downarrow 0,\]
which has the same scaling both in time and energy as the above estimate \eqref{eq:attbound}.
In analogy to \eqref{eq:attbound}, we find for the amplifier channel 
\[ \left\lVert \Lambda_t^{\text{amp}}-\Lambda_s^{\text{amp}} \right\rVert^{M,E}_{\diamond^{2\alpha}} \le 4 \zeta_{\alpha}  (1-\alpha)^{1-\alpha} \alpha^{\alpha} E^{\alpha} \ \vert t- s\vert^{\alpha}. \]
Finally, since $M = N+I$ it follows that $\left\lVert M \varphi \right\rVert^2 \le 2 (\left\lVert N \varphi \right\rVert^2+\left\lVert  \varphi \right\rVert^2)$ and thus by Proposition \ref{corr2}
\[ \left\lVert \Lambda_t^{\text{amp}}-\Lambda_s^{\text{amp}} \right\rVert^{N,E}_{\diamond^{2\alpha}} \le 4 \zeta_{\alpha}  (1-\alpha)^{1-\alpha} \alpha^{\alpha} \sqrt{2 E^{2\alpha} +2 }\ \vert t-s \vert^{\alpha}. \]
\end{ex} 

\begin{ex}[Linear quantum Boltzmann equation, \cite{A02,HV2}]
\label{lqb}
Since this example describes scattering effects, that depend on the ratio of mass parameters, we exceptionally include physical constants in this example.
Consider a particle of mass $M$ whose motion without an environment is described by the self-adjoint Schr\"odinger operator $H_0 = -\frac{\hbar^2}{2M}\Delta + V.$ The linear quantum Boltzmann equation describes the motion of the particle in the presence of an additional ideal gas of particles with mass $m$ distributed according to the Maxwell-Boltzmann distribution $\mu_{\beta}(p) = \frac{1}{\pi^{3/2} p_{\beta}^3} e^{- \left\lvert p \right\rvert^2 /p_{\beta}^2}$ where $p_{\beta} = \sqrt{2m/\beta}.$ 

Here, we discuss for simplicity the linear quantum Boltzmann equation under the Born approximation of scattering theory \cite{HV2}: Let $m_{\text{red}} = mM/(m+M)$ be the reduced mass and $n_{\text{gas}}$ the density of gas particles. We assume that the scattering potential between the gas particles and the single particle is of short-range and smooth such that $V \in \mathscr S(\mathbb R^3)$ where $\mathscr S(\mathbb R^3)$ is the Schwartz space \cite{RS1}. In the Born approximation the scattering amplitude becomes $f(p) = -\frac{m_{\text{red}}}{2\pi \hbar^2} \mathcal F(V)(p/\hbar),$ where $\mathcal F$ is the Fourier transform.  

The presence of the ideal gas leads then to a constant energy shift $H_{\text{per}} = -2\pi \hbar^2 \tfrac{n_{\text{gas}}}{m_{\text{red}}} \Re(f(0))$ in the Hamiltonian $H = H_0 + H_{\text{per}}$ and also to an additional dissipative part \cite{HV}: Let $P = -i \hbar \nabla_x$ be the momentum operator, then we introduce operators
\begin{equation}
L(P,k) = \sqrt{\sqrt{\frac{\beta m}{2\pi}} \frac{n_{\text{gas}}}{m_{\text{red}} \left\lvert k \right\rvert}} f(-k) \exp\left(-\beta \frac{\left((1+\tfrac{m}{M})\left\lvert k \right\rvert^2+2\tfrac{m}{M} \langle P,k \rangle \right)^2 }{16m \left\lvert k \right\rvert^2} \right)
\end{equation}
where $\left\lVert \exp\left(-\beta \frac{\left((1+\tfrac{m}{M})\left\lvert k \right\rvert^2+2\tfrac{m}{M} \langle P,k \rangle \right)^2 }{16m \left\lvert k \right\rvert^2} \right) \right\rVert \le 1$ by the functional calculus.
The linear quantum Boltzmann equation for the state $\rho$ of the particle then reads
\[ \frac{d}{dt} \rho(t) = -i [H,\rho(t)] + \int_{\mathbb R^3} \left(e^{i \langle k, x \rangle}L(P,k)\rho L(P,k)^*e^{-i \langle k, x \rangle} - \frac{1}{2}\{ \rho,L(P,k)^*L(P,k) \}  \right) \ dk.\]
Lemma \ref{Favpet} and Proposition \ref{corr2} imply, since $H_{\text{per}}$ is a bounded perturbation and 
\[ \int_{\mathbb R^3} \left\lVert L(P,k)^*L(P,k) \right\rVert \ dk < \infty, \]
 that the dynamics of the linear quantum Boltzmann equation obeys the same asymptotics as the dynamics of a closed system evolving according to $\frac{d}{dt} \rho(t)= -i [H_0,\rho(t)].$
Thus, the QDS $(\Lambda_t)$ of the linear quantum Boltzmann equation satisfies for $E> \inf(\sigma(\vert H_0 \vert))$
\[ \left\lVert \Lambda_t-\Lambda_s \right\rVert^{\vert H_0\vert,E}_{\diamond^{2\alpha}}  = \mathcal O\left(E^{\alpha}\vert t- s \vert^{\alpha} \right). \]
\end{ex}
\hfill \break
\hfill \break
By combining the attenuator channel with the amplifier channel, and using an operator proportional to the number operator $N$ as the Hamiltonian part, we obtain the example of a damped and pumped harmonic oscillator which found, for example, applications in quantum optics, to describe a single mode of radiation in a cavity \cite{A02}:
\begin{ex}[Harmonic oscillator, \cite{A02a}]
We consider a scaled number operator as the Hamiltonian $H= \zeta a^*a$ for some $\zeta>0$ and damping $V(\rho):= \gamma_{\downarrow} a \rho a^*$ and pumping $W(\rho):=\gamma_{\uparrow} a^*\rho a$ operators and transition rates $\gamma_{\downarrow},\gamma_{\uparrow}\ge 0.$ The damping and pumping processes are described by Lindblad operators $L_{\downarrow}:=\sqrt{\gamma_{\downarrow}} a $ and $L_{\uparrow}:=\sqrt{\gamma_{\uparrow}} a^*.$
The operator $K=-\frac{1}{2}\left(L_{\downarrow}^*L_{\downarrow} + L_{\uparrow}^*L_{\uparrow}\right)$ is then dissipative and self-adjoint, such that Theorem \ref{theo:theo3} applies, and implies that the QDS $(\Lambda_t)$ satisfies for any $E>0$
\[ \left\lVert \Lambda_t-\Lambda_s \right\rVert^{H,E}_{\diamond^{2\alpha}}  = \mathcal O\left( E^{\alpha}\vert t- s \vert^{\alpha} \right). \]
\end{ex}
Next, we study the evolution of quantum particles under Brownian motion which is obtained as the diffusive limit of the quantum Boltzmann equation that we discussed in Example \ref{lqb} \cite[Section $5$]{HV2}.

\begin{ex}[Quantum Brownian motion,\cite{AS,V04}]
Consider the Hamiltonian of a harmonic oscillator $H=-\frac{d^2}{dx^2} +x^2$ and Lindblad operators for $j \in \left\{1,2\right\}$ given by $L_j:= \gamma_j x +  \beta_j \frac{d}{dx} $ where $\gamma_j, \beta_j \in \mathbb C.$ In particular, choosing $\gamma_j=\beta_j $ turns $L_j$ into the \emph{annihilation operator} $L_j = \gamma_j \left(\frac{d}{dx}+x\right)$ and $L^*$ into the \emph{creation operator}  $L_j^* = \gamma_j \left(-\frac{d}{dx}+x\right)$ which have been considered in the previous example. 

The Lindblad equation for quantum Brownian motion reads  
\begin{equation}
\begin{split}
\partial_t \rho 
&= -i[H,\rho]+\tfrac{i\lambda}{2}\left([p,\left\{x,\rho\right\}]-[x,\left\{p,\rho\right\}]\right) - D_{pp}[x,[x,\rho]]-D_{xx} [p,[p,\rho]]\\
&\quad +D_{xp}[p,[x,\rho]]+D_{px}[x,[p,\rho]]
\end{split}
\end{equation}
with diffusion coefficients $D_{xx} = \tfrac{\left\lvert \gamma_1 \right\rvert^2+\left\lvert \gamma_2 \right\rvert^2}{2}, D_{pp} = \tfrac{\left\lvert \beta_1 \right\rvert^2+\left\lvert \beta_2 \right\rvert^2}{2}, D_{xp}=D_{px}=-\Re \tfrac{\gamma_1^*\beta_1+\gamma_2^*\beta_2}{2}$
and $\lambda= \operatorname{Im}\left(\gamma_1^*\beta_1+ \gamma_2^*\beta_2 \right).$

The operator $K=-\tfrac{1}{2} \sum_{j=1}^2 L_j^*L_j$ is then relatively $H$-bounded and $G=iH-K$ is the generator of a contraction semigroup on $D(H)$. By Theorem \ref{theo:theo3}, the QDS $(\Lambda_t)$ of quantum Brownian motion satisfies for $E>\inf(\sigma(H))$ and $\alpha \in (0,1]$
\[ \left\lVert \Lambda_t-\Lambda_s \right\rVert^{H,E}_{\diamond^{2\alpha}}  = \mathcal O\left( E^{\alpha}\vert t- s \vert^{\alpha} \right).\]
\end{ex}
The field of quantum optics is a rich source of open quantum systems to which the convergence Theorem \ref{theo:theo3} applies and we discuss a few of them in the following example:
\begin{ex}[Quantum optics/Jaynes-Cummings model \cite{CGQ03}]
Systems that consist of a harmonic oscillator coupled to two-level systems are among the common illustrative examples considered in quantum optics and within this theory are called \emph{Jaynes-Cummings models}. 
A particular example of a Jaynes-Cummings model is a two-level ion coupled to a harmonic trap of strength $\nu>0$ located at the node of a standing light wave. For a detuning parameter $\Delta$ and \emph{Rabi} frequency $\Omega$, a Master equation with Hamiltonian
\[H =I_{\mathbb C^2} \nu a^*a + \frac{\Delta}{2} \sigma_z - \frac{\Omega}{2} \left( \sigma_+ + \sigma_- \right) \sin\left(\eta (a+a^*) \right),\]
where $\eta$ is the \emph{Lamb-Dicke} parameter, and with Lindblad operators $L= \sqrt{\Gamma} \sigma_{-}, L^*= \sqrt{\Gamma}\sigma_{+}$ has been proposed in \cite{CBPZ} for this model.  Here, $\Gamma$ is the decay rate of the excited state of the two-state ion.
The Hilbert space is therefore $\ell^2(\mathbb N) \otimes \mathbb C^2$ and as the Lindblad operators are just bounded operators, all conditions of Theorem \ref{theo:theo3} are trivially satisfied. Thus, it follows that for $E>0$ the QDS $(\Lambda_t)$ satisfies 
\[ \left\lVert \Lambda_t-\Lambda_s \right\rVert^{\nu a^*a,E}_{\diamond^{2\alpha}}  = \mathcal O\left( E^{\alpha}\vert t- s \vert^{\alpha} \right).\]
More generally, plenty of models in quantum optics are special cases of the following form \cite{CGQ03}:
Consider Hamiltonians $H$ with $h_j \in \mathbb C^{M \times M}$ 
\[ H =  \left(h_j \prod_{k=1}^N (a_k^*)^{n_k} (a_k)^{m_k} + \operatorname{H.a.} \right)\]
acting on a Hilbert space $\mathcal{H} = \ell^2(\mathbb N)^{\otimes N} \otimes \mathbb C^M.$ The Lindblad operators are of the form $L_k= \lambda_k a_k$ or $L_k = \lambda_k a_k^*$ \footnote{For notational simplicity, we suppress the tensor products with the identity on all other factors.} where $a_k$ is the annihilation operator of the $k$-th factor of the tensor product $\ell^2(\mathbb N)^{\otimes N}$ and $\lambda_k \ge 0$ a positive semi-definite matrix acting on $\mathbb C^M.$

Hence, the operators $-\frac{1}{2} L_k^*L_k$ are self-adjoint and dissipative and for a large class of Hamiltonians $H$ the asymptotics of Theorem \ref{theo:theo3} can be applied.
\end{ex}

\section{Generalized relative entropies and quantum speed limits}
\label{QSL}
We start with some immediate consequences of Proposition \ref{prop:schroe}, Proposition \ref{prop:tarate}, and Theorem \ref{theo:theo3} on certain generalized relative entropies and distance measures which are dominated by the trace norm:
\begin{defi} 
For $\alpha \in (0,1) \cup (1,\infty)$, $\alpha$-\emph{Tsallis} and $\alpha$-\emph{R\'enyi divergences} (see \emph{e.g.} \cite{NN11}) are respectively defined as follows for $\rho,�\sigma \in \mathscr D(\mathcal H)$ with $\supp(\rho) \subseteq \supp(\sigma)$
\begin{equation}
\begin{split}
D^{\text{Tsallis}}_{\alpha}(\rho \vert\vert \sigma) &= \tfrac{1}{\alpha-1} \left(   \operatorname{tr} \left(\rho^{\alpha} \sigma^{1-\alpha} \right) -1 \right) \text{ and } 
 D^{\text{R\'enyi}}_{\alpha}(\rho \vert\vert \sigma) = \tfrac{1}{\alpha-1} \log\left(\operatorname{tr}\left( \rho^{\alpha} \sigma^{1-\alpha}\right)\right).
\end{split}
\end{equation}

Of particular interest to us are the $\alpha=1/2$-divergences: The $\alpha=1/2$-Tsallis divergence is, up to a prefactor, the square of the \emph{Hellinger distance} \cite{RSI} and satisfies
\[ D^{\text{Tsallis}}_{1/2}(\rho \vert \vert \sigma) = \left\lVert \sqrt{\rho}-\sqrt{\sigma} \right\rVert^2_{2} = 2 \left( 1- \operatorname{tr}\left( \sqrt{\rho}\sqrt{\sigma}\right) \right).\]

The form $A(\rho,\sigma):=\operatorname{tr}\left( \sqrt{\rho}\sqrt{\sigma}\right)$ appearing in $ D^{\text{Tsallis}}_{1/2}$ is known as the \emph{Bhattacharrya coefficient}; it links $ D^{\text{Tsallis}}_{1/2}$ to $D^{\text{R\'enyi}}_{1/2}$:
\[ D^{\text{R\'enyi}}_{1/2}(\rho \vert\vert \sigma) = -2 \log A(\rho,\sigma) =-2 \log \left(1- \tfrac{D^{\text{Tsallis}}_{1/2}(\rho \vert \vert \sigma) }{2} \right).\]
Consider also the fidelity of two states $\rho,\sigma$ that we denote by 
\begin{equation}
\label{eq:fid}
F(\rho,\sigma):=\operatorname{tr}\sqrt{ \sqrt{\rho} \sigma  \sqrt{\rho}} = \left\lVert \sqrt{\rho} \sqrt{\sigma} \right\rVert_1.
\end{equation}
It is related to the {\em{trace distance}} via the Fuchs-van de Graaf inequalities  \cite{FG}:
\begin{equation}
\label{eq:FvG}
2(1-F(\rho,\sigma)) \le \left\lVert \rho-\sigma \right\rVert_1 \le 2\sqrt{1-F(\rho,\sigma)^2}. 
\end{equation}
\emph{Bures angle} $\theta$ and \emph{Bures distance} $d_B$ are respectively defined as 
\[\theta(\rho,\sigma) :=  \arccos \left( F(\rho, \sigma) \right)\text{ and }d_{B}(\rho,\sigma):=\sqrt{2\left(1- F(\rho,\sigma)\right)}.\] 
\end{defi}
As a Corollary of Proposition \ref{prop:tarate} for closed quantum systems and Theorem \ref{theo:theo3} for open quantum systems, we obtain:
\begin{corr}
\label{eq:corrangle}
For closed quantum systems and states $\rho \in \mathscr D(\mathcal{H})$ such that $\operatorname{tr}\left(\rho \left\lvert H \right\rvert^{2 \alpha} \right) \le E^{2\alpha}$ it follows with the notation introduced in Definition \ref{defsemigroup} that:
\begin{itemize}
\item The Bures distance and Bures angle satisfy 
\[ d_{B}(T^{\operatorname{vN}}_t(\rho),T^{\operatorname{vN}}_s(\rho)) \le \sqrt{2 g_{\alpha} E^{\alpha}\ \vert t- s \vert^{\alpha}} \text{ and }\]
\[ \theta(T^{\operatorname{vN}}_t(\rho),T^{\operatorname{vN}}_s(\rho)) \le \arccos \left(\operatorname{max}\left\{1- g_{\alpha} E^{\alpha}\ \vert t- s \vert^{\alpha},-1 \right\}\right).\]
\item For the $1/2$-divergences we obtain 
\begin{equation*}
\begin{split}
&D^{\text{Tsallis}}_{1/2}(T^{\operatorname{vN}}_t(\rho) \vert \vert T^{\operatorname{vN}}_s(\rho)) \le 2 g_{\alpha}E^{\alpha}\ \vert t- s \vert^{\alpha}\text{ and } \\
&D^{\text{R\'enyi}}_{1/2}(T^{\operatorname{vN}}_t(\rho) \vert \vert T^{\operatorname{vN}}_s(\rho)) \le -2\log \left(\left(1- g_{\alpha}E^{\alpha}\ \vert t- s \vert^{\alpha}\right)_{+}\right)
\end{split}
\end{equation*}
\end{itemize}
where $(a)_{+}:=\max \{a,0\}.$
For open quantum systems satisfying the conditions of Theorem \ref{theo:theo3} and states $\rho$ satisfying $\operatorname{tr}\left(\rho \left\lvert H \right\rvert^{2 \alpha} \right) \le E^{2\alpha}$(or $\operatorname{tr}\left(\rho \left\lvert K \right\rvert^{2 \alpha} \right) \le E^{2\alpha}$)
we obtain for $\omega_{\bullet}$ as in \eqref{eq:sigma}
\begin{itemize}
\item For the $1/2$-divergences it follows that 
\begin{equation*}
\begin{split}
&D^{\text{Tsallis}}_{1/2}(\Lambda_t(\rho) \vert \vert \Lambda_s(\rho))  \le \omega_{\bullet} \vert t-s \vert^{\alpha}\text{ and } \\
&D^{\text{R\'enyi}}_{1/2}(\Lambda_t(\rho) \vert \vert \Lambda_s(\rho)) \le -2\log \left(\left(1-\tfrac{\omega_{\bullet}}{2}\vert t-s \vert^{\alpha}\right)_{+}\right).
\end{split}
\end{equation*}

\item For the Bures distance and Bures angle, we obtain
\[d_B(\Lambda_t(\rho), \Lambda_s(\rho))  \le \ \sqrt{\omega_{\bullet} \vert t-s \vert^{\alpha}}\text{ and }\]
\[\theta(\Lambda_t(\rho),\Lambda_s(\rho)) \le \arccos \left(\operatorname{max}\left\{1-\tfrac{\omega_{\bullet}}{2}\vert t-s \vert^{\alpha},-1 \right\}\right).  \]
\end{itemize}
\end{corr}
\begin{proof}
It suffices to show that all quantities can be estimated by the trace norm. Proposition \ref{prop:tarate} then provides the upper bounds for closed systems and Theorem \ref{theo:theo3} yields the bounds for open systems.
For estimates on Bures distances and Bures angles an application of the Fuchs-van de Graaf inequality \cite{FG}, \eqref{eq:FvG}, shows that $d_{B}(\rho,\sigma)^2 \le  \left\lVert  \rho - \sigma \right\rVert_{1}$ and  $\theta(\rho,\sigma) \le \arccos \left( 1- \tfrac{\left\lVert \rho- \sigma \right\rVert_1}{2} \right).$
The Powers-St\o rmer inequality \cite[Lemma $4.1$]{PS} implies that $D^{\text{Tsallis}}_{1/2}(\rho \vert \vert \sigma) \le   \left\lVert \rho-\sigma \right\rVert_{1}.$ \end{proof}
The study of quantum speed limits, see also the review article \cite{DC}, is concerned with the minimal time for the system needed to evolve from one state of expected energy $E$ to another state that is a certain \emph{distance} away from the initial state.  
It has been shown in \cite{ML,LT} that the minimal time of a closed quantum system to evolve from an initial state $\ket{\varphi_0}$ to
another state that is orthogonal to it, under the evolution given by the Schr\"odinger equation \eqref{eq:Schroe}, with positive semi-definite Hamiltonian $H$, satisfies
\begin{equation}
\begin{split}
\label{eq:tmin1}
t_{\text{min}} \ge  \tfrac{\pi}{2} \operatorname{max}\left\{ \frac{1}{\Braket{\varphi_0 |H|\varphi_0}},\frac{1}{\sqrt{\Braket{\varphi_0 |H^2|\varphi_0}-\Braket{\varphi_0 |H|\varphi_0}^2}}\right\},
\end{split}
\end{equation} 
and showed that this bound can be saturated. For arbitrary (Bures) angles $\theta \in [0,\pi/2]$ this bound was (partially numerically) extended by Giovannetti, Lloyd, and Maccone \cite{GLM,GLM2,GLM3} to 
\begin{equation}
\begin{split}
\label{eq:tmin2}
 t_{\text{min}} \ge \operatorname{max} \left\{\frac{1}{\sqrt{\Braket{\varphi_0 |H^2|\varphi_0}-\Braket{\varphi_0 |H|\varphi_0}^2}} \theta, \frac{2}{\pi E} \theta^2\right\}.
 \end{split}
\end{equation} 
While the quantum speed limits for closed quantum system still yield non-trivial statements for dynamics generated by unbounded operators, non-trivial estimates for open quantum systems with unbounded operators do not seem to exist.
Let us begin by mentioning some results that hold for open quantum systems with bounded generators. 
In \cite{CEPH} and \cite{UK} a bound on the purity has been stated saying that to reach a purity $p_{\text{fin}}:=\operatorname{tr}(\rho(t)^2)$ from a purity $p_{\text{start}}:=\operatorname{tr}(\rho(0)^2)$ the minimal time needed is bounded from below by 
\begin{equation}
\begin{split}
\label{eq:purity}
 t_{\text{min}} \ge \operatorname{max} \left\{ \frac{\vert \log(p_{\text{fin}})- \log(p_{\text{start}}) \vert}{ 4 \sum_{k} \left\lVert L_k \right\rVert_{2}},\frac{\vert \log(p_{\text{fin}})- \log(p_{\text{start}}) \vert}{\left\lVert \mathcal L-\mathcal L^* \right\rVert} \right\}, 
  \end{split}
\end{equation} 
where $L_k$ are the Lindblad operators, and $\mathcal L$ is the generator of the associated QDS.
Furthermore, a bound on the quantum speed limit in terms of the operator norm of the generator has been derived in \cite{DL}. 
In the following remark we see that all these bounds have a pathological behaviour for certain infinite-dimensional systems and cannot be sharp in general:
\begin{rem}
\label{rem:fail}
Consider a closed system with Hamiltonian $S=-\frac{d^2}{dx^2}$ on $\mathbb R$. 
The state $\psi \in L^2(\mathbb R)$ with Fourier transform $\mathcal F(\psi)(x) = \frac{c}{(1+x^2)^{1/2}}$ where $c>0$ is such that $\psi$ is of unit norm. Then, $\langle S \psi, \psi \rangle = \infty$ whereas $\operatorname{tr}(S^{\alpha}\rho)<\infty$ for $\alpha<1/4$.   
Thus, the above bounds \eqref{eq:tmin1} and \eqref{eq:tmin2} reduce to the trivial bound $t_{\text{min}} \ge 0.$

For infinite-dimensional open quantum systems, the first term in the bound on the purity \eqref{eq:purity} reduces to zero if the Lindblad operators are not Hilbert-Schmidt, which is the case for all examples presented in Section \ref{OQS}. In particular, if the Lindblad operators are unbounded, then the bound simplifies to $ t_{\text{min}} \ge 0.$
\end{rem}

We can now state the proof of Theorem \ref{theo:speedlimit}:
\begin{proof}[Proof of Theorem \ref{theo:speedlimit}]
The first estimate on the minimal time of the Schr\"odinger dynamics, follows from the polarization identity of the Hilbert space inner product 
\begin{equation*}
\left\lVert (T_t^S-I)x \right\rVert^2 = 2-2 \Re \langle T_t^Sx,x\rangle \le g_{\alpha}^2 E^{2\alpha} t^{2\alpha},
\end{equation*}
and Proposition \ref{prop:schroe}, which after rearranging yields the claim.
For the estimates on the Bures angle we rearrange the estimates in Corollary \ref{eq:corrangle}, and for the estimate on the purity we rearrange the estimate in Corollary \ref{corr:pur}.
\end{proof}

\section{Entropy and capacity bounds}
\label{ECB}
In this section, we obtain explicit continuity bounds for different families of entropies of quantum states, and various constrained classical capacities of quantum channels in infinite dimensions.

The capacity of a channel is the maximal rate at which information can be transmitted through it reliably. Unlike a classical channel, a quantum channel has various different capacities. These depend, for example, on the nature of the information transmitted (classical or quantum), the nature of the input states (product or entangled), the nature of the allowed measurements at the output of the channel (individual or collective), the availability of any auxiliary resource (\emph{e.g.}~prior shared entanglement between the sender and the receiver), the presence or absence of feedback from the receiver to the sender, \emph{etc.}. Moreover, if the quantum channel is infinite-dimensional, then one needs to impose an energy constraint on the inputs to the channel (to ensure that the capacities are finite). Due to the energy constraint, the resulting capacity is called the {\em{constrained capacity}} of the channel. Here we consider three different constrained capacities for transmission of {\em{classical information}} through an infinite-dimensional quantum channel: $(i)$ the {\em{constrained product-state capacity}}, which is the capacity evaluated under the additional constraint that the inputs are product states, $(ii)$ the {\em{constrained classical capacity}}, for which the only constraint is the energy constraint, and $(iii)$ the {\em{constrained entanglement-assisted classical capacity}}, which corresponds to the case in which the sender and the receiver have prior shared entanglement\footnote{To simplify the nomenclature, we henceforth suppress the word {\em{constrained}} when referring to the different capacities.}. If $\Phi: \mathcal T_1({\mathcal H}_A) \rightarrow \mathcal T_1({\mathcal H}_B)$
denotes an infinite-dimensional quantum channel, then the energy constraint on an input state $\rho$ to the channel is given by 
$\tr(H_A\rho) \leq E$, where $H_A$ is the Hamiltonian of the input system $A$\footnote{Since our continuity bounds on the capacities are refinements of those obtained by Shirokov in \cite{S18}, we closely follow the notations and definitions of \cite{S18}.}. For $n$ identical copies of the channel, the energy constraint is $\tr(H_{A^n}\rho^{A^n}) \leq nE$, where $\rho^{A^n}\in \mathscr D({\mathcal {H}}_A^{\otimes n})$ and
$$H_{A^n}= H_A \otimes I^{\otimes n-1} + I \otimes H_A \otimes I^{\otimes n-2} + \ldots + I^{\otimes n-1} \otimes H_A.$$
The capacities are evaluated in the asymptotic limit ($n \to \infty$). For their operational definitions see \cite{H03}. Obviously these capacities depend not only on the channel, $\Phi$, but also on $H_A$ and $E$. We denote the three different classical capacities introduced above as follows: $(i)$ $ C^{(1)}(\Phi, H_A, E)$, $(ii)$ $ C(\Phi, H_A, E)$,  and $(iii)$ $ C_{ea}(\Phi, H_A, E)$. Expressions for these capacities have been evaluated \cite{H03} and are given by equations (\ref{eq:Hol-cap}), (\ref{eq:classcap}) and (\ref{eq:eac}), respectively.

Besides classical capacities, we also study convergence of entropies in this section. The quantum relative entropy for states $\rho,\sigma \in \mathscr D(\mathcal H)$ such that ${\rm supp}(\rho) \subseteq {\rm supp}(\sigma)$, is defined as
\begin{equation}
\label{eq:relent}
D(\rho\vert\vert \sigma) = \operatorname{tr}(\rho (\log(\rho)-\log(\sigma))),
\end{equation}
and the conditional entropy of a bipartite state $\rho_{AB} \in \mathscr D(\mathcal{H}_A \otimes \mathcal H_B)$ is given by 
\begin{equation}
\label{eq:condent}
H(A\vert B)_{\rho}:=S(\rho_{AB})-S(\rho_{B}).
\end{equation}
If the underlying Hilbert space is infinite-dimensional, the von Neumann entropy depends discontinuously on the states and is even unbounded in every neighbourhood: More precisely, let $\varepsilon>0$, then in the $\varepsilon$-neighbourhood (in trace distance) of any state $\rho$, there is another state $\rho'$ (say) for which $S(\rho')=\infty$ \cite{We78}. In general the von Neumann entropy is only lower semicontinuous \emph{i.e.}\ given a state $\rho$, if $(\rho_n)_{n \in \mathbb{N}}$ denotes a sequence of states such that $\left\lVert \rho_n-\rho \right\rVert_{1} \xrightarrow[n \rightarrow \infty]{} 0$, then $S(\rho) \le \liminf_n S(\rho_n)$ \cite{We78}. Although, this explains why there are no continuity bounds for the entropy of states in infinite dimensions, the following observation shows that under additional assumptions, such continuity estimates can indeed be derived: Let $H$ be a self-adjoint operator such that a Gibbs state $\gamma(\beta):= \frac{e^{-\beta H}}{\operatorname{tr} \left(e^{-\beta H} \right)} \in \mathscr D(\mathcal H)$ is well-defined for all $\beta>0$ \footnote{A sufficient condition for $H\ge 0$ to define a Gibbs state is that the resolvent of $H$ is a Hilbert-Schmidt operator.}, the sequence of states $(\rho_n)$ converge in trace norm $\left\lVert \rho_n-\rho \right\rVert_{1} \rightarrow 0$, and the energies $\operatorname{tr}(\rho_n H ), \operatorname{tr}(\rho H )$ are uniformly bounded, then the entropies converge $S(\rho) = \lim_{n \rightarrow \infty} S(\rho_n)$ as well \cite{We78}. Thus, continuity bounds on the von Neumann entropy can be expected to hold for energy-constrained states when the underlying Hamiltonian defines a Gibbs state for all inverse temperatures. Indeed, in \cite{W15} for entropies and \cite{S18} for capacities, such continuity estimates have been established which are fully explicit up to the asymptotic behaviour of the Gibbs state for high energies. It is precisely this asymptotic behaviour that we discuss in this section.

We now want to compare the delicate continuity properties of the von Neumann entropy with the properties of the Tsallis-$(T_q)$ and R\'enyi-$(S_q)$ entropies:
\begin{defi}
The q-\emph{Tsallis entropy} is for $q>1$, using the $q$-Schatten norm, defined by   
\[ T_{q}(\rho) := \tfrac{1}{q-1} \left(1 - \left\lVert \rho \right\rVert_{q}^q \right).\]
The q-\emph{R\'enyi entropy} is for $q>1$, using the $q$-Schatten norm, defined by 
\[ S_{q}(\rho) := \tfrac{1}{q-1} \log\left(\left\lVert \rho \right\rVert_{q}^q \right) =\tfrac{q}{q-1} \log\left(\left\lVert \rho \right\rVert_{q} \right) .\]
\end{defi} 
Unlike the von Neumann entropy, our next Proposition shows that the Tsallis and R\'enyi entropies are Lipschitz continuous, without any assumptions on the expected energy of the state or the Hamiltonian:
\begin{prop}
Let $\rho,\sigma \in \mathscr D(\mathcal H)$ be two states and $\alpha \in (0,1]$. Then the q-Tsallis entropy satisfies the global Lipschitz estimates
\[ \vert T_q(\rho)-T_q(\sigma) \vert \le \tfrac{q}{q-1} \Vert \rho-\sigma \Vert_q \le \tfrac{q}{q-1} \Vert \rho-\sigma \Vert_1 .\]
Assume now that there is additionally some $\delta>0$ such that $\left\lVert \rho \right\rVert_{q} \ge \delta>0$ and $\left\lVert \rho-\sigma \right\rVert_{q} \le \varepsilon<\delta.$ Then the q-R\'enyi entropy satisfies the local Lipschitz condition 
\[  \vert S_{q}(\rho)-S_{q}(\sigma) \vert \le \tfrac{q}{(q-1)(\delta-\varepsilon)}  \Vert \rho-\sigma \Vert_q \le \tfrac{q}{(q-1)(\delta-\varepsilon)}  \Vert \rho-\sigma \Vert_1 .\] 
In particular, under the assumptions of Theorem \ref{theo:theo3}, it follows that for states $\rho$ with $\operatorname{tr} \left(\vert H \vert^{2\alpha} \rho\right) \le E^{2\alpha}$ (or $\operatorname{tr} \left(\vert K \vert^{2\alpha} \rho\right) \le E^{2\alpha}$)
 we obtain for any $c>0$ and $t,s>0$ for the QDS $(\Lambda_t)$ of an open quantum system with $\omega_{\bullet}$ as in \eqref{eq:sigma}
\[ \vert T_q(\Lambda_t\rho)-T_q(\Lambda_s \rho) \vert \le \tfrac{q}{q-1} \omega_{\bullet}  \ \vert t- s \vert^{\alpha} .\]
If the initial state satisfies additionally $\left\lVert \rho \right\rVert_{q} \ge \delta>0$ then up to sufficiently short times $t<\delta/\omega_{\bullet}$  
\[  \vert S_{q}(\Lambda_t(\rho))-S_{q}(\rho) \vert \le \tfrac{q\omega_{\bullet}}{(q-1)(\delta-t\omega_{\bullet})} t^{\alpha} .\] 
\end{prop}
\begin{proof}
The statement on the Tsallis entropy follows directly from
\begin{equation*}
\begin{split}
\left\lvert T_q(\rho)-T_q(\sigma) \right\rvert = \tfrac{1}{q-1} \left\lvert \left\lVert \rho \right\rVert_{q}^q - \left\lVert \sigma \right\rVert_{q}^q \right\rvert \stackrel{(1)}{\le}  \tfrac{q}{q-1} \left\lVert  \rho - \sigma       \right\rVert_{q} \stackrel{(2)}{\le}  \tfrac{q}{q-1} \left\lVert  \rho - \sigma \right\rVert_{1},
\end{split}
\end{equation*}
where we used the mean-value theorem for the function $f(\rho) = \left\lVert \rho \right\rVert_q^q$ on states for which $\left\lVert \rho \right\rVert_q \le 1$ and the inverse triangle inequality in (1), and $\left\lVert \rho \right\rVert_{q} \le \left\lVert \rho \right\rVert_{1} $ in (2). \\
The additional assumptions for the R\'enyi entropy imply that $\left\lVert \sigma \right\rVert_q \ge  \delta - \left\lVert \rho-\sigma \right\rVert_{q} \ge \delta-\varepsilon>0$ which we need for the local Lipschitz condition on the logarithm. Proceeding as for the Tsallis entropy this shows
\begin{equation*}
\begin{split}
\left\lvert S_{q}(\rho)-S_{q}(\sigma) \right\rvert \le  \tfrac{q}{(q-1)(\delta-\varepsilon)} \left\lvert \left\lVert \rho \right\rVert_{q} - \left\lVert \sigma \right\rVert_{q} \right\rvert \le  \tfrac{q}{(q-1)(\delta-\varepsilon)}\left\lVert \rho-\sigma \right\rVert_q\le  \tfrac{q}{(q-1)(\delta-\varepsilon)}\left\lVert \rho-\sigma \right\rVert_1.
\end{split}
\end{equation*}
\end{proof}
It is well-known that the Gibbs state $\gamma(\beta):= \frac{e^{-\beta H}}{\operatorname{tr} \left(e^{-\beta H} \right)}$ maximizes the von Neumann entropy among all states $\rho$ that satisfy $\operatorname{tr} \left( \rho H \right) \le E$ with $E>\inf(\sigma( H  ))$. The inverse temperature $\beta(E)$ entering the Gibbs state is given as the unique solution to 
\begin{equation}
\label{eq:beta}
 \operatorname{tr}\left(e^{-\beta(E) H} (H-E)\right)=0.
 \end{equation}
In our next remark we state the equivalence of high temperatures and high energies in the defining equation \eqref{eq:beta} of the Gibbs state:
\begin{rem}
\label{rem:highe}
By splitting up the terms in low energy and high energy regimes we find
\begin{equation*}
\begin{split}
0=\operatorname{tr}\left(e^{-\beta(E)H} (H-E)\right)=\underbrace{ \sum_{\lambda \in \sigma(H); \lambda \le E} e^{-\beta (E)\lambda} (\lambda-E)}_{=:(1)}  +  \underbrace{\sum_{\lambda \in \sigma(H); \lambda > E} e^{-\beta (E)\lambda} (\lambda-E)}_{=:(2)}  
\end{split}
\end{equation*}
For any finite energy, the term $(1)$ is a finite sum, while $(2)$ is an infinite sum (since the operator $H$ is unbounded). Thus, if the energy would remain finite, as $\beta(E) \downarrow 0$, then $(1)$ is finite whereas $(2)$ becomes infinite. Conversely, if the temperature would remain finite  ($\beta>0$) as $E \rightarrow \infty,$ then $(1)$ tends to negative infinity while $(2)$ vanishes by the dominated convergence theorem. 
\end{rem}

A straightforward calculation shows that the entropy of the Gibbs state satisfies \cite[p.7]{W17}
\begin{equation}
\label{eq:entgibbs}
S(\gamma(\beta(E)))=  \log\left(\operatorname{tr}\left(e^{-\beta(E)  H }\right)\right) + \beta(E)E. 
\end{equation}
 In the proof of \cite[Prop.1]{S06} it is shown that $\lim_{\varepsilon \downarrow 0} \varepsilon   S(\gamma(\beta(E/\varepsilon))) =0.$ In the following, we want to derive precise asymptotics of \eqref{eq:entgibbs} in the high energy limit and discuss applications of it. 

Before entering the general theory, let us study the fully explicit case of the harmonic oscillator first:
\begin{ex}[Harmonic oscillator]
\label{harmosc}
Let $H_{\text{osc}}=a^*a+\tfrac{1}{2}$ be the Hamiltonian of the Harmonic oscillator and $\sigma(H_{\text{osc}}):=\left\{n+1/2; n \in \mathbb N_0 \right\}$ its spectrum. Then the solution $\beta(E)$ of the equation
\[ \operatorname{tr}\left(e^{-\beta(E) H_{\text{osc}}} (H_{\text{osc}}-E)\right)=0 \text{ for } E>1/2\]
is given by $\beta(E)=-\log\left(\frac{2E-1}{2E+1} \right).$ In particular, $\beta(E) = 1/E+\mathcal O(1/E^3).$
Based on this asymptotic law, we deduce that the Gibbs state
$\gamma(\beta(E))= \tfrac{e^{-\beta(E) H_{\text{osc}}}}{ \operatorname{tr}\left(e^{-\beta(E) H_{\text{osc}}}\right)}$ has entropy 
\begin{equation*}
\begin{split}
 S(\gamma(\beta(E))) 
 &=  \log\left(\operatorname{tr}\left(e^{-\beta(E) H_{\text{osc}}}\right)\right) + \beta(E)E =\log\left(\tfrac{\sqrt{4E^2-1}}{2}  \right) -\log\left(\tfrac{2E-1}{2E+1} \right)E.
 \end{split}
 \end{equation*}
 We stress that this shows that for the special case when the Hamiltonian is the harmonic oscillator, then $S(\gamma(\beta(E)))$ behaves like $\log(E)$ as $E \rightarrow \infty.$  
 \end{ex}
Our aim in this section is to show that, in some sense, the logarithmic divergence of the entropy of the Gibbs state, as $E \rightarrow \infty$, is not a special feature of the harmonic oscillator but universal for many classes of Hamiltonians. This result allows us then to state explicitly a rate of convergence in continuity bounds on entropies and capacities. \newline
 We start with some preliminary related ideas: \newline
 Let $H$ be a self-adjoint operator with compact resolvent on $L^2(X, d\nu(x))$. The spectral function $e_H$ of $H$, is defined as \cite[(17.5.5)]{H07} for all $x,y \in X$
 \[ e_H(x,y,E)  := \sum_{\lambda_j \in \sigma(H); \lambda_j \le E} \varphi_j(x) \overline{\varphi_j(y)}\]
 where $\varphi_j$ are the eigenfunctions corresponding to the eigenvalue $\lambda_j$ of the operator $H$. The number of eigenvalues of $H$ that are at most of energy $E,$ counted with multiplicities, is then given by 
 \[ N_H(E) = \int_{X} e_H(x,x,E) \ d\mu(x) = \sum_{\lambda_j \le E; \lambda_j \in \sigma\left(H\right)} 1.\]
 The famous Weyl law \cite{I16} gives an asymptotic description of $N_H$ for certain classes of operators in the limit of high energies, and shows that this distribution is universal. In many cases, even the precise asymptotics of eigenvalues is known.
We will show that to estimate the entropy of the Gibbs state at high energies for arbitrary Hamiltonians, it suffices to estimate the ratio of the following two auxiliary functions for high energies
 \begin{equation}
 \begin{split}
 \label{eq:auxfunc}
N_{H}^{\uparrow}(E)&:= \sum_{\lambda+\lambda' \le E; \lambda,\lambda' \in \sigma\left(H\right)}\lambda^{2} \ \text{   and   } \
N_{H}^{\downarrow}(E):= \sum_{\lambda+\lambda' \le E; \lambda,\lambda' \in \sigma\left(H\right)}\lambda \lambda'.
 \end{split}
 \end{equation}
We also observe that the simple estimate $2 \lambda \lambda' \le \lambda^2+\lambda'^2$ implies that $N_{H}^{\uparrow}(E)\ge N_{H}^{\downarrow}(E)$ where Weyl's law ensures that these two functions have a universal asymptotic behaviour as $E \rightarrow \infty$ for large classes of operators. 
The next theorem shows that the high energy asymptotics for the entropy of the Gibbs state is uniquely determined by the high energy spectrum of the Hamiltonian expressed in terms of functions defined in \eqref{eq:auxfunc}. 

\begin{theo}
 \label{theo:enttheo}
Let $H$ be an unbounded self-adjoint operator satisfying the \hyperref[ass:Gibbs]{Gibbs hypothesis}. Assume that the limit $\xi:=\lim_{\lambda \rightarrow \infty} \frac{N_{H}^{\uparrow}(\lambda) }{N_{H}^{\downarrow}(\lambda)}>1$ exists, such that $\eta:=\left(\xi-1\right)^{-1}$ is well-defined. 
Let the inverse temperature $\beta(E)$ be given as the solution of \eqref{eq:beta}. For high energies, the inverse temperature satisfies the asymptotic law
  \begin{equation}
  \begin{split}
  \label{eq:betaasym}
  \beta(E) = \frac{\eta}{E} (1+o(1))\text{ as }E \rightarrow \infty.  
   \end{split}
   \end{equation}
 In the same high energy limit the partition function satisfies 
  \begin{equation}
  \begin{split}
  \label{eq:partfunc}
   Z_{H}(\beta(E)):=\operatorname{tr}\left(e^{-\beta(E) H} \right) = \kappa E^{\eta}(1+o(1))\text{ as }E \rightarrow \infty
   \end{split}
   \end{equation}
   where $\kappa= \lim_{E \rightarrow \infty}\tfrac{1}{E^{\eta}} \sum_{\lambda \in \sigma(H)}e^{-\beta(E)\lambda}$ is a constant.
  Finally, the entropy of the Gibbs state satisfies 
   \begin{equation*}
  \begin{split}
S(\gamma(E))=  \eta \log\left( E  \right) (1+o(1)) \text{ as }E \rightarrow \infty.
   \end{split}
   \end{equation*}
 \end{theo}
 \begin{proof}
The derivative of the inverse temperature as a function of the \emph{inverse} energy satisfies
\begin{equation}
\label{eq:derivativebeta}
 \beta'(E^{-1}) = \frac{1}{E^{-1'}(\beta)} =\frac{-1}{ \frac{d}{d\beta} \left(\frac{Z_{H}(\beta)}{Z_{H}'(\beta)}\right)}=\frac{1}{\frac{Z_{H}(\beta)Z_{H}''(\beta)}{Z_{H}'(\beta)^2}-1},
 \end{equation}
where we used \eqref{eq:beta} in the second equality. 
We obtain then for the two-sided Laplace transform of the auxiliary function $N_{H}^{\uparrow}$ 
\begin{equation}
\begin{split}
\label{eq:Laplace1}
(\mathcal LN_{H}^{\uparrow})(\beta) &\stackrel{(1)}{=} \int_{-\infty}^{\infty} \sum_{\lambda+\lambda' \le s; \lambda,\lambda' \in \sigma\left(H\right)}\lambda^{2} e^{-\beta s} \ ds \stackrel{(2)}{=} \sum_{\lambda \in \sigma(H)} \sum_{\lambda' \in \sigma(H)} \lambda^2 \int_{\lambda+\lambda'}^{\infty} e^{-\beta s} \ ds \\
&\stackrel{(3)}{=} \frac{1}{\beta}  \sum_{\lambda \in \sigma(H)} \sum_{\lambda' \in \sigma(H)} \lambda^2 e^{-\beta(\lambda+\lambda')}
\end{split}
\end{equation}
where we used the definition of the two-sided Laplace transform in (1), Fubini's theorem to get (2), and by computing the integral we obtained (3). 
By an analogous calculation, we find that for $G(\beta):=\sum_{\lambda,\lambda' \in \sigma\left(H\right)}\lambda \lambda' e^{-\beta(\lambda+\lambda')}, $
\begin{equation}
\label{eq:Laplace2}
(\mathcal LN_{H}^{\downarrow})(\beta) =\frac{G(\beta)}{\beta}.
\end{equation}
The quotient of \eqref{eq:Laplace1} and \eqref{eq:Laplace2} allows us to recover the factor appearing in \eqref{eq:derivativebeta}
\begin{equation}
\begin{split}
\label{eq:identity}
\frac{Z_{H}(\beta)Z_{H}''(\beta)}{Z_{H}'(\beta)^2} =\frac{ \mathcal L(N_{H}^{\uparrow})(\beta)}{\mathcal L(N_{H}^{\downarrow})(\beta)}.
\end{split}
\end{equation}
From the existence of the limit $\xi = \lim_{\lambda \rightarrow \infty} \frac{N_{H}^{\uparrow}(\lambda) }{N_{H}^{\downarrow}(\lambda)}$ in the assumption of the theorem, we conclude that for any $\varepsilon>0$ there is $\lambda_0>0$ large enough such that for all $\lambda \in \mathbb R$
\begin{equation*}
(\xi-\varepsilon)\indic_{[\lambda_0,\infty)}(\lambda)N_{H}^{\downarrow}(\lambda)\le  \indic_{[\lambda_0,\infty)}(\lambda)N_{H}^{\uparrow}(\lambda) \le (\xi+\varepsilon)\indic_{[\lambda_0,\infty)}(\lambda)N_{H}^{\downarrow}(\lambda).
\end{equation*}
Hence, by applying the two-sided Laplace transform to this inequality we infer that for all $\beta >0$ by decomposing $\indic_{[\lambda_0,\infty)} = 1-\indic_{(-\infty,\lambda_0)},$ 
\begin{equation*}
  \begin{split}
(\xi-\varepsilon) \mathcal L((1-\indic_{(-\infty,\lambda_0)})N_{H}^{\downarrow})(\beta) &\le \mathcal L((1-\indic_{(-\infty,\lambda_0)})N_{H}^{\uparrow})(\beta) \\
&\le (\xi+\varepsilon) \mathcal L((1-\indic_{(-\infty,\lambda_0)})N_{H}^{\downarrow})(\beta).
  \end{split}
\end{equation*}
By adding $\mathcal L(\indic_{(-\infty,\lambda_0)}N_{H}^{\uparrow})(\beta)$ to the inequality and dividing by $ \mathcal L(N_{H}^{\downarrow})(\beta)$ we conclude from \eqref{eq:Laplace2} that
\begin{equation*}
  \begin{split}
&(\xi-\varepsilon)\left(1-\frac{\beta \mathcal L(\indic_{(-\infty,\lambda_0)}N_{H}^{\downarrow})(\beta)}{G(\beta)} \right)+\frac{\beta \mathcal L(\indic_{(-\infty,\lambda_0)}N_{H}^{\uparrow})(\beta)}{G(\beta) } \le \frac{\mathcal L(N_{H}^{\uparrow})(\beta)}{\mathcal L(N_{H}^{\downarrow})(\beta)}\\
&\le(\xi+\varepsilon)\left(1-\frac{\beta \mathcal L(\indic_{(-\infty,\lambda_0)}N_{H}^{\downarrow})(\beta)}{G(\beta) } \right)+\frac{\beta \mathcal L(\indic_{(-\infty,\lambda_0)}N_{H}^{\uparrow})(\beta)}{G(\beta)}. 
  \end{split}
\end{equation*}
Thus, since $\varepsilon>0$ is arbitrary we obtain as $\beta \rightarrow 0^+$ from the previous inequality, since by the Gibbs hypothesis $\liminf_{\beta \downarrow 0} G(\beta)>0,$
\[\lim_{\beta \rightarrow 0^+} \frac{\mathcal L(N_{H}^{\uparrow})(\beta)}{\mathcal L(N_{H}^{\downarrow})(\beta)}=\xi. \]
Hence, for high temperatures, \emph{i.e.}\ high energies by Remark \ref{rem:highe}, we get by \eqref{eq:derivativebeta} and \eqref{eq:identity}:
\begin{equation*}
\begin{split}
\lim_{E \rightarrow \infty} \beta'(E^{-1}) = \eta.
\end{split}
\end{equation*}
By differentiating the partition function with respect to $E$ and using \eqref{eq:beta}, we find that the partition function satisfies the differential equation
\[ \frac{d Z_{H}(\beta(E))}{dE} = -E \beta'(E) Z_{H}(\beta(E))\]
and since $\beta'(E) = -\frac{\eta}{E^2}(1+o(1)),$ we find that the partition function satisfies for some $\kappa>0$
\[ Z_{H}(\beta(E)) = \kappa E^{\eta}(1+o(1)), \text{ as } E \rightarrow \infty,\]
where
\[ \kappa = \lim_{E \rightarrow \infty}\tfrac{1}{E^{\eta}} \sum_{\lambda \in \sigma(H)}e^{-\beta(E)\lambda}.\]
Thus, by using \eqref{eq:entgibbs}, this implies that
\begin{equation*}
S(\gamma(E))=  \log\left(\kappa E^{\eta}(1+o(1))\right) + \eta(1+o(1)) = \eta \log(E)(1+o(1)), \  \text{ as } E \rightarrow \infty. 
\end{equation*}
\end{proof}

\begin{ex}
\label{standard-examples}
The entropy of the Gibbs state for the quantum harmonic oscillator as in Example \ref{harmosc} satisfies 
\[S(\gamma(E)) = \log(E)(1+o(1)), \ \text{ as } E \rightarrow \infty.\]
The entropy of the Gibbs state for regular Sturm-Liouville operators defined through
\[ (Hy)(x) = - \tfrac{1}{r(x)}(py')'(x)+\tfrac{q(x)}{r(x)}y(x),\]
on bounded intervals $(a,b)$ with $r,q \in C[a,b], \ p \in C^1[a,b],$ and $p(x),r(x)>0$ for $x \in [a,b]$ satisfies
\[S(\gamma(E)) = \tfrac{1}{2}\log(E)(1+o(1)) \ \text{ as } E \rightarrow \infty.\]
The entropy of the Gibbs state for multi-dimensional second order differential operators \cite[Sec. 17.5]{H07}
\[ H = - \sum_{j,k=1}^n \frac{\partial}{\partial x_j }\left( g^{jk}  \frac{\partial}{\partial x_k } \right) + \sum_{j=1}^n b^j \frac{\partial}{\partial x_j}+c \]
on bounded open subsets $\Omega$ of $\mathbb R^n$ with smooth boundary, Dirichlet boundary condition, and positive semi-definite matrix $(g^{jk})$ on $\overline{\Omega}$ such that $H$ is self-adjoint on $L^2(X,�d\mu(x))$ satisfies
\[S(\gamma(E)) = \tfrac{n}{2}\log(E)(1+o(1)), \ \text{ as } E \rightarrow \infty.\]
\end{ex}
\begin{proof}[Calculation:]
Instead of just referring to Example \ref{harmosc} for the harmonic oscillator, we apply Theorem \ref{theo:enttheo}: \\
\textsc{Harmonic oscillator:}  By applying the Cauchy product formula, we find from the Harmonic oscillator spectrum $\left\{ n+1/2; n \in \mathbb N_0 \right\}$
\begin{equation*}
\begin{split}
N_{H}^{\uparrow}(n+1/2) &= \sum_{k=0}^n \sum_{i=0}^k (i+\tfrac{1}{2})^2 = \frac{(2n^2+6n+3)(2+n)(1+n)}{24} \text{ and }\\
N_{H}^{\downarrow}(n+1/2) &= \sum_{k=0}^n \sum_{i=0}^k (i+\tfrac{1}{2})(k-i+\tfrac{1}{2}) = \frac{(n^2+3n+3)(2+n)(1+n)}{24}
\end{split}
\end{equation*}
such that $\eta= \frac{1}{\lim_{\lambda \rightarrow \infty} \frac{N_{H}^{\uparrow}(\lambda) }{N_{H}^{\downarrow}(\lambda)}-1} =1.$ \\
By Theorem \ref{theo:enttheo} it follows that 
\begin{equation}
S(\gamma(E)) = \log(E)(1+o(1)), \ \text{ as } E \rightarrow \infty.
\end{equation}
\textsc{Sturm-Liouville operator:} 
The spectrum of the Sturm-Liouville operators obeys high energy asymptotics \cite[Theorem $5.25$]{T12}
\[\sigma(H)=\left\{ n^2\pi^2 \left(\int_a^b \sqrt{\tfrac{r(t)}{p(t)}} \ dt \right)^{-2}+ \mathcal O(n) ; n \in \mathbb N \right\}\] such that for $\gamma:=\int_a^b \sqrt{\tfrac{r(t)}{p(t)}} \ dt$
\begin{equation*}
\begin{split}
N_{H}^{\uparrow}\left(n^2\pi^2 \gamma^{-2}\right) &=\gamma^{-4} \int_0^{\pi/2} \int_0^{n}\left(\pi^2 r^2 \cos^2(\varphi) \right)^2 r \ dr \ d \varphi \ (1+o(1))\\
&=\frac{ \pi^5 n^6}{32\gamma^{2}}(1+o(1))  \text{ and }\\
N_{H}^{\downarrow}\left(n^2\pi^2\gamma^{-2}\right) &=\gamma^{-4}\int_0^{\pi/2} \int_0^{n} \pi^4 r^4 \left(\cos^2(\varphi) \sin^2(\varphi) \right) r \ dr \ d \varphi \ (1+o(1)) \\
&= \frac{ \pi^5 n^6}{96\gamma^{2}}(1+o(1))
\end{split}
\end{equation*}
from which we obtain that $\eta= \frac{1}{\lim_{\lambda \rightarrow \infty} \frac{N_{H}^{\uparrow}(\lambda) }{N_{H}^{\downarrow}(\lambda)}-1} =\frac{1}{2}$ and thus by Theorem \ref{theo:enttheo}
\begin{equation}
S(\gamma(E)) = \tfrac{1}{2}\log(E)(1+o(1)), \ \text{ as } E \rightarrow \infty.
\end{equation}
\textsc{Multi-dimensional operators:}
The $m$-th eigenvalue of the second order operator are known to satisfy \cite[Sec. 17.5]{H07} $\lambda_m \approx \frac{4\pi^2 }{(C_n \left\lvert \Omega \right\rvert)^{2/n}}m^{2/n}$ as $m\rightarrow \infty$ where $C_n:=\tfrac{\pi^{n/2}}{\Gamma(\tfrac{n}{2}+1)}.$
For our calculation, we may drop the prefactor of the eigenvalues when taking the quotient of $N_{H}^{\uparrow}(\lambda)$ and $N_{H}^{\downarrow}(\lambda)$. Approximating the series by integrals yields
\begin{equation*}
\begin{split}
E^{-1}(\beta) &= \left(\int_0^{\infty} m^{2/n} e^{-\beta m^{2/n}} \ dm \right)^{-1} \left(\int_0^{\infty} e^{-\beta m^{2/n}} \ dm \right)(1+o(1)) \\ \newline
& =\tfrac{2\beta}{n}(1+o(1)), \text{ as } \beta \downarrow 0
\end{split}
\end{equation*}
from which we conclude by Theorem \ref{theo:enttheo}
\begin{equation}
S(\gamma(E)) = \frac{n}{2} \log(E)(1+o(1)), \ \text{ as } E \rightarrow \infty.
\end{equation}
\end{proof}
\begin{figure}[h!]
\centering
\begin{subfigure}{.5\textwidth}
  \includegraphics[width=7.3cm]{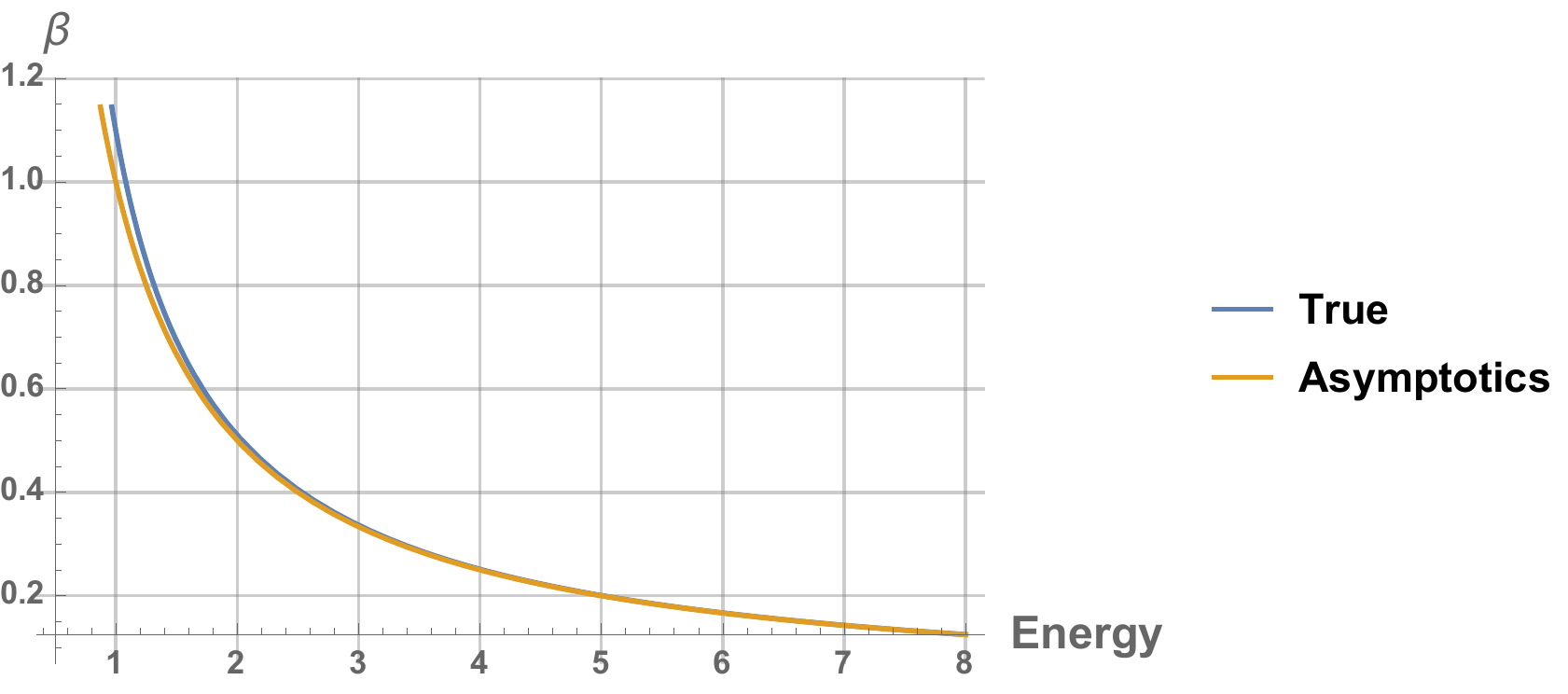}
  \subcaption{Inverse temperature of Gibbs state \\ for the quantum harmonic oscillator.}
  \label{fig:sub3}
\end{subfigure}%
\begin{subfigure}{.5\textwidth}
  \includegraphics[width=7.6cm]{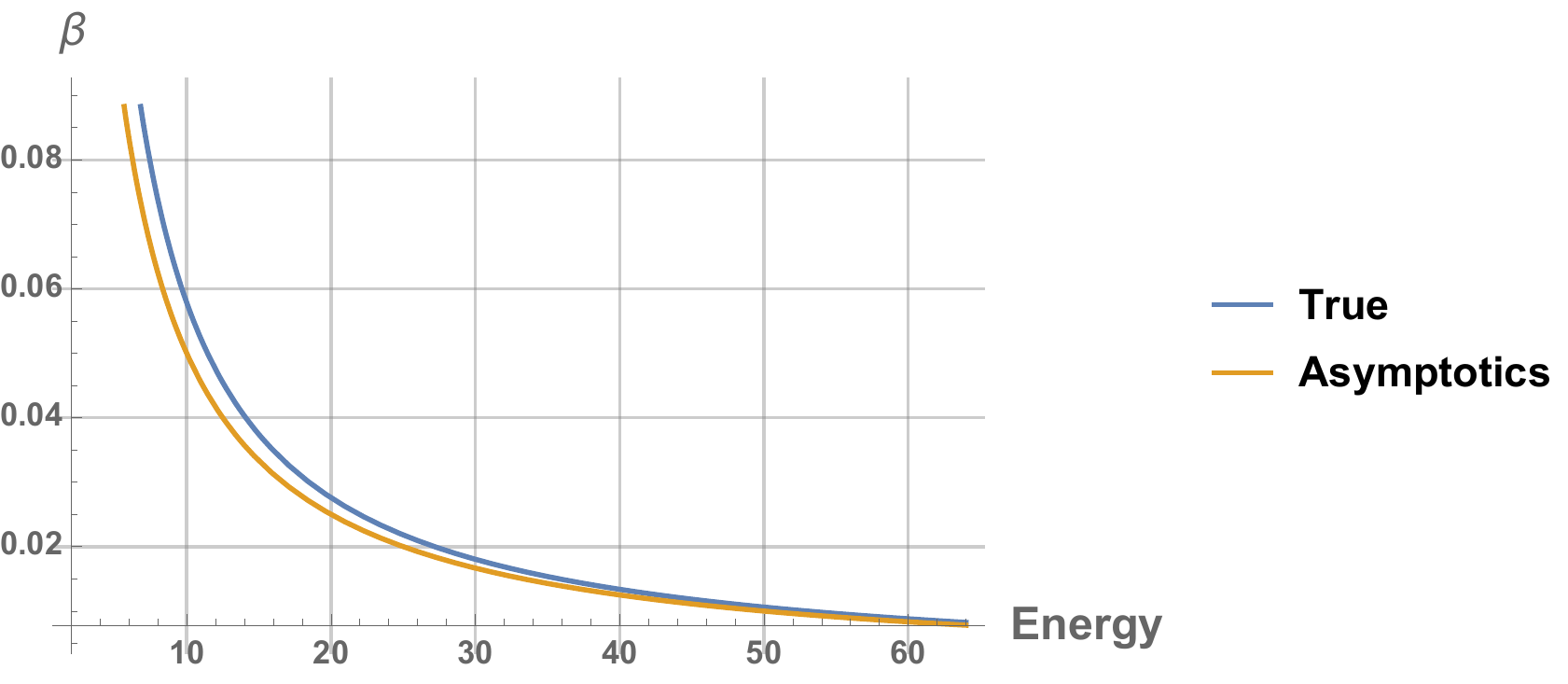}
     \subcaption{Inverse temperature of Gibbs state \\ for the operator $-\frac{d^2}{dx^2}$ on $[0,1/\sqrt{8}]$.}
  \label{fig:sub4}
\end{subfigure}
\caption{Asymptotics of inverse temperature \eqref{eq:betaasym} of the Gibbs state compared with the true solution $\beta(E)$.}
\label{fig:inversetemp}
\end{figure}
In Figure \hyperref[fig:sub3]{4(A)} we compare the true inverse temperature $\beta(E)$ of the Gibbs state for the quantum harmonic oscillator as in Example \ref{harmosc} with the asymptotic law $\beta(E)\approx \frac{1}{E}$ obtained from Theorem \ref{theo:enttheo}. In Figure \hyperref[fig:sub4]{4(B)} we compare the inverse temperature of the Gibbs state for the Hamiltonian describing a particle in a box of length $\frac{1}{\sqrt{8}}$ with the asymptotic law $\beta(E) \approx \frac{1}{2E}$ we obtained in Example \ref{standard-examples}. 
The following Proposition, which relies on Theorem \ref{theo:enttheo}, shows that for large generic classes of Schr\"odinger operators with compact resolvent, the entropy of the Gibbs states obeys a universal high energy asymptotic behaviour.

The \hyperref[corr:entropyesm]{Proposition [Entropy convergence]} then follows as an application of Theorem \ref{theo:enttheo}, which provides an explicit rate of convergence for entropies on infinite-dimensional Hilbert spaces:
\begin{proof}[Proof: Proposition [Entropy convergence]]
Under the assumptions stated in the Corollary, Lemmas $15$ and $16$ in \cite{W17} show that the von Neumann entropy satisfies
\begin{equation*}
\begin{split}
\left\lvert S(\rho)-S(\sigma) \right\rvert &\le 2 \varepsilon S(\gamma(E/\varepsilon)) +h(\varepsilon) \text{ and} \\
\left\lvert S(\rho)-S(\sigma) \right\rvert &\le ( \varepsilon'+2\delta) S(\gamma(E/\delta)) +h(\varepsilon')+h(\delta).
\end{split}
\end{equation*}
The conditional entropy satisfies by \cite[Lemma $17$]{W17} 
\begin{equation*}
\begin{split}
 \left\lvert H(A \vert B)_{\rho}-H(A \vert B)_{\sigma} \right\rvert \le &2(\varepsilon'+4 \delta)S(\gamma(E/\delta)) +(1+\varepsilon') h(\tfrac{\varepsilon'}{1+\varepsilon'})+2h(\delta).
 \end{split}
\end{equation*}
Combining these results with Theorem \ref{theo:enttheo} yields the claim of the Proposition. 
\end{proof}

Another correlation measure for a bipartite state $\rho_{AB} \in \mathscr D(\mathcal H_A \otimes \mathcal H_B)$ is the quantum mutual information (QMI) 
\[ I(A;B)_{\rho} = D(\rho_{AB} \vert \vert \rho_{A} \otimes \rho_{B}) \ge 0,\]
and is defined in terms of the relative entropy \eqref{eq:relent}. 
Let $\Phi:\mathcal T_1(\mathcal H_A) \rightarrow \mathcal T_1(\mathcal H_B)$ be a quantum channel and $H_{A}$ a positive semi-definite operator on $\mathcal H_A$. The entanglement-assisted capacity $C_{\text{ea}}$ satisfies then under the energy-constraint
\begin{equation}\label{eq:eac} 
C_{\text{ea}}(\Phi, H_{A},E) = \sup_{\operatorname{tr}(H_{A}\rho)\le E} I(B;C)_{(\Phi \otimes I_{\mathcal H_C})(\widehat{\rho})},
\end{equation}
where $\widehat{\rho}$ is a pure state in $\mathscr D(\mathcal H_A \otimes \mathcal H_C)$ with reduced state $\rho \in \mathscr D(\mathcal H_A).$ 

The two Corollaries \ref{corr:QMI} and \ref{corr:EAC} of Theorem \ref{theo:enttheo} that are stated in Appendix C provide convergence rates on QMI and hence on $C_{\text{ea}}$. 

 We continue our discussion of attenuator and amplifier channels, that were defined in Example \ref{attchannel} by studying their convergence of entropies.
\begin{ex}[Entropy bounds for attenuator and amplifier channels]
\label{ex:achannel}
We start by discussing how the expected energy of output states of these channels with time-dependent attenuation and amplification parameters behave as a function of time.

Let $\rho^{\text{att}}$ and $\rho^{\text{amp}}$ be the time-evolved states under the attenuator and amplifier channels, \emph{i.e.}\ $\rho^{\text{att}}(t)= \Lambda^{\text{att}}_t(\rho_0^{\text{att}})$ and $\rho^{\text{amp}}(t)= \Lambda^{\text{amp}}_t(\rho_0^{\text{amp}})$, with $\rho_0^{\text{att}}$ and $\rho_0^{\text{amp}}$ denoting arbitrary initial states.
Differentiating the expectation value $\operatorname{tr}( N \rho^{\text{att}}(t) )$ with respect to time shows that, for the attenuator channel, the expectation value $\operatorname{tr}( N \rho^{\text{att}}(t) )$ is a decreasing function of time
\begin{equation*}
\begin{split}
\tfrac{d}{dt}\operatorname{tr}( N \rho^{\text{att}}(t) ) &= \operatorname{tr}( N \mathcal L^{\text{att}} \rho^{\text{att}}(t) )= - \operatorname{tr}( N^2  \rho^{\text{att}}(t) ) +  \operatorname{tr}( a^*Na  \rho^{\text{att}}(t) ) \\
&= - \operatorname{tr}( N^2  \rho^{\text{att}}(t) ) +  \operatorname{tr}( N(N-1)  \rho^{\text{att}}(t) ) = - \operatorname{tr}(N \rho^{\text{att}}(t) ),
\end{split}
\end{equation*}
whereas for the amplifier channel, a similar computation shows that
\begin{equation*}
\begin{split}
\tfrac{d}{dt}\operatorname{tr}( M \rho^{\text{amp}}(t) ) &= \operatorname{tr}(M \rho^{\text{amp}}(t)).
\end{split}
\end{equation*}
Hence, it follows that $\operatorname{tr}( N \rho^{\text{att}}(t) ) = \operatorname{tr}( N \rho^{\text{att}}_0 )e^{-t}$ and $\operatorname{tr}( M \rho^{\text{amp}}(t) ) = \operatorname{tr}( M \rho^{\text{amp}}_0 )e^{t}.$
Let $\varepsilon>0$ and $t_0$ be sufficiently small such that $t_0 \le \frac{1}{E} \left(\frac{2\varepsilon}{ \zeta_{1/2}(1-\alpha)^{(1-\alpha)/2} \alpha^{\alpha/2}}\right)^{1/\alpha}.$ Then by \eqref{eq:attbound} specialising this bound for $\alpha=1/2$, shows that  $\left\lVert {\Lambda}^{\text{att}}_{t+s}-{\Lambda}^{\text{att}}_s \right\rVert^{N,E}_{\diamond^{1}} \le  2\varepsilon$ and $\left\lVert {\Lambda}^{\text{amp}}_{t+s}-{\Lambda}^{\text{amp}}_s \right\rVert^{M,E}_{\diamond^{1}} \le  2\varepsilon.$ Thus, by \hyperref[corr:entropyesm]{Proposition [Entropy convergence]}, for times $t \in (0,t_0)$ and $s>0$ such that
\begin{equation*}
\operatorname{tr}(\rho_0^{\text{att}} N) \le E e^{s} \quad {\hbox{and}} \quad
\operatorname{tr}(\rho_0^{\text{amp}} M) \le Ee^{-(t_0+s)} ,
\end{equation*}
%
we find in terms of the binary entropy $h$ 
\begin{equation*}
\begin{split}
&\left\lvert S\left(\rho^{\text{att}}(t+s)\right)-S\left(\rho^{\text{att}}(s)\right) \right\rvert \le 2 \varepsilon  \log\left( E/\varepsilon \right)  (1+o(1)) + h(\varepsilon) \text{ and } \\
&\left\lvert S\left(\rho^{\text{amp}}(t+s)\right)-S\left(\rho^{\text{amp}}(s)\right)
\right\rvert \le 2 \varepsilon \log\left( E/\varepsilon \right) (1+o(1)) + h(\varepsilon).
\end{split}
\end{equation*}
%
 \end{ex}

\subsection{Capacity bounds}
\label{capbounds}
Another application of the high energy asymptotics of the entropy of the Gibbs state are bounds on capacities of quantum channels. Concerning these bounds, we need to introduce, before stating our result, the definition of an ensemble, its barycenter, and the Holevo quantity \cite{S18}.
\begin{defi}
A Borel probability measure $\mu$ on the set of states $\mathscr D(\mathcal H) \subseteq \mathcal T_1(\mathcal H)$ is called an ensemble of quantum states. The expectation value $\overline{\rho} \in \mathscr D(\mathcal H)$
\[ \overline{\rho} = \int_{\mathscr{D}(\mathcal H)} \rho \ d\mu(\rho) \]
is called its barycenter. The expected energy of the barycenter state is defined as $E(\mu) = \operatorname{tr}(H\overline{\rho}).$
The Holevo quantity of the ensemble is defined, if $S(\overline{\rho})< \infty$,  as 
\begin{equation}\label{chi}
\chi(\mu) = S(\overline{\rho}) -  \int_{\mathscr{D}(\mathcal H)} S(\rho) \ d\mu(\rho). 
\end{equation}
For a quantum channel $\Phi: \mathcal T_1({\mathcal H}_A) \rightarrow \mathcal T_1({\mathcal H}_B)$, the pushforward ensemble $(\Phi_{*}(\mu))(B)= \mu(\Phi^{-1}(B))$ is defined as the pushforward measure for all Borel sets $B$ and is itself an ensemble on the final space of $\Phi.$
\end{defi}
\begin{rem}
If the ensemble is of the form $\mu = \sum_{i=1}^{\infty} p_i \delta_{\rho_i}$ for probabilities $p_i \ge 0 $ summing up to one $\sum_{i=1}^{\infty} p_i=1$ and delta distributions associated with states $\rho_i \in \mathscr D(\mathcal H)$ then the ensemble is also called discrete. In this case the barycenter state is just 
\[ \overline{\rho}= \sum_{i=1}^{\infty} p_i \rho_i \in \mathscr D(\mathcal H).\]
Let $\Phi$ be a quantum channel, then the pushforward ensemble of such a discrete ensemble becomes just $\Phi_{*}(\mu) = \sum_{i=1}^{\infty} p_i \delta_{\Phi(\rho_i)}$
\end{rem}
Discrete ensembles play a particularly important role in the study of capacities. Let $D_E$ be the set of discrete ensembles with barycenter state $\overline{\rho}$ of energy less than $E$ under a positive semi-definite Hamiltonian.
Let $\Phi$ be a channel, $H$ a positive semi-definite Hamiltonian, and $\mu$ a discrete ensemble. The constrained product-state classical capacity is known to be given by the Holevo capacity $\chi^*(\Phi):=\sup_{\mu \in D_E} \chi(\Phi_{*}(\mu))$, defined in terms of the Holevo quantity, by 
\begin{equation}\label{eq:Hol-cap}
C^{(1)}(\Phi, H,E) = \sup_{\mu \in D_E} \chi(\Phi_{*}(\mu)).
\end{equation}
The full classical capacity is given in terms of $C^{(1)}$ as follows
\begin{equation}
\label{eq:classcap}
C(\Phi,H,E) = \lim_{n \rightarrow \infty} \tfrac{1}{n} C^{(1)}\left(\Phi^{\otimes n} , H \otimes I^{\otimes n-1}+ I \otimes H \otimes I^{\otimes n-2}..+ I^{\otimes {n-1}} \otimes H, E\right).
\end{equation}
With those definitions at hand, we can finish the proof of \hyperref[corr:Capact]{Proposition [Capacity convergence]}.
\begin{proof}[Proof: [Convergence of capacities]]
From  \cite[Proposition $6$]{S18} it follows that
\begin{equation}
\begin{split}
\vert C^{(1)}(\Phi,H_{A}, E)- C^{(1)}(\Psi,H_{A}, E)\vert \le &\varepsilon(2t+r_{\varepsilon}(t))S(\gamma(k(E)E/(\varepsilon t)))\\
&+2g(\varepsilon r_{\varepsilon}(t)) +2h(\varepsilon t) \text{ and} \\
\vert C(\Phi,H_{A}, E)- C(\Psi,H_{A}, E)\vert \le &2\varepsilon(2t+r_{\varepsilon}(t))S(\gamma(k(E)E/(\varepsilon t)))\\
&+2g(\varepsilon r_{\varepsilon}(t)) +4h(\varepsilon t).
\end{split}
\end{equation}
and the result follows immediately from Theorem \ref{theo:enttheo}.
\end{proof}
\section{Open problems}
\label{probs}
Concerning the first part of the paper, it would be desirable to study extensions of
our work to non-autonomous systems, such as systems described by a Schr\"odinger
operator with {\em{time-dependent}} potentials. For the Schr\"odinger equation, an application
of the variation of constants formula yields a bound for such systems as well (see Proposition \ref{prop:schroe}).
This should also work, under suitable assumptions, for non-autonomous open quantum systems.
However, more mathematical care may be needed for the latter.

To answer the important questions: $(i)$ ``How fast can entropy increase?''-for {\\em{any}} infinite-dimensional open quantum
system whose dynamics is governed by a QDS, and $(ii)$ ``How fast can information be transmitted?''-through any quantum channel
(obtained by freezing the time parameter in the QDS), it is necessary to find bounds on the expected energy for the state of the 
underlying open quantum system over time (as has been done 
for the case of the attenuator and amplifier channels in Example 12\footnote{In fact, in Example 12, explicit expressions, and not just bounds, have been obtained.}.) 
To our knowledge, such bounds have not been obtained in full generality yet. 

The first step to answer these two questions was provided
by Winter~\cite{W15} and Shirokov~\cite{S18}, who derived continuity bounds on entropies and
capacities, respectively. Our paper provides, as a second step, a time-dependent bound on
the evolution of the expected energy of the state of the open quantum system, 
which enters these continuity bounds through the energy constraint. 
Understanding the behaviour of this expected energy as a function of time is needed in
order to infer, from the continuity bounds, how fast entropies and capacities can change.

It would be furthermore desirable to extend Theorem \ref{theo:enttheo} to higher-order terms. In Figure
\hyperref[fig:sub3]{4(A)} we see that the leading-order approximation for the inverse temperature provided
by Theorem \ref{theo:enttheo} is almost indistinguishable from the true solution for the harmonic
oscillator whereas the leading-order approximation in Figure \hyperref[fig:sub3]{4(B)} for the particle in a box
seems to converge somewhat slower than the true solution. A
better understanding of higher order terms should be able to capture these behaviours more
precisely.
\begin{appendix}

\section{Properties of ECD norms}
The following Proposition states necessary technical conditions on the eigenbasis of states satisfying an energy-constrained condition.
\begin{prop}
\label{prop1}
Let $S$ be positive semi-definite. Let $\rho = \sum_{i=1}^{\infty} \lambda_i \vert \varphi_i \rangle \langle \varphi_i \vert$ be a state. Then $\operatorname{tr}(S\rho)=\infty$ if there is $\ket{\varphi_i} \notin D(\sqrt{S})$ with $\lambda_i \neq 0.$ Analogously, a state $\rho$ satisfies $\operatorname{tr}(S\rho S)=\infty$ if there is $\varphi_i \notin D(S)$ with $\lambda_i \neq 0.$
The converse implications hold if $\rho$ is of finite-rank.
\end{prop}
\begin{proof}
Let $S$ be a positive semi-definite operator. The spectral theorem implies that $\ket{\varphi} \in D(\sqrt{S})$ if and only if $\operatorname{tr}(S\vert \varphi \rangle \langle \varphi \vert)<\infty$
\begin{equation*}
\begin{split}
\operatorname{tr}(S\vert \varphi \rangle \langle \varphi \vert)&=\sup_{n \in \mathbb N}\operatorname{tr}(S\mathcal E^S_{[0,n]}\vert \varphi \rangle \langle \varphi \vert)=\sup_{n \in \mathbb N}\left\langle S\mathcal E^S_{[0,n]} \varphi, \varphi  \right\rangle \\
&=\sup_{n \in \mathbb N} \int_0^n \lambda \ d\langle \mathcal E^S_{\lambda}\varphi,\varphi \rangle =\int_0^{\infty} \lambda \ d\langle \mathcal E^S_{\lambda}\varphi,\varphi \rangle.\end{split}
\end{equation*}

Hence, let $\rho \in \mathscr D(\mathcal{H})$ be a state with spectral decomposition $\rho = \sum_{i=1}^{\infty} \lambda_i \ \vert \varphi_i \rangle \langle \varphi_i \vert$ such that there exists $\varphi_i \notin D(\sqrt{S})$ and $\lambda_i \neq 0$. Then $\operatorname{tr}(S\rho)=\infty.$ This follows immediately from
\[\operatorname{tr}(S\rho)= \sup_{n \in \mathbb N}\operatorname{tr}(S\mathcal E^S_{[0,n]}\rho)= \sup_{n \in \mathbb N}\sum_{i \in \mathbb N} \lambda_i \operatorname{tr}(S\mathcal E^S_{[0,n]}\vert \varphi_i \rangle \langle \varphi_i \vert) =\sum_{i \in \mathbb N} \lambda_i \operatorname{tr}(S\vert \varphi_i \rangle \langle \varphi_i \vert).\]

For the operator domain, it follows that $\ket{\varphi} \in D(S)$ if and only if $\operatorname{tr}(S \vert \varphi \rangle \langle \varphi \vert S)<\infty$ as we can deduce from
\begin{equation*}
\begin{split}
\operatorname{tr}(S\vert \varphi \rangle \langle \varphi \vert S)&=\sup_{n \in \mathbb N}\left\lVert S\mathcal E^S_{[0,n]} \varphi \right\rVert^2=\int_0^{\infty} \lambda^2 \ d\langle \mathcal  E_{\lambda}^S\varphi,\varphi \rangle.
\end{split}
\end{equation*}
Just like for the form domain, this implies for a state with eigendecomposition $\rho = \sum_{i}�\lambda_i \vert \varphi_i \rangle \langle \varphi_i \vert$ it follows that $\operatorname{tr}(S \rho S)=\infty$ if there is $\ket{\varphi_i} \notin  D(S)$ such that $\lambda_i \neq 0.$
\end{proof}

\section{Dynamics of QDS in ECD norms}
The following Proposition is an adaptation of the uniform boundedness principle to the $\alpha$-ECD norm and can be applied as a perturbation theorem for convergence in $\alpha$-ECD norms.
\begin{prop}
\label{theorem1}
Let $S$ be a positive semi-definite operator, $\alpha \in (0,1]$, and $E> \operatorname{inf}(\sigma(S)).$ We then define the closed set
\[\mathcal A^E:=\left\{ \rho \in \mathscr D(\mathcal{H} \otimes \mathcal{H}'); \operatorname{tr}(S^{\alpha} \rho_{\mathcal{H}}S^{\alpha} )\le E^{2\alpha}\right\}.\]
Let $H$ be a self-adjoint operator such that for all $\rho \in \mathcal A^E$ 
\[\operatorname{tr}\left((\vert H \vert^{\alpha} \otimes_{\pi} I_{\mathcal T_1(\mathcal{H}')}) \rho (\vert H \vert^{\alpha}\otimes_{\pi}  I_{\mathcal T_1(\mathcal{H}')} ) \right) = \operatorname{tr}\left(\vert H \vert^{\alpha} \rho_\mathcal{H} \vert H \vert^{\alpha} \right)< \infty.\]
Then the $H$-associated strongly continuous one-parameter group $T^{\operatorname{vN}}_t\rho= e^{-itH}\rho e^{itH}$ is $\alpha$-H\"older continuous with respect to the $\alpha$-ECD norm generated by $S$ and satisfies 
\[\left\lVert T^{\operatorname{vN}}_t-T^{\operatorname{vN}}_s \right\rVert_{\diamond^{2\alpha}}^{S,E} \le 2 g_{\alpha} \left\lVert \vert H \vert^{\alpha}\right\rVert_{\diamond^{2\alpha}}^{S,E} \vert t-s \vert^{\alpha} \text{ with }\left\lVert \vert H \vert^{\alpha} \right\rVert_{\diamond^{2\alpha}}^{S,E}< \infty. \] 

\end{prop}
\begin{proof}
We can bound by the Cauchy-Schwarz inequality 
\[\sup_{n \in \mathbb N} \left\lVert \left(\left\lvert H\right\rvert^{\alpha} \mathcal E^{\vert H \vert}_{[0,n]}   \otimes_{\pi} I_{\mathcal T_1(\mathcal{H}')}\right) \rho \right\rVert_{1} \le   \sqrt{\operatorname{tr}\left(\vert H \vert^{\alpha} \rho_\mathcal{H} \vert H \vert^{\alpha} \right)}< \infty.\]
This allows us to define a family of closed sets
\[ A^E_m:=\left\{ \rho \in \mathcal A^E: \sup_{n}  \left\lVert \left(\left\lvert H  \right\rvert^{\alpha} \mathcal E^{\vert H \vert}_{[0,n]} \otimes_{\pi} I_{\mathcal T_1(\mathcal{H}')}\right) \rho \right\rVert_{1} \le m \right\}\]
that exhaust  $ \mathcal A^E = \bigcup_{m \in \mathbb N} A^E_m$ by assumption. The set $\mathcal A^E$ is closed in $\mathcal T_1(\mathcal{H} \otimes \mathcal{H}') = \mathcal T_1(\mathcal H) \otimes_{\pi} \mathcal T_1(\mathcal H')$ and thus complete. 
\hyperref[Baire]{Baire's} theorem implies that one of the sets $A^E_m$ has non-empty interior, \emph{i.e.}\ there is $\rho_0 \in A^E_m$ and $\varepsilon>0$ such that the closed ball (in trace distance) $\overline{B(\rho_0,2\varepsilon)}$ is contained in $A^E_m.$

Thus, let $\rho \in \mathcal A^E$ be arbitrary, then the auxiliary density matrix $\rho_{\text{aux}}:=(1-\varepsilon)\rho_0 + \varepsilon \rho $ is an element of $ A^E_m$ as well. Moreover, $\left\lVert \rho_{\text{aux}}-\rho_0 \right\rVert_{1} \le 2 \varepsilon.$ Thus, $\rho_{\text{aux}}$ is an element of $ A^E_m.$
By the definition of $A^E_m$ we therefore obtain, since $\rho$ was an arbitrary element of $\mathcal A^E$, immediately that 
$\left\lVert \vert H \vert^{\alpha}\right\rVert_{\diamond^{2\alpha}}^{S,E} $ must be finite.
We then obtain
\begin{equation*}
\begin{split}
\frac{1}{t^{\alpha}}\left\lVert (T^{\operatorname{vN}}_t \otimes_{\pi} I_{\mathcal T_1(\mathcal{H}')} -I) (\rho) \right\rVert_{1} &\le \frac{2}{t^{\alpha}} \left\lVert (T^{\text{S}}_t \otimes_{\pi} I_{\mathcal T_1(\mathcal{H}')} -I)(\rho)  \right\rVert_{1}  \\
& \le  2 \zeta_{\alpha}(1-\alpha)^{\tfrac{1-\alpha}{2}} \alpha^{\tfrac{\alpha}{2}}\left\lVert (\vert H \vert^{\alpha} \otimes_{\pi} I_{\mathcal T_1(\mathcal{H}')})(\rho )\right\rVert_{1} \\
& \le 2 \zeta_{\alpha}(1-\alpha)^{\tfrac{1-\alpha}{2}} \alpha^{\tfrac{\alpha}{2}}\left\lVert \vert H \vert^{\alpha} \right\rVert_{\diamond^{2\alpha}}^{S,E}
\end{split}
\end{equation*}
where we used the triangle inequality to get the first estimate, Lemma \ref{Favard} for the second one, and the definition of the ECD-norm for the last one.
\end{proof}
\section{Capacity bounds}
In the following let $h(x):=-x\log(x)-(1-x)\log(1-x)$ be the binary entropy, $g(x):=(x+1)\log(x+1)-x\log(x),$ and  $r_{\varepsilon}(t) = \frac{1+t/2}{1-\varepsilon t}$ a function on $(0,\frac{1}{2\varepsilon}].$
\begin{corr}[QMI]
\label{corr:QMI}
Consider quantum systems $A,B,C$, quantum channels $\Phi, \ \Psi: \mathcal T_1(\mathcal H_A) \rightarrow \mathcal T_1(\mathcal H_B)$, and energies $E_1,..,E_n.$ 
Let $H_{A}$ be a positive semi-definite operator on $\mathcal H_A$ and $H_{B}$ a positive semi-definite operator on $\mathcal H_B,$ with $H_{B}$ satisfying the \hyperref[ass:Gibbs]{Gibbs hypothesis}. 
We also assume that the limit $\xi:=\lim_{\lambda \rightarrow \infty} \frac{N_{H_B}^{\uparrow}(\lambda) }{N_{H_B}^{\downarrow}(\lambda)}>1$ exists such that $\eta:=\left(\xi-1\right)^{-1}$ is well-defined.  \\
Let $\rho \in \mathscr{D}(\mathcal H^{\otimes n}_A \otimes \mathcal H_C)$ denote a state of the composite system $A_1A_2..A_nC$ such that \newline $E_{A} = \operatorname{max}_{1\le k \le n} \operatorname{tr}(H_{A} \rho_{\mathcal H_{A_k}})< \infty$ where $\mathcal H_{A_k}$ is the $k$-th factor in the tensor product $\mathcal H^{\otimes n}_A.$
If the channels are such that $\tfrac{1}{2}\left\lVert \Phi-\Psi \right\rVert_{\diamond^{1}}^{H_A,E_A} \le \varepsilon,$ and for $k=1,.,n$ both $\operatorname{tr}(H_{B}\Phi(\rho_{A_k})), \operatorname{tr}(H_{B}\Psi(\rho_{A_k}))\le E_k$ then for all $t \in (0,1/(2 \varepsilon))$ with $E=\frac{1}{n} \sum_{k=1}^n E_k,$
\begin{equation*}
\begin{split}\vert I(B^n;C)_{(\Phi^{\otimes n} \otimes I_{C})(\rho)}-I(B^n;C)_{(\Psi^{\otimes n} \otimes I_{C})(\rho)}\vert \le &2n\varepsilon(2t+r_{\varepsilon}(t))\eta \log(E/(\varepsilon t))(1+o(1))\\
&+2ng(\varepsilon r_{\varepsilon}(t)) +4nh(\varepsilon t), \text{ as }\varepsilon \downarrow 0.
\end{split}
\end{equation*}
\end{corr}
\begin{proof}
By \cite[Proposition $5$]{S18} it follows that 
\begin{equation*}
\begin{split}\vert I(B^n;C)_{(\Phi^{\otimes n} \otimes I_{C})(\rho)}-I(B^n;C)_{(\Psi^{\otimes n} \otimes I_{C})(\rho)}\vert \le &2n\varepsilon(2t+r_{\varepsilon}(t))S(\gamma(E/(\varepsilon t))\\
&+2ng(\varepsilon r_{\varepsilon}(t)) +4nh(\varepsilon t), \text{ as }\varepsilon \downarrow 0,
\end{split}
\end{equation*}
which together with Theorem \ref{theo:enttheo} gives the claim of the Corollary.
\end{proof}
\begin{corr}[EAC]
\label{corr:EAC}
Let $A,B$ be two quantum systems and $H_{A}$ be a positive semi-definite operator on $\mathcal H_A$ satisfying the \hyperref[ass:Gibbs]{Gibbs hypothesis}. We also assume that the limit $\xi:=\lim_{\lambda \rightarrow \infty} \frac{N_{H_A}^{\uparrow}(\lambda) }{N_{H_A}^{\downarrow}(\lambda)}>1$ exists such that $\eta:=\left(\xi-1\right)^{-1}$ is well-defined and take $E> \inf(\sigma(H))$.  \\
Let $\Phi, \ \Psi: \mathcal T_1(\mathcal H_A) \rightarrow \mathcal T_1(\mathcal H_B)$ be two quantum channels such that $\tfrac{1}{2} \left\lVert \Phi-\Psi \right\rVert_{\diamond^{1}}^{H_A,E} \le \varepsilon$ then for $t \in (0,\tfrac{1}{2\varepsilon}]$ the EAC satisfies
\begin{equation}
\begin{split}
\vert C_{\text{ea}}(\Phi,H_{A}, E)- C_{\text{ea}}(\Psi,H_{A}, E)\vert \le &2\varepsilon(2t+r_{\varepsilon}(t))\eta \log(E/(\varepsilon t))(1+o(1))\\
&+2g(\varepsilon r_{\varepsilon}(t)) +4h(\varepsilon t), \text{ as }\varepsilon \downarrow 0.
\end{split}
\end{equation}
\end{corr}
\begin{proof}
By \cite[Proposition $7$]{S18} it follows that in terms of the Gibbs state $\gamma(E/(\varepsilon t)$
\begin{equation}
\begin{split}
\vert C_{\text{ea}}(\Phi,H_{A}, E)- C_{\text{ea}}(\Psi,H_{A}, E)\vert \le &2\varepsilon(2t+r_{\varepsilon}(t)) S(\gamma(E/(\varepsilon t))\\
&+2g(\varepsilon r_{\varepsilon}(t)) +4h(\varepsilon t), \text{ as }\varepsilon \downarrow 0.
\end{split}
\end{equation}
Combining this result with Theorem \ref{theo:enttheo} yields the claim. 
\end{proof}
\begin{corr}[Holevo quantity]
\label{corr:Holevo}
Let $A,B$ be two quantum systems, $E>0$, and $\mu$ any ensemble of states on $\mathcal H_A$ whose barycenter has expected energy $E(\mu).$ 
Let $H_{A}$ be a positive semi-definite operator on $\mathcal H_A$ and $H_B$ a positive semi-definite operator on $\mathcal H_B$ satisfying the \hyperref[ass:Gibbs]{Gibbs hypothesis}. We also assume that the limit $\xi:=\lim_{\lambda \rightarrow \infty} \frac{N_{H_B}^{\uparrow}(\lambda) }{N_{H_B}^{\downarrow}(\lambda)}>1$ exists such that $\eta:=\left(\xi-1\right)^{-1}$ is well-defined.  \\
Let $\Phi, \ \Psi: \mathcal T_1(\mathcal H_A) \rightarrow \mathcal T_1(\mathcal H_B)$ be two quantum channels such that both 
\[\operatorname{tr}(H_{B} \Phi(\overline{\rho})),\operatorname{tr}(H_{B} \Psi(\overline{\rho})) \le E\]
and $\tfrac{1}{2} \left\lVert \Phi-\Psi \right\rVert_{\diamond^{1}}^{H_A,E(\mu)} \le \varepsilon$. Then for $t \in (0,\frac{1}{2\varepsilon}]$ the Holevo quantity satisfies
\begin{equation*}
\begin{split}\vert \chi(\Phi_{*}(\mu))-\chi(\Psi_{*}(\mu)) \vert \le &\varepsilon(2t+r_{\varepsilon}(t))\eta \log(E/(\varepsilon t))(1+o(1))\\
&+2g(\varepsilon r_{\varepsilon}(t)) +2h(\varepsilon t), \text{ as }\varepsilon \downarrow 0.
\end{split}
\end{equation*}
\end{corr}
\begin{proof}
From  \cite[Proposition $4$]{S18} it follows that
\begin{equation*}
\begin{split} \vert \chi(\Phi(\mu))-\chi(\Psi(\mu)) \vert \le&\varepsilon(2t+r_{\varepsilon}(t))S(\gamma(E/(\varepsilon t)))+2g(\varepsilon r_{\varepsilon}(t)) +2h(\varepsilon t)
\end{split}
\end{equation*}
such that the claim of the Corollary follows from Theorem \ref{theo:enttheo}.
\end{proof}

\end{appendix}

\smallsection{Acknowledgement}
Support by the EPSRC grant EP/L016516/1 for the University of Cambridge CDT,  the CCA is gratefully acknowledge (S.B.). N.D. is grateful to Pembroke College and DAMTP for support. 


\begin{thebibliography}{0}
\bibitem[A02]{A02} Alicki, R. (2002). \emph{A search for a border between classical and quantum worlds.} Phys. Rev. A 65, 034104.
\bibitem[A02a]{A02a} Alicki R. (2002). \emph{Invitation to Quantum Dynamical Semigroups.} In: Garbaczewski P., Olkiewicz R. (eds) Dynamics of Dissipation. Lecture Notes in Physics, vol 597. Springer, Berlin, Heidelberg.
\bibitem[AS04]{AS} Arnold, A. and Sparber, C. (2004). \emph{Quantum Dynamical Semigroups for Diffusion Models with Hartree Interaction}. Commun. Math. Phys. (2004) 251: 179.
\bibitem[Au07]{Au} Audenaert, K.M.R. (2007). \emph{A sharp continuity estimate for the
von Neumann entropy}, J. Math. Phys. A: Math. Theor.
40(28):8127-8136.
\bibitem[E06]{E06} Eisert, J. (2006). \emph{Discrete Quantum States versus Continuous Variables}, in Lectures on Quantum Information by Dagmar Bruss (Editor), Gerd Leuchs (Editor), Wiley. 
\bibitem[C15]{Chang} Chang, M. (2015). \emph{Quantum Stochastics}.  Cambridge University Press. 
\bibitem[CBPZ92]{CBPZ} Cirac, J., Blatt, R., Parkins, A., and Zoller, P. (1992). \emph{Preparation of Fock States by Observation of Quantum Jumps in an Ion Trap.} Physical Review Letters Vol. 70, Number 6.
\bibitem[CEPH13]{CEPH} del Campo, A., Egusquiza, I. L., Plenio, M. B., and Huelga, S. F. (2013). \emph{Quantum Speed Limits in Open System Dynamics}. Phys. Rev. Lett. 110, 050403 
\bibitem[CF98]{CF} Chebotarev, A. and Fagnola, F.(1998). \emph{Sufficient conditions for conservativity of minimal quantum dynamical semigroups}. J. Funct. Anal. 153, 382-��404.
\bibitem[CGQ03]{CGQ03} Chebotarev A., Garcia J., Quezada R. (2003). \emph{Interaction Representation Method for Markov Master Equations in Quantum Optics.} In: Rebolledo R. (eds) Stochastic Analysis and Mathematical Physics II. Trends in Mathematics. Birkh\"auser, Basel.
\bibitem[CW03]{CW03} Christandl, M. and Winter, A. (2003). \emph{Squashed entanglement? An additive entanglement measure}, J. Math. Phys. 45(3):829-840.
\bibitem[Da77]{Da} Davies, E. B. (1977). \emph{Quantum dynamical semigroups and the neutron diffusion equation}. Reports on Mathematical Physics 11.2. 169-188.
\bibitem[DL13]{DL} Deffner, S. and Lutz, E. (2013). \emph{Quantum Speed Limit for Non-Markovian Dynamics}. Phys. Rev. Lett.,111:010402.
\bibitem[DTG16]{DTG} De Palma, G., Trevisan, D., and Giovannetti, V. (2016). \emph{Passive States Optimize the Output of Bosonic Gaussian Quantum Channels.}  IEEE Transactions on Information Theory, Volume: 62, Issue: 5.
\bibitem[DC17]{DC} Deffner, S. and Campbell, S. (2017). \emph{Quantum speed limits: From Heisenberg's uncertainty principle to optimal quantum control}. Journal of Physics A Mathematical and Theoretical 50(45).
\bibitem[EN00]{EN} Engel, K-J. and Nagel, R. (2000). \emph{One-Parameter Semigroups for Linear Evolution Equations}. Springer. Graduate Texts in Mathematics.
\bibitem[EN06]{EN2} Engel, K-J. and Nagel, R. (2000).  \emph{A Short Course on Operator Semigroups}. Springer-Verlag New York
\bibitem[GLM03]{GLM}Giovannetti, V., Lloyd, S., and Maccone, L. (2003). \emph{Quantum limits to dynamical evolution.} Phys.
Rev. A, 67:052109.
\bibitem[GLM03a]{GLM2}Giovannetti, V., Lloyd, S., and Maccone, L. (2003). \emph{The quantum speed limit.} In D. Abbott, J. H.
Shapiro, and Y. Yamamoto, editors, Fluctuations and Noise in Photonics and Quantum
Optics, volume 5111 of Proc. SPIE, pages 1-6.
\bibitem[GLM04]{GLM3}Giovannetti, V., Lloyd, S., and Maccone, L. (2003). \emph{The speed limit of quantum unitary evolution.} J.
Opt. B, 6:807.
\bibitem[GS11]{GS} Gustafson, S. J. and Sigal, I.M. (2011). \emph{Mathematical Concepts of Quantum Mechanics.} Universitext. Springer.
\bibitem[F78]{Fr} Frigerio, A. (1978). \emph{Stationary states of quantum dynamical semigroups}. Communications in Mathematical Physics 63, 269-276.
\bibitem[FG99]{FG}  Fuchs, C.A. and van de Graaf, J. (1999).  \emph{Cryptographic Distinguishability Measures for Quantum Mechanical States}, IEEE Trans. Inf. Theory 45, 1216.
\bibitem[FR01]{FR} Fagnola, F. and Rebolledo, R. (2001). \emph{On the existence of stationary states for quantum dynamical semigroups}. Journal of Mathematical Physics 42.3.
\bibitem[GKS76]{GKS} Gorini, V., Kossakowski, A., and Sudarshan, E. (1976). \emph{Completely positive dynamical semigroups of N-level systems.} Journal of Mathematical Physics. 17, Nr. 5, p. 821-825.
\bibitem[H03]{H03} A.S.Holevo,(2003). \emph{Classical capacities of quantum channels with constrained
inputs}, Probability Theory and Applications. V.48. N.2. 359-374.
\bibitem[H\"o07]{H07} H\"ormander, L. (2007). \emph{The Analysis of Linear Partial Differential Operators III}. Springer. 
\bibitem[HV08]{HV} Hornberger, K. and Vacchini, B. (2008). \emph{Monitoring derivation of the quantum linear Boltzmann equation}. Phys. Rev. A 77, 022112
\bibitem[HV09]{HV2} Hornberger, K. and Vacchini, B. (2009). \emph{Quantum linear Boltzmann equation.}Physics Reports
Volume 478, Issues 4-6, Pages 71-120.
\bibitem[I72]{I} Ichinose, T. (1972). \emph{Operators on Tensor Products of Banach Spaces.} Transactions of the AMS, Volume 170.
\bibitem[I16]{I16} Ivrii, V. (2016). \emph{100 years of Weyl's law.} Bulletin of Mathematical Sciences. Volume 6, Issue 3, pp 379-452.
\bibitem[Lin76]{Lin} Lindblad, G. (1976). \emph{On the generators of quantum dynamical semigroups}. Commun. Math. Phys. (1976) 48: 119. 
\bibitem[LMS00]{LMS} Lenzi, E.K., Mendes, R.S., and da Silva, L.R. (2000). \emph{Statistical mechanics based on Renyi entropy}. Physica A 280. 337-345.
\bibitem[LS09]{LS09} Leung, D. and Smith,G. (2009). \emph{Continuity of quantum channel capacities},
Commun. Math. Phys. 292(1):201-215.
\bibitem[LT09]{LT} Levitin, L. and Toffoli, Y. (2009). \emph{Fundamental limit on the rate of quantum dynamics: The unified bound is tight.} Phys. Rev. Lett., 103:160502.
\bibitem[ML98]{ML} Margolus, N. and Levitin, L. (1998). \emph{The maximum speed of dynamical evolution.} Physica D,
120:188.
\bibitem[MT91]{MT}Mandelstam L., Tamm I. (1991). \emph{The Uncertainty Relation Between Energy and Time in Non-relativistic Quantum Mechanics.} In: Bolotovskii B.M., Frenkel V.Y., Peierls R. (eds) Selected Papers. Springer, Berlin, Heidelberg
\bibitem[N18]{N18} Nair, R. (2018). \emph{Quantum-limited loss sensing: Multiparameter estimation and Bures distance between loss channels.} \arXiv{1804.02211}.
\bibitem[NS86]{NS86} Nagel, R. and Schlotterbeck, U. (1986). \emph{One-parameter Semigroups of Positive Operators, Basic results on semigroups on Banach spaces} Lecture Notes in Mathematics, Volume 1184, Springer.
\bibitem[NN11]{NN11} Nielsen, F. and Nock, R. (2011). \emph{On R\'enyi and Tsallis entropies and divergences for exponential families}. IEEE Transactions on Information Theory 22(10):1-7.
\bibitem[P93]{P} Pfeifer, P. (1993). \emph{How fast can a quantum state change with time?} Phys. Rev. Lett., 70:3365.
\bibitem[PLOB17]{PLOB}Pirandola,S., Laurenza, R., Ottaviani,C., and Banchi, L. (2017). \emph{Fundamental Limits of Repeaterless Quantum Communications}, Nature Comm.7:15043.
\bibitem[PS70]{PS} R.T. Powers and E. St\o rmer (1970). \emph{Free states of the canonical anticommutation relations.} Commun. Math. Phys. 16, 1-33.
\bibitem[RS1]{RS1} Reed, M. and Simon, B. (1981). \emph{Methods of Modern Mathematical Physics 1: Functional Analysis}. Elsevier.
\bibitem[RS2]{RS2} Reed, M. and Simon, B. (1975). \emph{Methods of Modern Mathematical Physics 2: Fourier Analysis, Self-Adjointness}. Elsevier.
\bibitem[RSI]{RSI} Roga, W., Spehner, D., and Illuminiati, F. (2016). \emph{Geometric measures of quantum correlations: characterization, quantification, and comparison by distances and operations}. Journal of Physics A: Mathematical and Theoretical. Volume 49, Number 23.
\bibitem[SCMC18]{SCMC} Shanahan, B., Chenu, A., Margolus, N., and del Campo, A. (2018). \emph{Quantum Speed Limits across the Quantum-to-Classical Transition}. Phys. Rev. Lett. 120, 070401.
\bibitem[Si15]{S15} Simon, B. (2015). \emph{Operator Theory: A Comprehensive Course in Analysis, Part IV} American Mathematical Society. 
\bibitem[Shi06]{S06} Shirokov, M. E. (2006). \emph{Entropy characteristics of subsets
of states.} Izvestiya: Mathematics 70(6):1265-1292.
\bibitem[Shi18]{S18} Shirokov, M. E. (2018). \emph{On the 
-Constrained Diamond Norm and Its Application in Quantum Information Theory,} Problems of Information Transmission, Volume 54, Issue 1, pp 20-33.
\bibitem[SSRW15]{SSRW15} Sutter, D., Scholz, V., Renner, R., and Winter, A. (2015). \emph{Approximate
Degradable Quantum Channels}, IEEE International Symposium on Information Theory (ISIT).
\bibitem[T12]{T12} Teschl, G. (2012). \emph{Ordinary Differential Equations and Dynamical Systems},
Graduate Studies in Mathematics 140, Amer. Math. Soc., Providence, 2012.
\bibitem[T78]{T78} Triebel, H. (1978). \emph{Interpolation Theory, Function Spaces, Differential Operators}, North-Holland. 
\bibitem[T88]{T} Tsallis, C. (1988). \emph{Possible generalization of Boltzmann-Gibbs statistics,} Journal of statistical physics, vol. 52, no. 1-2, pp. 479-487.
\bibitem[UK16]{UK} Uzdin, R. and Koslo, R. (2016). \emph{Speed limits in Liouville space for open quantum systems}. EPL
(Europhysics Letters), 115:40003.
\bibitem[V04]{V04} Vacchini, B. (2004). \emph{Quantum optical versus quantum Brownian motion master equation in terms of covariance and equilibrium properties}. Journal of Mathematical Physics 2002 43:11, 5446-5458.
\bibitem[Wi18]{Wi18} Wilde, M. (2018). \emph{Strong and uniform convergence in the teleportation simulation of bosonic Gaussian channels}. Physical Review A 97, 062305.
\bibitem[W15]{W15} Winter, A. (2015). \emph{Tight Uniform Continuity Bounds for Quantum Entropies: Conditional Entropy, Relative Entropy Distance and Energy Constraints.} Communications in Mathematical Physics 347(1)
\bibitem[W17]{W17} Winter, A. (2017). \emph{Energy-constrained diamond norm with applications to the
uniform continuity of continuous variable channel capacities.} \arXiv{1712.10267}.
\bibitem[We78]{We78} Wehrl, A. (1978). \emph{General properties of entropy.} Rev. Mod. Phys. 50, 221.
\bibitem[YHW08]{YHW08} Yang, D., Horodecki, M., and Wang, Z. (2008). \emph{An Additive and
Operational Entanglement Measure: Conditional Entanglement
of Mutual Information}, Phys. Rev. Lett. 101:140501.
\end{thebibliography}
\end{document}